\documentclass[journal]{IEEEtran}
\usepackage{soul}

\newcommand{\egl}[1]{{\color{red}{[EGL: #1]}}}
\newcommand{\tktk}[1]{{\color{blue}{[TT: #1]}}}
\newcommand{\hi}[1]{{\color{teal}{[HI: #1]}}}

\usepackage{cite}
\ifCLASSINFOpdf
  \usepackage[pdftex]{graphicx}
\else
\fi
\usepackage{amsmath}
\usepackage{amssymb}
\usepackage{bm}
\usepackage{flushend}
\usepackage{mathtools}

\usepackage{algcompatible}
\usepackage{algorithm}
\usepackage{algpseudocode}
\algrenewcommand\algorithmicindent{0em}
\makeatletter
\algnewcommand{\LineComment}[1]{\Statex \hskip\ALG@thistlm #1}

\makeatother
\usepackage{setspace}
\usepackage{upgreek}
\let\Algorithm\algorithm
\renewcommand\algorithm[1][]{\Algorithm[#1]\setstretch{1}}

\usepackage{color}

\newtheorem{Lemma}{Lemma}

\ifCLASSOPTIONcompsoc
  \usepackage[caption=false,font=normalsize,labelfont=sf,textfont=sf]{subfig}
\else
  \usepackage[caption=false,font=footnotesize]{subfig}
\fi
\usepackage[acronym,shortcuts]{glossaries}
\usepackage{tcolorbox}
\definecolor{mycolor}{RGB}{242,242,242}

\newacronym{5G}{5G}{fifth generation}
\newacronym{6G}{6G}{sixth generation}
\newacronym{ADR}{ADR}{antenna decentralized rate}
\newacronym{AER}{AER}{activity error rate}
\newacronym{AP}{AP}{access point}
\newacronym{ADC}{ADC}{analog-to-digital converter}
\newacronym{ABP}{ABP}{approximate belief propagation}
\newacronym{ADD}{ADD}{annealed discrete denoiser}
\newacronym{AMP}{AMP}{approximate message passing}
\newacronym{AUD}{AUD}{active user detection}
\newacronym{AWGN}{AWGN}{additive white Gaussian noise}
\newacronym{BLMMSE}{BLMMSE}{Bussgang linear MMSE}
\newacronym{BC}{BC}{belief combining}
\newacronym{BER}{BER}{bit error rate}
\newacronym{BiGAMP}{BiGAMP}{bilinear generalized approximate message passing}
\newacronym{BIP}{BIP}{bilinear inference problem}
\newacronym{CAMP}{CAMP}{convolutional AMP}
\newacronym{GAMP}{GAMP}{generalized AMP}
\newacronym{GAMP_full}{GAMP}{generalized approximate message passing}
\newacronym{GF}{GF}{grant-free}
\newacronym{BiGaBP}{BiGaBP}{bilinear Gaussian belief propagation}
\newacronym{BOD}{BOD}{Bayes-optimal denoiser}
\newacronym{BP}{BP}{belief propagation}
\newacronym{BS}{BS}{base station}
\newacronym{CAP}{CAP}{central AP}
\newacronym{CCU}{CCU}{central computing unit}
\newacronym{CDF}{CDF}{cumulative distribution function}
\newacronym{CLT}{CLT}{central limit theorem}
\newacronym{CPU}{CPU}{central processing unit}
\newacronym{CE}{CE}{channel estimation}
\newacronym{CF-mMIMO}{CF-mMIMO}{cell-free massive MIMO}
\newacronym{CSI}{CSI}{channel state information}
\newacronym{CSIDCO}{CSIDCO}{complex SIDCO}
\newacronym{DCC}{DCC}{dynamic cooperation clustering}
\newacronym{DFT}{DFT}{discrete Fourier transform}
\newacronym{DoF}{DoF}{degrees of freedom}
\newacronym{DQ}{DQ}{de-quantization}
\newacronym{DL}{DL}{deep learning}
\newacronym{DU}{DU}{deep unfolding}
\newacronym{eMBB}{eMBB}{enhanced mobile broadband}
\newacronym{ECF}{ECF}{estimate-compress-forward}
\newacronym{EP}{EP}{expectation propagation}
\newacronym{EPA}{EPA}{approximate EP}
\newacronym{FA}{FA}{false alarm}
\newacronym{FN}{FN}{factor node}
\newacronym{FG}{FG}{factor graph}
\newacronym{FTN}{FTN}{Faster-than-Nyquist}
\newacronym{GaBP}{GaBP}{Gaussian belief propagation}
\newacronym{GLM}{GLM}{generalized linear model}
\newacronym{GVAMP}{GVAMP}{generalized VAMP}
\newacronym{GMAMP}{GMAMP}{generalized MAMP}
\newacronym{IC}{IC}{interference cancellation}
\newacronym{IDD}{IDD}{iterative detection and decoding}
\newacronym{i.i.d.}{i.i.d.}{independent and identically distributed}
\newacronym{JACDE}{JACDE}{joint activity, channel and data estimation}
\newacronym{JACE}{JACE}{joint activity and channel estimation}
\newacronym{JCDE}{JCDE}{joint channel and data estimation}
\newacronym{KLD}{KLD}{Kullback-Leibler divergence}
\newacronym{LE}{LE}{linear estimator}
\newacronym{LSA}{LSA}{latent semantic analysis}
\newacronym{MAC}{MAC}{multiple-access channel}
\newacronym{MAMP}{MAMP}{memory AMP}
\newacronym{MAP}{MAP}{maximum \textit{a posteriori} probability}
\newacronym{MPDQ}{MPDQ}{message passing de-quantization}
\newacronym{MD}{MD}{miss-detection}
\newacronym{MF}{MF}{matched filter}
\newacronym{MFB}{MFB}{matched filter bound}
\newacronym{MF-EP}{MF-EP}{matched filter EP}
\newacronym{MM}{MM}{moment matching}
\newacronym{MNS}{MNS}{minimum norm solution}
\newacronym{MIMO}{MIMO}{multiple-input multiple-output}
\newacronym{MU-MIMO}{MU-MIMO}{multi-user MIMO}
\newacronym{MU}{MU}{multi-user}
\newacronym{mMIMO}{mMIMO}{Massive multiple-input multiple-output}
\newacronym{ML}{ML}{maximum likelihood}
\newacronym{MMSE}{MMSE}{minimum mean-square error}
\newacronym{mMTC}{mMTC}{massive machine type communications}
\newacronym{MMV-AMP}{MMV-AMP}{multiple measurement vector approximate message passing}
\newacronym{MSE}{MSE}{mean square error}
\newacronym{MUD}{MUD}{multi-user detection}
\newacronym{NLE}{NLE}{nonlinear estimator}
\newacronym{NR}{NR}{new radio}
\newacronym{NMSE}{NMSE}{normalized mean square error}
\newacronym{NMMSE}{NMMSE}{non-linear MMSE}
\newacronym{OAMP}{OAMP}{orthogonal AMP}
\newacronym{OFDM}{OFDM}{orthogonal frequency-division multiplexing}
\newacronym{PDA}{PDA}{probabilistic data association}
\newacronym{PDF}{PDF}{probability density function}
\newacronym{PMF}{PMF}{probability mass function}
\newacronym{PPP}{PPP}{Poisson point process}
\newacronym{PSK}{PSK}{phase-shift keying}
\newacronym{QP}{QP}{quadratic program}
\newacronym{QPSK}{QPSK}{quadrature PSK}
\newacronym{QAM}{QAM}{quadrature amplitude modulation}
\newacronym{SIDCO}{SIDCO}{sequential iterative decorrelation via convex optimization}
\newacronym{SD}{SD}{sphere decoding}
\newacronym{SE}{SE}{state evolution}
\newacronym{SGA}{SGA}{scalar Gaussian approximation}
\newacronym{SIC}{SIC}{soft interference cancellation}
\newacronym{SISO}{SISO}{single-input single-output}
\newacronym{SID}{SID}{self-iterative detection}
\newacronym{SLM}{SLM}{standard linear model}
\newacronym{SNR}{SNR}{signal-to-noise ratio}
\newacronym{Soft IC}{Soft IC}{soft interference cancellation}
\newacronym{SotA}{SotA}{state-of-the-art}
\newacronym{SVD}{SVD}{singular value decomposition}
\newacronym{SPA}{SPA}{sum-product algorithm}
\newacronym{TX}{TX}{transmit}
\newacronym{RX}{RX}{receive}
\newacronym{RMSE}{RMSE}{root mean square error}
\newacronym{UAMP}{UAMP}{unitary AMP}
\newacronym{UE}{UE}{user equipment}
\newacronym{ULA}{ULA}{uniform linear array}
\newacronym{URA}{URA}{unsourced random access}
\newacronym{URLLC}{URLLC}{ultra reliable low latency communications}
\newacronym{VAMP}{VAMP}{vector AMP}
\newacronym{VGA}{VGA}{vector Gaussian approximation}
\newacronym{VN}{VN}{variable node}
\newacronym{w.r.t.}{w.r.t.}{with respect to}
\newacronym{ZF}{ZF}{zero-forcing}
\newacronym{flops}{flops}{floating point operations}
\newacronym{CS}{CS}{compressed sensing}
\newacronym{MP}{MP}{message passing}
\newacronym{MPA}{MPA}{message passing algorithm}
\newacronym{DNN}{DNN}{deep neural network}
\newacronym{ASB}{ASB}{adaptively scaled belief}
\newacronym{LLR}{LLR}{log-likelihood ratio}
\newacronym{BAd-VAMP}{BAd-VAMP}{bilinear adaptive vector AMP}
\newacronym{LIP}{LIP}{linear inference problem}
\newacronym{AoA}{AoA}{angle of arrival}
\newacronym{LS}{LS}{least-squares}
\newacronym{mmWave}{mmWave}{millimeter-wave}
\newacronym{NLS}{NLS}{non-linear least-squares}
\newacronym{LoS}{LoS}{line-of-sight}
\newacronym{MoM}{MoM}{method of moments}
\newacronym{LMMSE}{LMMSE}{linearly constrained MMSE}
\newacronym{TDD}{TDD}{time-division duplexing}
\newacronym{FDD}{FDD}{frequency-division duplexing}
\newacronym{csi}{CSI}{channel state information}
\newacronym{bs}{BS}{base station}
\newacronym{ue}{UE}{user equipment}
\newacronym{RF}{RF}{radio frequency}
\newacronym{CRLB}{CRLB}{Cram\'er-Rao lower bound}
\usepackage{comment}

\begin{document}
%
% \setlength{\abovedisplayskip}{2pt} 
% \setlength{\belowdisplayskip}{2pt}
%\title{Layered Belief Propagation for Low-complexity Large MIMO Detection Based on Statistical Beams}
\title{
Reciprocity Calibration of Dual-Antenna Repeaters via MMSE Estimation
}

\author{
Shoma~Hara,~\IEEEmembership{Graduate Student Member,~IEEE},
~Takumi~Takahashi,~\IEEEmembership{Member,~IEEE},
~Hiroki~Iimori,~\IEEEmembership{Member,~IEEE}, 
~Hideki~Ochiai,~\IEEEmembership{Fellow,~IEEE},
~Erik~G.~Larsson,~\IEEEmembership{Fellow,~IEEE} \\
%~Shinsuke~Ibi,~\IEEEmembership{Senior Member,~IEEE},
%and~Hideki~Ochiai,~\IEEEmembership{Fellow,~IEEE}% <-this % stops a space
\thanks{
%This manuscript was submitted on November  th, 2024 \textit{(Corresponding author: Takumi Takahashi)}.
%This work was financially supported by JSPS KAKENHI Grant Numbers JP21H01332, JP22H01483, and JP23K13335. 
% Note: This manuscript is based on the findings obtained in our arXiv preprint, which has been submitted to IEEE TSP, and is currently under peer review. The confidential resubmission version is attached separately.

%This work was financially supported in part by JST, CRONOS, Japan Grant Number JPMJCS24N1, and in part by MIC/FORWARD under Grant JPMI240710001.

% This work was supported in part by JSPS KAKENHI under Grant JP23K13335, JP23K22754, and JP25H01111; in part by JST, CRONOS, Japan under Grant JPMJCS24N1; and in part by MIC/FORWARD under Grant JPMI240710001. \textit{(Corresponding author: Takumi Takahashi.)}

%This work was supported in part by JST, CRONOS, Japan under Grant JPMJCS24N1; and in part by MIC/FORWARD under Grant JPMI240710001. 
%\textit{(Corresponding author: Takumi Takahashi.)}

%This work was supported in part by JST, CRONOS, Japan under Grant JPMJCS24N1; and in part by MIC/FORWARD under Grant JPMI240710001.

S. Hara, T. Takahashi and H. Ochiai are with the Graduate School of Engineering,  University of Osaka, 2-1 Yamada-oka, Suita, 565-0871, Japan (e-mail: s-hara@wcs.comm.eng.osaka-u.ac.jp, \{takahashi, ochiai\}@comm.eng.osaka-u.ac.jp). H. Iimori is with Ericsson Research, Ericsson Japan K. K., Yokohama SYMPHOSTAGE West Tower 12F, 5-1-2 Minato Mirai, Yokohama, 220-0012, Japan (e-mail: hiroki.iimori@ericsson.com). E. G. Larsson is with the Department of Electrical Engineering (ISY), Link{\"o}ping University, 581~83 Link{\"o}ping, Sweden (e-mail: erik.g.larsson@liu.se).

A conference version of this work is planned to be submitted to a special session at IEEE SPAWC 2026. This work has been submitted to the IEEE for possible publication. Copyright may be transferred without notice, after which this version may no longer be accessible.

}% <-this % stops a space

%S. Ibi is with Faculty of Science and Engineering, Doshisha University 1-3 Tataramiyakodani, Kyotanabe, 610-0394, Japan (e-mail: sibi@mail.doshisha.ac.jp).}% <-this % stops a space
\vspace{-7mm}
}


\markboth{Journal of \LaTeX\ Class Files,~Vol.~, No.~, November~2024}%
{T.~Takahashi \MakeLowercase{\textit{et al.}}: Bare Demo of IEEEtran.cls for IEEE Journals}



\maketitle

\vspace{-7mm}
\begin{abstract}
This paper proposes a novel Bayesian reciprocity calibration method that consistently ensures uplink and downlink channel reciprocity in repeater-assisted \ac{MIMO} systems.
The proposed algorithm is formulated under the \ac{MMSE} criterion.
Its Bayesian framework incorporates complete statistical knowledge of the signal model, noise, and prior distributions, enabling a coherent design that achieves both low computational complexity and high calibration accuracy.
To further enhance phase-alignment accuracy, which is critical for calibration tasks, we develop a von Mises denoiser that exploits the fact that the target parameters lie on the circle in the complex plane.
Simulation results demonstrate that the proposed \ac{MMSE} algorithm achieves substantially improved estimation accuracy compared with conventional deterministic \ac{NLS} methods, while maintaining comparable computational complexity.
Furthermore, the proposed method exhibits remarkably fast convergence, making it well suited for practical implementation.
\end{abstract}

\begin{IEEEkeywords}
Reciprocity calibration, dual-antenna repeater, MIMO systems, Bayesian inference, MMSE
\end{IEEEkeywords}

\glsresetall

\IEEEpeerreviewmaketitle

\vspace{-2mm}

\newcommand{\varw}{\sigma_{\mathrm{w}}^2}
\newcommand{\varx}{\sigma_{\mathrm{x}}^2}

\newcommand{\xab}[1]{{x}_{#1, \mathrm{A} \to \mathrm{B}}}
\newcommand{\vab}[1]{{v}_{#1, \mathrm{A} \to \mathrm{B}}}
\newcommand{\xb}[1]{{x}_{#1, \mathrm{B}}}
\newcommand{\vb}[1]{{v}_{#1, \mathrm{B}}}
\newcommand{\xba}[2]{{x}_{#1, #2, \mathrm{B} \to \mathrm{A}}}
\newcommand{\vba}[2]{{v}_{#1, #2,  \mathrm{B} \to \mathrm{A}}}
\newcommand{\xa}[2]{{x}_{#1, #2}^\mathrm{a}}
\newcommand{\va}[2]{{v}_{#1, #2}^ \mathrm{a}}

\newcommand{\olxm}[1]{\overline{x}_{#1}}
\newcommand{\olxnm}[2]{\overline{x}_{#1,#2}}
\newcommand{\olvm}[1]{\overline{v}_{#1}}
\newcommand{\olvnm}[2]{\overline{v}_{#1,#2}}
\newcommand{\hatxm}[1]{\hat{x}_{#1}}
\newcommand{\hatxnm}[2]{\hat{x}_{#1,#2}}
\newcommand{\hatvm}[1]{\hat{v}_{#1}}
\newcommand{\hatvnm}[2]{\hat{v}_{#1,#2}}
\newcommand{\ckxm}[1]{\check{x}_{#1}}
\newcommand{\ckxnm}[2]{\check{x}_{#1,#2}}
\newcommand{\ckvm}[1]{\check{v}_{#1}}
\newcommand{\ckvnm}[2]{\check{v}_{#1,#2}} 
\newcommand{\ckzn}[1]{\check{z}_{#1}} 
\newcommand{\psinmz}[2]{\psi_{#1,#2}^{\mathrm{z}}} 
\newcommand{\psinz}[1]{\psi_{#1}^{\mathrm{z}}} 
\newcommand{\psinms}[2]{\psi_{#1,#2}^{\mathrm{s}}} 
\newcommand{\psins}[1]{\psi_{#1}^{\mathrm{s}}} 
\newcommand{\1}{\mbox{1}\hspace{-0.25em}\mbox{l}}

\section{Introduction}
\label{Chap:intro}
% 
%Root ../main.tex
% 
%===============
% Section 1.
%===============

Massive \ac{MIMO} systems rely on \ac{TDD} channel reciprocity to enable scalable \ac{CSI} acquisition at the \ac{BS}~\cite{Marzetta2010,Larsson2014}. 
In \ac{TDD} operation, the \ac{BS} estimates uplink channels from pilot transmissions and reuses them for downlink precoding, leveraging the large antenna array to coherently combine low-power uplink signals. 
However, while reciprocity holds for the physical propagation channel, practical transceiver hardware exhibits non-reciprocal behavior due to mismatches between the \ac{TX} and \ac{RX} \ac{RF} chains~\cite{BourdouxRAWCON03}, which distort the effective uplink and downlink channels observed at the \ac{BS}. 
As a result, uplink-based channel estimates cannot be directly reused for downlink precoding without compensation, making reciprocity calibration a fundamental requirement for accurate beamforming and for realizing the full performance potential of reciprocity-based massive \ac{MIMO} systems~\cite{Bjornson2014_IT}.

The problem of reciprocity calibration in \ac{TDD}-based \ac{MIMO} systems was formally identified and formulated in early work~\cite{Kaltenberger2010}.
This line of research established a fundamental calibration model in which the observed uplink and downlink channels differ by multiplicative hardware responses and demonstrated that accurate downlink beamforming requires explicit compensation of these non-reciprocal effects.
Following this formulation, extensive research has developed reciprocity calibration methods for conventional massive \ac{MIMO} systems with co-located antenna arrays.
% and direct \ac{BS}–\ac{UE} links. 
%
Representative approaches include relative calibration techniques that exploit internal signal exchange or mutual coupling among antennas within the same array~\cite{Vieira2017}, as well as over-the-air schemes based on bidirectional pilot transmissions~\cite{JiangOTACalibration}. 
Under the assumption of a single-hop channel and co-located hardware, these methods can estimate the relative \ac{TX} and \ac{RX} responses of the \ac{BS} antenna elements up to a common scalar ambiguity, which is sufficient for multiuser downlink precoding since this scalar does not affect beamforming directions or interference suppression.

Recent cellular network architectures are evolving beyond conventional co-located \acp{BS} toward more spatially distributed deployments to improve coverage, spatial multiplexing capability, and user experience at the cell edge~\cite{Ngo2017,Emil2020cellfree,Papaza2020cellfree,Iimori2021GrantFree,Takahashi2023ADCcellfree}.
While distributed massive \ac{MIMO} systems have been extensively studied to improve spatial diversity, coverage uniformity, and cell-edge performance by geographically separating antenna elements, their practical deployment remains challenging. 
Distributed \ac{MIMO} architectures typically rely on high-capacity, low-latency fronthaul or backhaul connections to coordinate multiple remote radio units, which significantly increases deployment cost and operational complexity~\cite{Ammar2022}.
%, posing challenges in terms of infrastructure availability, synchronization, and long-term maintenance.

In light of these challenges, repeater-assisted massive \ac{MIMO} has emerged as a pragmatic alternative that enables network densification without the architectural complexity of fully distributed deployments~\cite{SaraComMag25, Vu2025RSMA, Jopanya2025SWARM, Iimori2023GLOBECOM, Andersson2026DynamicTDD, Topal2025RA_MIMO, Chowdhury2026DualAntenna, Bai2026RepeaterSwarm, Le2025ARRNOMA, Larsson2024StabilityRepeaters}. 
By deploying physically small and low-cost wireless repeaters within the cell, additional effective scattering paths are created, improving coverage and channel rank while keeping all multi-antenna signal processing centralized at the \ac{BS}. 
Unlike distributed access points, repeaters do not require dedicated backhaul, tight inter-site synchronization, or major changes to existing network architectures, making them attractive from both deployment and maintenance perspectives. 
Ongoing standardization efforts on network-controlled repeaters further highlight their relevance as a realistic step toward performance levels approaching those of distributed massive \ac{MIMO} systems~\cite{Carvalho2025NCRIntro}.

Conventional reciprocity calibration methods are developed under the assumption that the channel between the \ac{BS} and the \acp{UE} is a single-hop propagation channel, affected only by the \ac{TX} and \ac{RX} hardware chains at the transceiver endpoints. 
In repeater-assisted architectures, this assumption no longer holds. 
Repeaters introduce additional \ac{TX} and \ac{RX} chains with their own analog impairments, amplification gains, and phase responses, which become cascaded with the \ac{BS}-\ac{UE} channel~\cite{Nie2020}.
As a result, the effective end-to-end channel is composed not only of the propagation environment and the \ac{BS}/\ac{UE} hardware but also of unknown and potentially asymmetric forward and reverse responses of the repeater. 
This cascaded non-reciprocity fundamentally alters the structure of the calibration problem and renders standard massive \ac{MIMO} calibration techniques designed for co-located arrays and direct links insufficient.

Only recently has repeater-induced non-reciprocity been addressed explicitly in the literature, although earlier works, \textit{e.g.}, \cite{MaPIMRC2015}, had already suggested that repeaters may not operate effectively in \ac{TDD} systems due to their inherently non-reciprocal behavior.
To address this issue, a calibration method has recently been proposed in which the repeater-induced hardware asymmetries are explicitly modeled and incorporated into the reciprocity calibration problem through a coupled \ac{LS} formulation~\cite{Larsson2024}.
The resulting estimation problem is solved using a \ac{NLS} and alternating \ac{NLS} approach to handle the coupling among multiple unknown hardware responses.
%, thereby providing a general framework for calibrating repeater-assisted links.
%
While this formulation represents an important step toward repeater-aware reciprocity calibration, the proposed solution is developed within a deterministic \ac{LS} framework and treats the repeater-induced hardware impairments as unknown but unstructured quantities. 
In practice,  prior information on these impairments may be available.
%these impairments often exhibit statistical structure arising from hardware design, manufacturing tolerances, and operating conditions~\cite{Bjornson2014_IT,Gustavsson2014,Athley2015}. 
%
Leveraging such prior information has the potential to significantly improve calibration robustness and accuracy, especially under limited training overhead or low \ac{SNR} conditions. 
This observation motivates the development of the proposed calibration algorithm in this paper.
Our specific contributions are:

\begin{itemize}
\item We develop a Bayesian formulation of the repeater-aware reciprocity calibration problem by explicitly incorporating prior statistical models of the repeater-induced non-reciprocity and the measurement noise, enabling a principled treatment of hardware-induced impairments.
\item Building upon this formulation, we propose a novel \ac{MMSE}-based calibration algorithm that performs iterative Bayesian bilinear inference using \ac{PDA}~\cite{Ito2025bipda,Takahashi2024JCTDD,Rayan2025JCDRE}, where the \ac{MMSE} updates are implemented via von Mises denoisers and unknown prior statistics are estimated through the \ac{MoM}.
\item Through extensive simulations, we demonstrate that the proposed algorithm achieves significantly improved robustness and accuracy over the \ac{SotA} deterministic \ac{NLS} methods~\cite{Larsson2024} in large-scale and low \ac{SNR} regimes, while maintaining computational complexity comparable to the basic \ac{NLS} method. 
%Furthermore, the proposed approach remains well-conditioned in high-dimensional settings due to the regularization effect of the prior distributions, leading to fast convergence.
\end{itemize}

\textit{Notation}: Sets of real and complex-valued numbers are denoted by $\mathbb{R}$ and $\mathbb{C}$, respectively.
Vectors and matrices are denoted by lowercase and uppercase letters, respectively.
The $(i, j)$-th entry of a matrix $\bm{A}$ is denoted by $\bm{A}(i,j)$, and the $i$-th column of $\bm{A}$ is denoted by $(\bm{A})_i$.
The $a\times a$ square identity matrix is denoted by $\bm{I}_a$.
For scalars $a_1,\dots,a_M$, $\mathrm{diag}(a_1,\dots,a_M)$ denotes a diagonal matrix with the elements $a_1, \dots,a_M$ on its main diagonal.
The conjugate, transpose, and Hermitian transpose are denoted by $(\cdot)^*$, $(\cdot)^{\mathsf{T}}$, and $(\cdot)^{\mathsf{H}}$, respectively.
The trace and determinant of $\bm{A}$ are denoted by $\mathrm{tr}\left(\bm{A}\right)$ and $\mathrm{det}\left(\bm{A}\right)$, respectively.
The Kronecker product of matrices $\bm{A}$ and $\bm{B}$ is denoted by $\bm{A}\otimes \bm{B}$.
The stacking of the columns of a matrix $\bm{A}$ on top of one another is denoted by $\mathrm{vec}\left(\bm{A}\right)$.
% $\mathrm{Diag}(\bm a)$ denotes a diagonal matrix with the elements of the vector $\bm a$ on its main diagonal.
%
The symbol $\|\cdot\|_2$ and $\|\cdot\|_\mathrm{F}$ denote the Euclidean norm of a vector and the Frobenius norm of a matrix, respectively.
The real part of a complex quantity is denoted by $\Re\{\cdot\}$, and the imaginary unit is denoted by $\mathrm{j} \triangleq \sqrt{-1}$.
The magnitude and phase of a complex number $c$ are denoted by $|c|$ and $\mathrm{arg}\left( c \right)$, respectively.
The notation $a\sim\mathcal{P}$ indicates that a random variable $a$ follows a probability distribution $\mathcal{P}$.
The circularly symmetric complex Gaussian distribution with mean $\mu$ and covariance $\bm{\varLambda}$ is denoted by $\mathcal{CN}\left(\mu,\bm\varLambda\right)$.
The real Gaussian distribution with mean $\mu$ and variance $\sigma^2$ is denoted by $\mathcal{N}(\mu,\sigma^2)$.
The uniform distribution on $[a,b)$ is denoted by $\mathcal{U}[a,b)$.
The \ac{PDF} of $x$ is denoted by $p_{\mathsf{x}}(\cdot)$, and the conditional \ac{PDF} of $x$ given $y$ is denoted by $p_{\mathsf{x|y}}(\cdot|\cdot)$.
The expectation is denoted by $\mathbb{E}[\cdot]$.
Finally, the notation $\mathcal{O}(\cdot)$ denotes the computational complexity order.

% \vspace{-1mm}
% \section{Preliminaries}
% \label{Chap:intro}
% \input{TXT/1_intro}

\section{System Model}
\label{Chap:sys}

% 
%Root ../main.tex
% 
%===============
% Section 2.
%===============
% 
\begin{figure}[t]
\centering
\includegraphics[width=0.98\columnwidth,keepaspectratio=true]{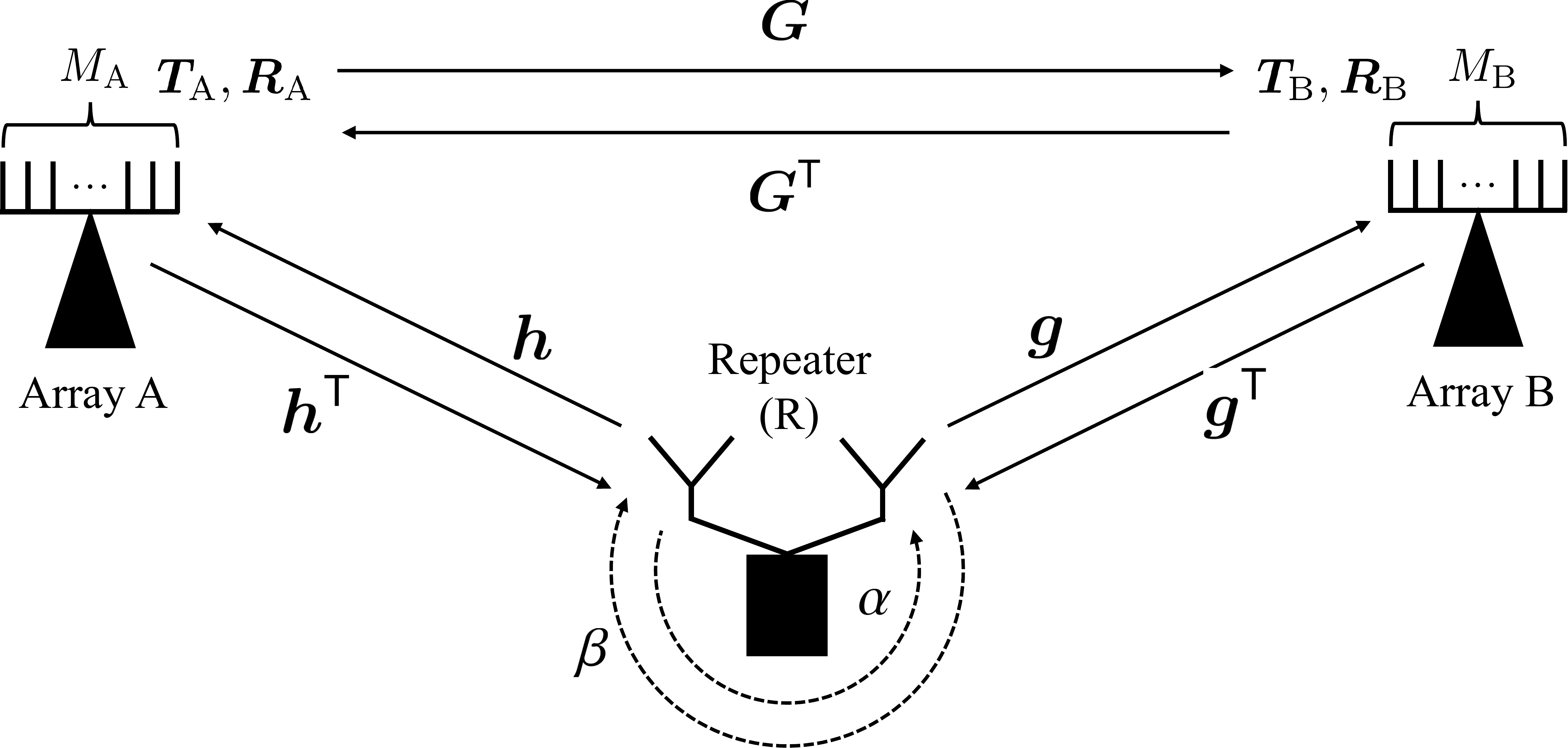}
\caption{Two antenna arrays, $\mathrm{A}$ and $\mathrm{B}$, and a repeater ($\mathrm{R}$). $\bm{G}$ denotes the propagation channel from $\mathrm{A}$ to $\mathrm{B}$ when $\mathrm{R}$ is turned off. Radio channels are represented by solid lines, whereas repeater gains are represented by dashed lines. The figure is adapted from~\cite{Larsson2024}.
% \egl{say that the figure is reproduced freely from [x]? (it looks identical)}}
% \tktk{Thank you. If we use it as a reproduction, we would probably need to contact the IEEE editorial office, so I’ll replace it with a slightly modified figure.}
}
\vspace{-4mm}
\label{fig:sys_model}
\end{figure}

Consider a communication scenario comprising two antenna arrays, denoted by $\mathrm{A}$ and $\mathrm{B}$, and a dual-antenna repeater $\mathrm{R}$, as illustrated in Fig.~\ref{fig:sys_model}~\cite{Larsson2024}. 
The numbers of antennas at arrays $\mathrm{A}$ and $\mathrm{B}$ are denoted by $M_{\mathrm{A}}$ and $M_{\mathrm{B}}$, respectively. 
In practice, $\mathrm{A}$ and $\mathrm{B}$ may represent a \ac{BS} and a mobile user terminal, or two distributed access points in a distributed \ac{MIMO} system.

The propagation channel between $\mathrm{A}$ and $\mathrm{B}$ consists of a direct path and an indirect path through the repeater $\mathrm{R}$. 
When the repeater is switched off, the direct channel from $\mathrm{A}$ to $\mathrm{B}$ is denoted by $\bm{G} \in \mathbb{C}^{M_{\mathrm{B}} \times M_{\mathrm{A}}}$, and, by reciprocity, the channel from $\mathrm{B}$ to $\mathrm{A}$ is given by $\bm{G}^{\mathsf{T}}$. 
We impose no structural assumptions on $\bm{G}$; it may include both \ac{LoS} and multipath components. 
Even when $\mathrm{R}$ is deactivated, it may scatter incident radio waves and thus serve as a reciprocal passive element, an effect implicitly captured in $\bm{G}$. 

The channel from $\mathrm{A}$ to $\mathrm{R}$ is denoted by $\bm{h}^{\mathsf{T}} \in \mathbb{C}^{1 \times M_{\mathrm{A}}}$, and the reverse channel from $\mathrm{R}$ to $\mathrm{A}$ is denoted by $\bm{h} \in \mathbb{C}^{M_{\mathrm{A}} \times 1}$. 
Similarly, the channel from $\mathrm{B}$ to $\mathrm{R}$ is represented by $\bm{g}^{\mathsf{T}} \in \mathbb{C}^{1 \times M_{\mathrm{B}}}$, and the reverse channel from $\mathrm{R}$ to $\mathrm{B}$ is denoted by $\bm{g} \in \mathbb{C}^{M_{\mathrm{B}} \times 1}$. 
When the repeater operates in its nominal configuration, the complex gains of the forward and reverse paths are denoted by $\alpha$ and $\beta$, respectively. 
In general, $\alpha \neq \beta$, implying that the repeater is nominally non-reciprocal.
This non-reciprocity arises from factors such as thermal fluctuations and hardware non-idealities.

In addition, the \ac{TX} and \ac{RX} branches of arrays $\mathrm{A}$ and $\mathrm{B}$ introduce hardware-induced non-reciprocity due to analog front-end impairments such as temperature drift and long-term component aging. 
This non-reciprocity is modeled by four diagonal matrices,
\begin{subequations}
\begin{eqnarray}
    \label{eq:ta}
    \bm{T}_{\mathrm{A}}
    \!\!&\!\!\triangleq\!\!&\!\!
    \mathrm{diag}\left(
    t_{\mathrm{A},1},t_{\mathrm{A},2}, \dots, t_{\mathrm{A}, M_{\mathrm{A}}}
    \right), \\
    \label{eq:ra}
    \bm{R}_{\mathrm{A}}
    \!\!&\!\!\triangleq\!\!&\!\!
    \mathrm{diag}\left(
    r_{\mathrm{A},1},r_{\mathrm{A},2}, \dots, r_{\mathrm{A},M_{\mathrm{A}}}
    \right), \\
    \label{eq:tb}
    \bm{T}_{\mathrm{B}}
    \!\!&\!\!\triangleq\!\!&\!\!
    \mathrm{diag}\left(
    t_{\mathrm{B},1},t_{\mathrm{B},2}, \dots, t_{\mathrm{B},M_{\mathrm{B}}}
    \right), \\
    \label{eq:rb}
    \bm{R}_{\mathrm{B}}
    \!\!&\!\!\triangleq\!\!&\!\!
    \mathrm{diag}\left(
    r_{\mathrm{B},1},r_{\mathrm{B},2}, \dots, r_{\mathrm{B},M_{\mathrm{B}}}
    \right), 
\end{eqnarray}
\label{eq:rc}%
\end{subequations}
whose diagonal entries represent the complex-valued reciprocity coefficients of the individual antenna elements, including  a possible drift of the oscillators at  $\mathrm{A}$ and $\mathrm{B}$.

%the diagonal of $\bm{R}_{\mathrm{A}}$ comprises $\{r_{\mathrm{A},1}, \dots, r_{\mathrm{A},M_{\mathrm{A}}}\}$ for the \ac{RX} chain of array $\mathrm{A}$, whereas that of $\bm{T}_{\mathrm{A}}$ comprises $\{t_{\mathrm{A},1}, \dots, t_{\mathrm{A}, M_{\mathrm{A}}}\}$ for the \ac{TX} chain. 
%Similarly, for array $\mathrm{B}$, we define $\bm{R}_{\mathrm{B}} = \mathrm{diag}(r_{B_1}, \dots, r_{B_{M_{\mathrm{B}}}})$ and $\bm{T}_{\mathrm{B}} = \mathrm{diag}(t_{B_1}, \dots, t_{B_{M_{\mathrm{B}}}})$. 
%

Using the above definitions, the effective noise-free channel from $\mathrm{A}$ to $\mathrm{B}$ can be expressed as
\begin{equation}
    \bm R_{\mathrm{B}}\left(\bm G + \alpha \bm g \bm h^{\mathsf{T}}\right)\bm T_{\mathrm{A}},
\end{equation}
while the channel from $\mathrm{B}$ to $\mathrm{A}$ is given by
\begin{equation}
     \bm R_{\mathrm{A}}\left(\bm G^{\mathsf{T}} + \beta \bm h \bm g^{\mathsf{T}}\right)\bm T_{\mathrm{B}}.
\end{equation}
In each expression, the first term represents the direct propagation channel, whereas the second term captures the repeater-assisted path.

For generality, neither array $\mathrm{A}$ nor $\mathrm{B}$ is assumed to be pre-calibrated.
Consequently, all parameters
\begin{equation}
\alpha,~ \beta,~ \bm G,~ \bm g,~ \bm h,~ \bm R_{\mathrm{A}},~ \bm T_{\mathrm{A}},~ \bm R_{\mathrm{B}},~ \bm T_{\mathrm{B}} \nonumber
\end{equation}
 are treated as \textit{a priori} unknown.

\section{State-of-the-Art Calibration Scheme}
\label{Chap:nls}
%===============
% Section 3.
%===============
%
We now describe the \ac{SotA} reciprocity calibration scheme for dual-antenna repeaters presented in~\cite{Larsson2024}. 
The key idea is to estimate the ratio between the forward and reverse repeater gains, $\alpha$ and $\beta$, using bi-directional pilot transmissions between arrays $\mathrm{A}$ and $\mathrm{B}$ under two distinct repeater configurations.
Specifically, two sets of measurements are conducted within one channel coherence interval:
\begin{itemize}
    \item[(i)] bi-directional transmission with $\mathrm{R}$ operating in its nominal configuration, and
    \item[(ii)] bi-directional transmission with $\mathrm{R}$ configured to apply a $\pi$-phase shift to both the forward and reverse gains.
\end{itemize}
This procedure yields four channel measurements as follows:
\begin{subequations}
\begin{eqnarray}
    \bm X_{\mathrm{AB}}^0 \!\!\!&=&\!\!\! \bm R_{\mathrm{B}}\left(\bm G + \alpha \bm g \bm h^{\mathsf{T}}\right)\bm T_{\mathrm{A}} + \bm W_{\mathrm{B}}^0, \\
    \bm X_{\mathrm{BA}}^0 \!\!\!&=&\!\!\! \bm R_{\mathrm{A}}\left(\bm G^{\mathsf{T}} + \beta \bm h \bm g^{\mathsf{T}}\right)\bm T_{\mathrm{B}} + \bm W_{\mathrm{A}}^0, \\
    \bm X_{\mathrm{AB}}^1 \!\!\!&=&\!\!\! \bm R_{\mathrm{B}}\left(\bm G - \alpha \bm g \bm h^{\mathsf{T}}\right)\bm T_{\mathrm{A}} + \bm W_{\mathrm{B}}^1, \\
    \bm X_{\mathrm{BA}}^1 \!\!\!&=&\!\!\! \bm R_{\mathrm{A}}\left(\bm G^{\mathsf{T}} - \beta \bm h \bm g^{\mathsf{T}}\right)\bm T_{\mathrm{B}} + \bm W_{\mathrm{A}}^1,
\end{eqnarray}
\label{eq:measurements}%
\end{subequations}
where $\bm W_{\mathrm{B}}^0 \in\mathbb{C}^{M_\mathrm{B}\times M_\mathrm{A}}$, $\bm W_{\mathrm{A}}^0\in\mathbb{C}^{M_\mathrm{A}\times M_\mathrm{B}}$, $\bm W_{\mathrm{B}}^1\in\mathbb{C}^{M_\mathrm{B}\times M_\mathrm{A}}$, and $\bm W_{\mathrm{A}}^1\in\mathbb{C}^{M_\mathrm{A}\times M_\mathrm{B}}$ denote the measurement noise matrices.

The objective of calibration is to determine the ratio
\begin{equation}
\label{eq:gamma}
    \gamma \triangleq \frac{\beta}{\alpha},
\end{equation}
since knowledge of $\gamma$ allows the repeater to be configured such that $\alpha = \beta$, thereby rendering it reciprocal to the network.

To address this estimation, the model is re-parameterized as
\begin{subequations}
\begin{eqnarray}
    \label{eq:H}
    \bm H
    \!\!&\!\!\triangleq\!\!&\!\!
    \bm R_{\mathrm{B}} \bm G \bm T_{\mathrm{A}} \in\mathbb{C}^{M_\mathrm{B}\times M_\mathrm{A}}, \\
    \label{eq:A}
    \bm A
    \!\!&\!\!\triangleq\!\!&\!\!
    \bm T_{\mathrm{A}}^{-1}\bm R_{\mathrm{A}} \in\mathbb{C}^{M_\mathrm{A}\times M_\mathrm{A}}, \\
    \label{eq:B}
    \bm B
    \!\!&\!\!\triangleq\!\!&\!\!
    \bm T_{\mathrm{B}}\bm R_{\mathrm{B}}^{-1} \in\mathbb{C}^{M_\mathrm{B}\times M_\mathrm{B}}, \\
    \label{eq:Z}
    \bm Z
    \!\!&\!\!\triangleq\!\!&\!\!
    \alpha \bm R_{\mathrm{B}}\bm g \bm h^{\mathsf{T}}\bm T_{\mathrm{A}} \in\mathbb{C}^{M_\mathrm{B}\times M_\mathrm{A}},
\end{eqnarray}
\label{eq:pram}
\end{subequations}
which leads to the equivalent measurement equations:
\begin{subequations}
\begin{eqnarray}
    \bm X_{\mathrm{AB}}^0
    \!\!&\!\!=\!\!&\!\!
    \bm H + \bm Z + \bm W_{\mathrm{B}}^0, \\
    \bm X_{\mathrm{BA}}^0
    \!\!&\!\!=\!\!&\!\!
    \bm A(\bm H + \gamma \bm Z)^{\mathsf{T}}\bm B + \bm W_{\mathrm{A}}^0, \\
    \bm X_{\mathrm{AB}}^1
    \!\!&\!\!=\!\!&\!\!
    \bm H - \bm Z + \bm W_{\mathrm{B}}^1, \\
    \bm X_{\mathrm{BA}}^1
    \!\!&\!\!=\!\!&\!\!
    \bm A(\bm H - \gamma \bm Z)^{\mathsf{T}}\bm B + \bm W_{\mathrm{A}}^1.
\end{eqnarray}
\end{subequations}

After straightforward calculations, the preprocessed measurements for the calibration task are obtained as follows:
\begin{subequations}
\begin{eqnarray}
\label{eq:R1}
\bm R_1
\!\!&\!\! \triangleq \!\!&\!\!
\frac{1}{2}\left(\bm X_{\mathrm{AB}}^0 + \bm X_{\mathrm{AB}}^1\right)
=
\bm H + \bm W_1, \\
\label{eq:R2}
\bm R_2 
\!\!&\!\! \triangleq \!\!&\!\!
\frac{1}{2}\left(\bm X_{\mathrm{AB}}^0 - \bm X_{\mathrm{AB}}^1\right)
=
\bm Z + \bm W_2, \\
\label{eq:R3}
\bm R_3
\!\!&\!\! \triangleq \!\!&\!\!
\frac{1}{2}\left(\bm X_{\mathrm{BA}}^0 + \bm X_{\mathrm{BA}}^1\right)
=
\bm A \bm H^{\mathsf{T}} \bm B + \bm W_3, \\
\label{eq:R4}
\bm R_4
\!\!&\!\! \triangleq \!\!&\!\!
\frac{1}{2}\left(\bm X_{\mathrm{BA}}^0 - \bm X_{\mathrm{BA}}^1\right)
=
\gamma \bm A \bm Z^{\mathsf{T}} \bm B + \bm W_4,
\end{eqnarray}
\label{eq:observation_reformulated}%
\end{subequations}
where
\begin{eqnarray}
    \bm{W}_1
    \!\!&\!\!\triangleq\!\!&\!\!
    \frac{1}{2}\left(\bm{W}_\mathrm{B}^0 + \bm{W}_\mathrm{B}^1\right),\quad
    \bm{W}_2
    \triangleq
    \frac{1}{2}\left(\bm{W}_\mathrm{B}^0 - \bm{W}_\mathrm{B}^1\right), \nonumber \\
    \bm{W}_3
    \!\!&\!\!\triangleq\!\!&\!\!
    \frac{1}{2}\left(\bm{W}_\mathrm{A}^0 + \bm{W}_\mathrm{A}^1\right),\quad
    \bm{W}_4
    \triangleq
    \frac{1}{2}\left(\bm{W}_\mathrm{A}^0 - \bm{W}_\mathrm{A}^1\right). \nonumber
\end{eqnarray}

Based on \eqref{eq:observation_reformulated}, the parameters defined in \eqref{eq:pram} are jointly estimated by minimizing the following \ac{NLS} criterion~\cite{Larsson2024}
\begin{eqnarray}
\label{eq:snls}
\!\!\!\!&\!\!\!\!&\!\!\!\!
    f = 
    \|\bm{R}_1 - \bm{H}\|^2_{\mathrm{F}} +
    \|\bm{R}_2 - \bm{Z}\|^2_{\mathrm{F}} \notag \\
\!\!\!\!&\!\!\!\!&\!\!\!\!
    \qquad +
    \|\bm{R}_3 - \bm{A}\bm{H}^\mathsf{T}\bm{B}\|^2_{\mathrm{F}} +
    \|\bm{R}_4 - \gamma\bm{A}\bm{Z}^\mathsf{T}\bm{B}\|^2_{\mathrm{F}}.
\end{eqnarray}
% \begin{eqnarray}
% \label{eq:nls}
%     \!\!\!\!&\!\!\!\!&\!\!\!\!
%     f=
%     \|\bm X_{\mathrm{AB}}^0 - (\bm H + \bm Z)\|^2_{\mathrm{F}}+
%     \|\bm X_{\mathrm{BA}}^0 - \bm A(\bm H + \gamma \bm Z)^{\mathsf{T}}\bm B\|^2_{\mathrm{F}} \nonumber \\
%     \!\!\!\!&\!\!\!\!&\!\!\!\!
%     +
%     \|\bm X_{\mathrm{AB}}^1 - (\bm H - \bm Z)\|^2_{\mathrm{F}}+
%     \|\bm X_{\mathrm{BA}}^1 - \bm A(\bm H - \gamma \bm Z)^{\mathsf{T}}\bm B\|^2_{\mathrm{F}}.
% \end{eqnarray}
%
The resulting solution provides a statistically consistent estimate of $\gamma$, enabling reciprocity calibration of the repeater.

\subsection{\ac{NLS} Algorithm}

% \egl{could this formulation with the pre-processed measurements be introduced upfront, to make this section a bit shorter?}
% \tktk{It’s been taken care of.}
Based on \eqref{eq:snls}, the following stepwise algorithm is employed for practical parameter estimation.

\subsubsection{Estimation of $\bm H$}

From the averaged measurements in \eqref{eq:R1},
% 
% \begin{equation}
% \label{R1}
%     \bm R_1 = \tfrac{1}{2}\left(\bm X_{\mathrm{AB}}^0 + \bm X_{\mathrm{AB}}^1\right),
% \end{equation}
% 
the minimizer \ac{w.r.t.} $\bm H$ is  
\begin{equation}
\label{eq:Hhat}
    \hat{\bm H} = \bm R_1.
\end{equation}

\subsubsection{Estimation of $\bm Z$}  
 
Using \eqref{eq:R2},
% 
% \begin{equation}
% \label{R2}
%     \bm R_2 = \tfrac{1}{2}\left(\bm X_{\mathrm{AB}}^0 - \bm X_{\mathrm{AB}}^1\right),
% \end{equation}
% 
the estimate of $\bm Z$ is obtained as
\begin{equation}
\label{eq:Zhat}
    \hat{\bm Z} = \mathcal{S}\{\bm R_2\},
\end{equation}
where $\mathcal{S}\{\cdot\}$ denotes the best rank-one approximation of a matrix, obtained from its dominant singular vector pair.

\subsubsection{Estimation of $\bm A$ and $\bm B$}
 
From \eqref{eq:R3},
% 
% \begin{equation}
% \label{R3}
%     \bm R_3 = \tfrac{1}{2}\left(\bm X_{\mathrm{BA}}^0 + \bm X_{\mathrm{BA}}^1\right),
% \end{equation}
% 
the matrices $\bm A$ and $\bm B$ are estimated by alternating projections.  
Starting from the initial values $\hat{\bm A} = \bm{I}_{M_\mathrm{A}}$ and $\hat{\bm B} = \bm{I}_{M_\mathrm{B}}$, each diagonal element is iteratively updated as
\begin{subequations}
\begin{eqnarray}
    \hat{\bm A}(i,i) \!\!\!&=&\!\!\! 
    \frac{ (\hat{\bm B}\hat{\bm H})_i^{\mathsf{H}} (\bm R_3^{\mathsf{T}})_i }
         { \| (\hat{\bm B}\hat{\bm H})_i \|_2^2 },\ \forall i\in\mathcal{A}, \\
    \hat{\bm B}(j,j) \!\!\!&=&\!\!\! 
    \frac{ (\hat{\bm A}\hat{\bm H}^{\mathsf{T}})_j^{\mathsf{H}} (\bm R_3)_j }
         { \| (\hat{\bm A}\hat{\bm H}^{\mathsf{T}})_j \|_2^2 },\ \forall j\in\mathcal{B},
\end{eqnarray}
\end{subequations}
where $\mathcal{A}\triangleq\left\{1,2,\ldots,M_\mathrm{A}\right\}$ and $\mathcal{B}\triangleq\left\{1,2,\ldots,M_\mathrm{B}\right\}$.
Normalization is applied at each iteration to prevent numerical instability~\cite{Larsson2024}.
 
\subsubsection{Estimation of $\gamma$}

Using \eqref{eq:R4},
\begin{equation}
\label{R4}
    \bm R_4 = \tfrac{1}{2}\left(\bm X_{\mathrm{BA}}^0 - \bm X_{\mathrm{BA}}^1\right),
\end{equation}
the gain ratio is estimated as
\begin{equation}
    \hat{\gamma} = 
    \frac{ \mathrm{tr}\!\left\{ (\hat{\bm A}\hat{\bm Z}^{\mathsf{T}} \hat{\bm B})^{\mathsf{H}} \bm R_4 \right\} }
         { \| \hat{\bm A}\hat{\bm Z}^{\mathsf{T}} \hat{\bm B} \|_\mathrm{F}^2 }.
\end{equation}

This \ac{NLS} algorithm yields an approximate solution to the original criterion given in \eqref{eq:snls}.
The procedure described above is summarized in Algorithm~\ref{alg:nls}.
% 
%Additional performance improvements can be obtained via alternating optimization, in which the updates of $\bm H$, $\bm A$, $\bm B$, $\bm Z$, and $\gamma$ are iteratively refined until convergence.
% 

%%%%%%%%%% ALG. 1 %%%%%%%%%%
%%%%%%%%%%%%%%%%%%%%%%%%%%%%%%%%%%%%%%%%%%%%
% Basic NLS
%%%%%%%%%%%%%%%%%%%%%%%%%%%%%%%%%%%%%%%%%%%%
\begin{algorithm}[t]
\caption{:\ NLS-based calibration algorithm}
\label{alg:nls}
\begin{algorithmic}[1]
\Require$\bm R_1,\bm R_2,\bm R_3,\bm R_4$
\Ensure $\hat{\bm H},\hat{\bm Z},\hat{\bm A},\hat{\bm B},\hat{\gamma}$

\vspace{0.5mm}
\LineComment{// 1) Estimation of $\bm H$} 
\State $\hat{\bm H} \gets \bm{R}_1$

\vspace{0.5mm}
\LineComment{// 2) Estimation of $\bm Z$} 
\State $\hat{\bm Z} \gets \mathcal{S}\{\bm{R}_2\}$ 

\vspace{0.5mm}
\LineComment{// 3) Estimation of $\bm A$ and $\bm B$} 
\State $\hat{\bm A} \gets \bm I_{M_{\mathrm{A}}},\ \hat{\bm B} \gets \bm I_{M_{\mathrm{B}}}$ \Comment{Initialization}

\Repeat

\State $\forall i\in\mathcal{A}: \hat{\bm A}(i,i)\gets \dfrac{(\hat{\bm B} \hat{\bm H})_i^\mathsf{H} (\bm R_3^\mathsf{T})_i}{\|(\hat{\bm B} \hat{\bm H})_i\|_2^2}$

\State $\forall j\in\mathcal{B}: \hat{\bm B}(j,j)\gets \dfrac{( \hat{\bm A} \hat{\bm H}^\mathsf{T} )_j^\mathsf{H} (\bm R_3)_j}{\|(\hat{\bm A} \hat{\bm H}^\mathsf{T})_j\|_2^2}$
% 
%   \For{each diagonal index $i$}
%     \State $\hat{\bm A}(i,i)\gets \dfrac{(\hat{\bm B} \hat{\bm H})_i^\mathsf{H} (\bm R_3^\mathsf{T})_i}{\|(\hat{\bm B} \hat{\bm H})_i\|_2^2}$
%   \EndFor
%   \For{each diagonal index $i$}
%     \State $\hat{\bm B}(i,i)\gets \dfrac{( \hat{\bm A} \hat{\bm H}^\mathsf{T} )_i^\mathsf{H} (\bm R_3)_i}{\|(\hat{\bm A} \hat{\bm H}^\mathsf{T})_i\|_2^2}$
%   \EndFor

\State $\hat{\bm A} \gets \hat{\bm A} \cdot \|\hat{\bm B}\|_\mathrm{F}$
\State $\hat{\bm B} \gets \hat{\bm B} / \|\hat{\bm B}\|_\mathrm{F}$
%\Comment{Normalization of $\hat{\bm B}$}

%\State Normalize $\hat{\bm A}$ and $\hat{\bm B}$

\Until{Num. of iterations reaches $N_{\mathrm{Iter}}$}

\vspace{0.5mm}
\LineComment{// 4) Estimation of $\gamma$}
\State $\hat\gamma \gets \dfrac{\operatorname{tr}\{ (\hat {\bm A} \hat {\bm Z}^{\mathsf{T}} \hat {\bm B})^\mathsf{H} \bm R_4 \}}{\| \hat {\bm A} \hat {\bm Z}^{\mathsf{T}} \hat {\bm B} \|_\mathrm{F}^2}$ 
\end{algorithmic}
\end{algorithm}

%%%%%%%%%%%%%%%%%%%%%%%%%%%%
%%%%%%%%%% ALG. 2 %%%%%%%%%%
%%%%%%%%%%%%%%%%%%%%%%%%%%%%%%%%%%%%%%%%%%%%
% Alternating-opt. NLS
%%%%%%%%%%%%%%%%%%%%%%%%%%%%%%%%%%%%%%%%%%%%
\begin{algorithm}[t!]
\caption{:\ Alternating-opt. NLS-based algorithm}
\label{alg:aonls}
\begin{algorithmic}[1]
\Require $\bm R_1,\bm R_2,\bm R_3,\bm R_4$
\Ensure $\hat{\bm H},\hat{\bm Z},\hat{\bm A},\hat{\bm B},\hat{\gamma}$

\vspace{0.5mm}
\noindent {/*\raisebox{3pt}[-5pt][0pt]{ ---------------- }} Initialization {\raisebox{3pt}[-5pt][0pt]{ ----------------- }}*/
\State Obtain initial estimates $\hat{\bm H},\hat{\bm Z},\hat{\bm A},\hat{\bm B},\hat{\gamma}$ using Alg. 1.

\vspace{0.5mm}
\noindent {/*\raisebox{3pt}[-5pt][0pt]{ --------- }} Alternating optimization {\raisebox{3pt}[-5pt][0pt]{ --------- }}*/
\Repeat
\LineComment{// 1) Update of $\bm H$} 
\State $\bm r \gets \left[\mathrm{vec}(\bm R_1);\ \mathrm{vec}(\bm R_3^\mathsf{T})\right]$, $\bm{\varXi} \gets \left[\bm I_{M_{\mathrm{A}}}\otimes \bm I_{M_{\mathrm{B}}};\ \hat{\bm A}\otimes \hat{\bm B}\right]$
\State $\mathrm{vec}(\hat{\bm H}) \gets (\bm{\varXi}^\mathsf{H}\bm{\varXi})^{-1}\bm{\varXi}^\mathsf{H} r$

\vspace{0.5mm}
\LineComment{// 2) Update of $\bm A$ and $\bm B$} 
\Repeat
  \State $\forall i\in\mathcal{A}: \hat{\bm A}(i,i)\gets
  \dfrac{
     (\hat{\bm B}\hat{\bm H})_i^\mathsf{H} (\bm R_3^\mathsf{T})_i
     + \hat{\gamma}^{*} (\hat{\bm B}\hat{\bm Z})_i^\mathsf{H} (\bm R_4^\mathsf{T})_i
  }{
     \|(\hat{\bm B}\hat{\bm H})_i\|_2^2 
     + |\hat{\gamma}|^2\|(\hat{\bm B}\hat{\bm Z})_i\|_2^2
  }$
  \State $\forall j\in\mathcal{B}: \hat{\bm B}(j,j)\gets
  \dfrac{
     (\hat{\bm A}\hat{\bm H}^\mathsf{T})_j^\mathsf{H} (\bm R_3)_j
     + \hat{\gamma}^{*} (\hat{\bm A}\hat{\bm Z}^\mathsf{T})_j^\mathsf{H} (\bm R_4)_j
  }{
     \|(\hat{\bm A}\hat{\bm H}^\mathsf{T})_j\|_2^2 
     + |\hat{\gamma}|^2\|(\hat{\bm A}\hat{\bm Z}^\mathsf{T})_j\|_2^2
  }$

\State $\hat{\bm A} \gets \hat{\bm A} \cdot \|\hat{\bm B}\|_\mathrm{F}$
\State $\hat{\bm B} \gets \hat{\bm B} / \|\hat{\bm B}\|_\mathrm{F}$
%\Comment{Normalize $\hat{\bm B}$}

\Until{Num. of iterations reaches $N_{\mathrm{Iter}}$.}

\vspace{0.5mm}
\LineComment{// 3) Update of $\bm Z$}
\State $\hat{\bm Z} \gets 
\mathcal{S}\!\left\{
  \bm R_2 +
  \dfrac{\hat{\gamma}^{*}\hat{\bm B}^{-1}\bm R_4^\mathsf{T}\hat{\bm A}^{-1}}{1+|\hat{\gamma}|^2}\,
\right\}$

\vspace{0.5mm}
\LineComment{// 4) Update of $\gamma$}
\State $\hat\gamma \gets 
\dfrac{
  \mathrm{tr}\!\left\{
    (\hat{\bm A}\hat{\bm Z}\hat{\bm B})^\mathsf{H}\bm R_4
  \right\}
}{
  \|\hat{\bm A}\hat{\bm Z}\hat{\bm B}\|_\mathrm{F}^2
}$

\Until{the objective function $f$ no longer decreases.}

\end{algorithmic}
\end{algorithm}

%%%%%%%%%%%%%%%%%%%%%%%%%%%%

\subsection{Alternating-Optimization NLS Algorithm}
\label{opt}

From \eqref{eq:observation_reformulated}, the information of $\bm{H}$ is contained not only in $\bm{R}_1$ but also in $\bm{R}_3$, and the information of $\bm{Z}$ is contained not only in $\bm{R}_2$ but also in $\bm{R}_4$.
Therefore, further performance improvement can be achieved by performing alternating optimization over all unknown variables, \textit{i.e.}, $\bm{H}$, $\bm{Z}$, $\bm{A}$, $\bm{B}$, and $\bm{\gamma}$.
For details, we refer the reader to~\cite{Larsson2024}, and for completeness, only the pseudocode is provided in Algorithm \ref{alg:aonls}.
However, this approach incurs a higher computational complexity.

\section{Proposed MMSE Algorithm}
\label{chap:prop}
% 
%Root ../main.tex
% 
%===============
% Section 5.
%===============

The \ac{SotA} methods are designed based on the \ac{NLS} criterion and therefore cannot exploit available statistical information, such as the \ac{SNR} and the prior distributions of the unknown variables.
Consequently, significant performance improvements can be expected by introducing appropriate mathematical models as prior information for the unknown variables and designing the algorithm according to the \ac{MMSE} criterion.
Motivated by this observation, this section develops a Bayesian inference algorithm that explicitly incorporates prior information and performs \ac{MMSE} estimation.

%\hl{Fundamentally, the problem is Bayes-optimal setting in nature when viewed from the perspective of estimating each individual variable.} \egl{not clear}\tktk{Yes, in the context of the present problem, this statement is not correct. Since not all prior distributions are known in this problem, it cannot be regarded as a Bayes-optimal setting.}
% To distinguish it from the well-known \ac{LMMSE} estimator, we refer to this proposed algorithm as \ac{NMMSE}.
% \egl{can't we just call it MMSE? No need to talk about LMMSE.}\tktk{We can!}
% \egl{okay, great, let's go for that}

\subsection{Prior Statistical Modeling of the Unknown Variables}

In wireless communication systems, the statistical characteristics and mathematical models of unknown variables can often be inferred from prior observations or estimations, particularly when these variables arise from hardware-induced effects or long-term statistical behaviors.
In this subsection, we introduce the statistical quantities and mathematical models explicitly employed in the proposed method.

\subsubsection{Complex-Valued Reciprocity Coefficients}

One of the well-established probabilistic models widely adopted to characterize hardware impairments in \ac{MIMO} arrays is the multiplicative stochastic impairment model~\cite{Gustavsson2014,Athley2015}.
In this model, the residual impairments after compensation are represented by multiplicative amplitude and phase errors.
Under this assumption, the complex-valued reciprocity coefficients introduced in \eqref{eq:rc} are modeled as
\begin{subequations}
\begin{eqnarray}
    t_{\mathrm{A},i}
    \!\!&\!\!=\!\!&\!\! \left(1+\epsilon_{\mathrm{TA},i}\right)e^{\mathrm{j}\theta_{\mathrm{TA},i}},\quad i\in\mathcal{A}, \\
    r_{\mathrm{A},i}
    \!\!&\!\!=\!\!&\!\! \left(1+\epsilon_{\mathrm{RA},i}\right)e^{\mathrm{j}\theta_{\mathrm{RA},i}},\quad i\in\mathcal{A}, \\
    t_{\mathrm{B},j}
    \!\!&\!\!=\!\!&\!\! \left(1+\epsilon_{\mathrm{TB},j}\right)e^{\mathrm{j}\theta_{\mathrm{TB},j}},\quad j\in\mathcal{B}, \\
    r_{\mathrm{B},j}
    \!\!&\!\!=\!\!&\!\! \left(1+\epsilon_{\mathrm{RB},j}\right)e^{\mathrm{j}\theta_{\mathrm{RB},j}},\quad j\in\mathcal{B}, 
\end{eqnarray}
\label{eq:rc_model}%
\end{subequations}
where $\epsilon_{\mathrm{TA},i}$, $\epsilon_{\mathrm{RA},i}$, $\epsilon_{\mathrm{TB},j}$, and $\epsilon_{\mathrm{RB},j}$ denote the amplitude errors, while $\theta_{\mathrm{TA},i}$, $\theta_{\mathrm{RA},i}$, $\theta_{\mathrm{TB},j}$, and $\theta_{\mathrm{RB},j}$ represent the phase errors, respectively.
%follow a real-valued Gaussian distribution $\mathcal{N}(0,\sigma_\epsilon^2)$, and  follow $\mathcal{N}(0,\sigma_\theta^2)$.
%
As discussed in Section II, the dominant sources of the complex-valued reciprocity coefficients at individual antenna elements are phase drift and phase errors induced by different oscillators employed at $\mathrm{A}$ and $\mathrm{B}$.
Therefore, when modeling the residual amplitude errors after compensation as real-valued Gaussian random variables following $\mathcal{N}(0,\sigma_\epsilon^2)$~\cite{Gustavsson2014}, we assume that their variance is sufficiently small, \textit{i.e.}, $\sigma_\epsilon\ll 1$.
% and $\sigma_\epsilon^2\approx 0$.

\subsubsection{Measurement Noise}

The observation noise may include, in addition to Gaussian noise originating from thermal noise,   spatially colored interference.
Moreover, depending on the sampling scheme, temporal correlations may also arise.
Therefore, in this paper, in order to appropriately account for such spatio-temporal correlations, the observation noise matrix in \eqref{eq:measurements} is modeled as a circularly symmetric multivariate complex Gaussian random matrix.
Assuming that the spatial (column-wise) and temporal (row-wise) correlations are separable and admit a Kronecker representation, the \acp{PDF} of the observation noise matrices in \eqref{eq:measurements} are expressed as
\begin{subequations}
\begin{eqnarray}
\!\!\!\!&\!\!\!\!\!\!\!\!&\!\!\!\!
p_{\bm{\mathsf{W}}_\mathrm{A}^\mathrm{x}}\left(\bm{W}_\mathrm{A}^\mathrm{x}\right)
=
\frac{\exp\left[-
\mathrm{tr}\left(
\bar{\bm{\varOmega}}_\mathrm{A}^{-1}
\bm{W}_\mathrm{A}^\mathrm{x}
\bar{\bm{\varPsi}}_\mathrm{A}^{-1}
\left(\bm{W}_\mathrm{A}^\mathrm{x}\right)^\mathrm{H}
\right)\right]}{\pi^{M_\mathrm{A}M_\mathrm{B}}\mathrm{det}\left(\bar{\bm{\varOmega}}_\mathrm{A}\right)^{M_\mathrm{B}}\mathrm{det}\left(\bar{\bm{\varPsi}}_\mathrm{A}\right)^{M_\mathrm{A}}},
\\
\!\!\!\!&\!\!\!\!\!\!\!\!&\!\!\!\!
p_{\bm{\mathsf{W}}_\mathrm{B}^\mathrm{x}}\left(\bm{W}_\mathrm{B}^\mathrm{x}\right)
=
\frac{\exp\left[-
\mathrm{tr}\left(
\bar{\bm{\varOmega}}_\mathrm{B}^{-1}
\bm{W}_\mathrm{B}^\mathrm{x}
\bar{\bm{\varPsi}}_\mathrm{B}^{-1}
\left(\bm{W}_\mathrm{B}^\mathrm{x}\right)^\mathrm{H}
\right)\right]}{\pi^{M_\mathrm{B}M_\mathrm{A}}\mathrm{det}\left(\bar{\bm{\varOmega}}_\mathrm{B}\right)^{M_\mathrm{A}}\mathrm{det}\left(\bar{\bm{\varPsi}}_\mathrm{B}\right)^{M_\mathrm{B}}},
\end{eqnarray}
\label{eq:covariance}%
\end{subequations}
where $\mathrm{x}\in\left\{0,1\right\}$.
The spatial correlation matrices at $\mathrm{A}$ and $\mathrm{B}$ are denoted by $\bar{\bm{\varOmega}}_\mathrm{A}\in\mathbb{C}^{M_\mathrm{A}\times M_\mathrm{A}}$ and $\bar{\bm{\varOmega}}_\mathrm{B}\in\mathbb{C}^{M_\mathrm{B}\times M_\mathrm{B}}$, respectively, while the temporal correlation matrices are denoted by $\bar{\bm{\varPsi}}_\mathrm{A}\in\mathbb{C}^{M_\mathrm{B}\times M_\mathrm{B}}$ and $\bar{\bm{\varPsi}}_\mathrm{B}\in\mathbb{C}^{M_\mathrm{A}\times M_\mathrm{A}}$.
For notational simplicity, we further define $\bm{\varOmega}_\mathrm{A} \triangleq \bar{\bm{\varOmega}}_\mathrm{A}/2$, $\bm{\varOmega}_\mathrm{B} \triangleq \bar{\bm{\varOmega}}_\mathrm{B}/2$, $\bm{\varSigma}_\mathrm{A} \triangleq \bar{\bm{\varSigma}}_\mathrm{A}/2$, and $\bm{\varSigma}_\mathrm{B} \triangleq \bar{\bm{\varSigma}}_\mathrm{B}/2$ to denote the covariance matrices of the pre-processed measurement noise in \eqref{eq:observation_reformulated}.

In what follows, these covariance matrices are assumed to be known (or can be accurately estimated \textit{a priori}) and are used in the proposed method.
However, if such information is not fully available, the covariance matrices can be constructed using only the available information, and the method can be executed accordingly.
Note that the proposed method does not require complete knowledge of the covariance matrices.
Even when the full covariance matrices are available, the proposed method can be implemented using only their diagonal elements (\textit{i.e.}, the variance matrices) to reduce the computational complexity.
The specific implementation of the proposed method is left to the system designer, depending on practical constraints.

\subsection{Overview}

%%%%%%%%%%%%%%%%%%%%%%%%%%%%%%%%%%%%%%
% FIG: FG
%%%%%%%%%%%%%%%%%%%%%%%%%%%%%%%%%%%%%%
\begin{figure}[t]
\centering
\includegraphics[width=0.9\columnwidth,keepaspectratio=true]{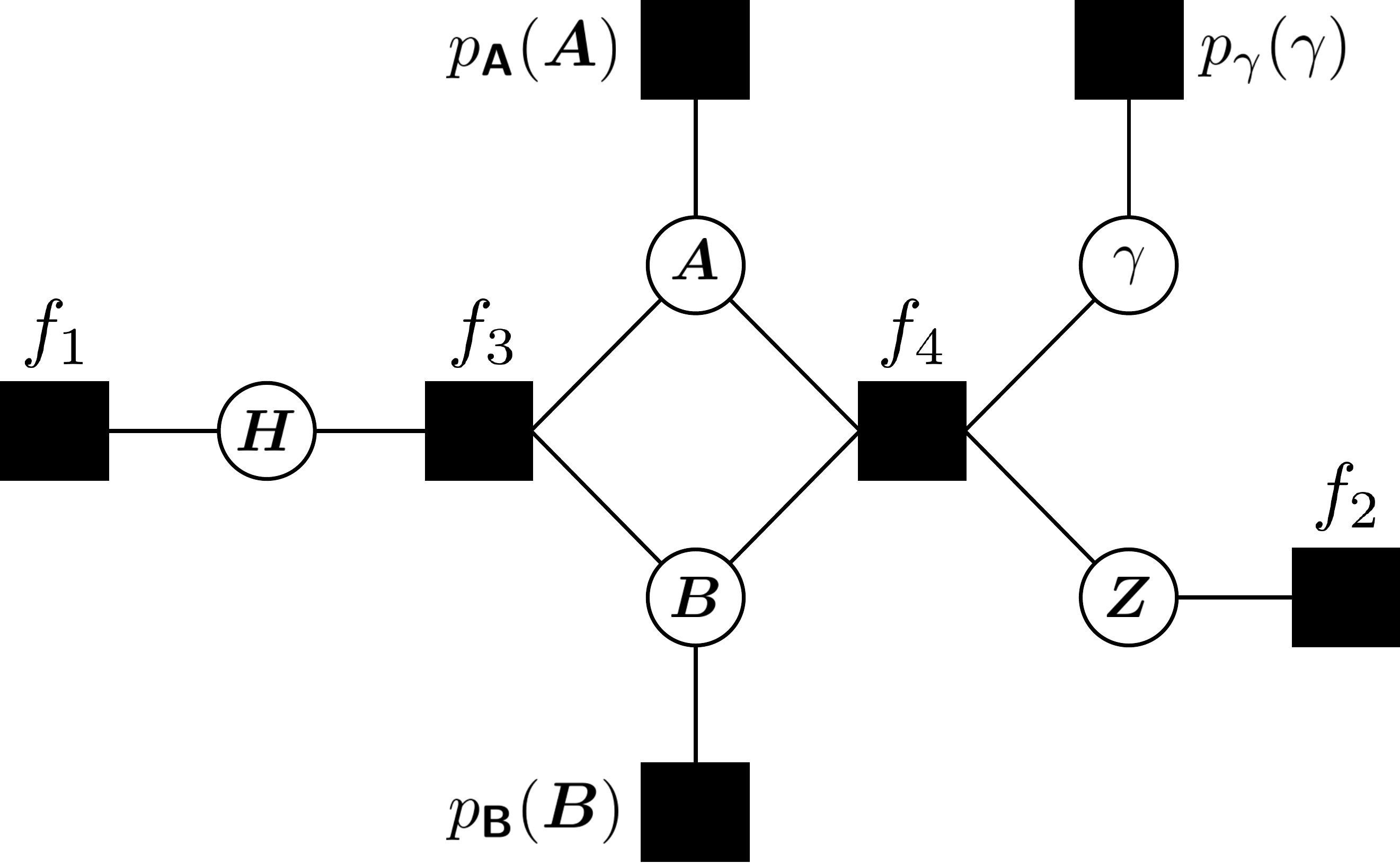}
\caption{Illustration of the \ac{FG} designed for the proposed \ac{MMSE} algorithm.}
\vspace{-4mm}
\label{fig:FG}
\end{figure}
%%%%%%%%%%%%%%%%%%%%%%%%%%%%%%%%%%%%%%

To provide an overview of the proposed algorithm, we present the \ac{FG} corresponding to the reciprocity calibration task in Fig.~\ref{fig:FG}.
In this graph, the square nodes ($\blacksquare$) represent \acp{FN} associated with the observations, where $f_1$, $f_2$, $f_3$, and $f_4$ correspond to \eqref{eq:R1}, \eqref{eq:R2}, \eqref{eq:R3}, and \eqref{eq:R4}, respectively.
The circular nodes ($\bigcirc$) represent \acp{VN} corresponding to the unknown variables.

When loops exist in the \ac{FG}, the estimation procedure plays a crucial role in determining the accuracy and stability of the algorithm.
The proposed algorithm follows an estimation procedure similar to that of the basic \ac{NLS} algorithm, in which the estimates of the unknown variables and their associated \acp{MSE}, if available, are sequentially propagated over the \ac{FG}.
Finally, by integrating the messages propagated over the \ac{FG} from all observations, $\gamma$ is estimated in the \ac{MMSE} sense.

The key factors that enable performance improvements over the \ac{NLS} algorithm are threefold: i) the explicit incorporation of the statistical properties of the measurement noise within the Bayesian framework; ii) the use of the prior statistical modeling introduced in \eqref{eq:rc_model} for the reciprocity coefficients (\textit{i.e.}, $\bm{A}$ and $\bm{B}$) as prior probability distributions; and iii) the estimation of the prior distribution of $\gamma$ via the \acf{MoM}, followed by \ac{MMSE} estimation based on the resulting prior.
In the following, we describe the details of the proposed algorithm.

\subsection{MMSE Algorithm}
\label{NMMSE}

\subsubsection{Estimation of $\bm H$}

From \eqref{eq:H}, since $\bm{H}$ includes the direct channel $\bm{G}$ from $\mathrm{A}$ to $\mathrm{B}$, it is practically difficult to assume a statistical prior model for $\bm{H}$.
Therefore, as in the \ac{NLS} algorithm, we use \eqref{eq:Hhat} as the point estimate of $\bm{H}$ and propagate the associated covariance matrices representing the estimation error information, \textit{i.e.}, $\bm{\varOmega}_\mathrm{B}$ and $\bm{\varPsi}_\mathrm{B}$.

\subsubsection{Estimation of $\bm Z$}

Similarly, from \eqref{eq:Z}, since $\bm{Z}$ includes the indirect channels $\bm{g}$ and $\bm{h}$ via $\mathrm{R}$, it is difficult to assume a statistical prior model for $\bm{Z}$.
Therefore, as in the \ac{NLS} algorithm, we propagate \eqref{eq:Zhat} as the point estimate of $\bm{Z}$.
Since the estimation error of the rank-one approximation cannot be analytically evaluated, only the estimate is propagated.

\subsubsection{Estimation of $\bm A$ and $\bm B$}

%%%%%%%%%%%%%%%%%%%%%%%%%%%%%%%%%%%%%%
% FIG: f3
%%%%%%%%%%%%%%%%%%%%%%%%%%%%%%%%%%%%%%
\begin{figure}[t]
\centering
\includegraphics[width=0.9\columnwidth,keepaspectratio=true]{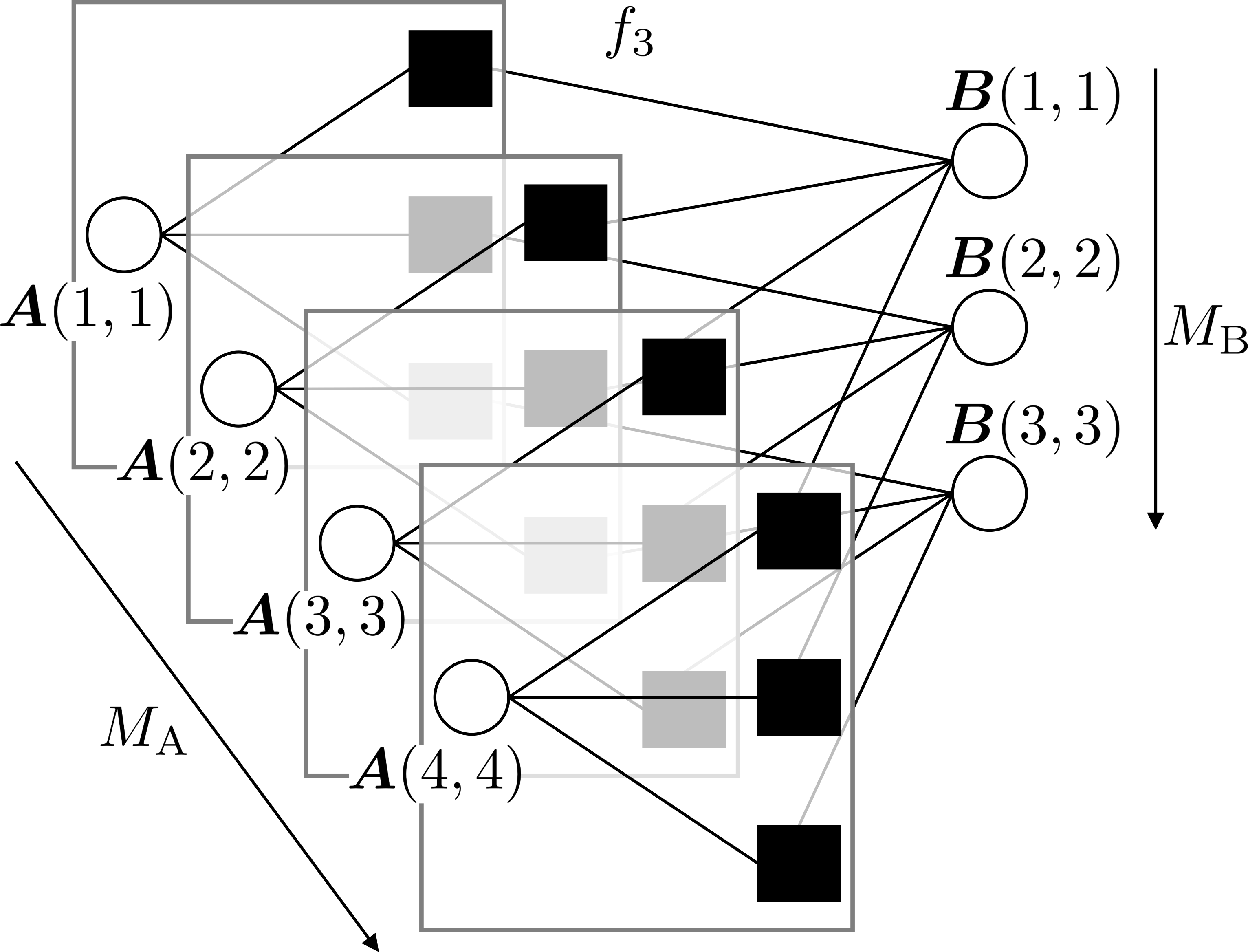}
\caption{Illustration of the tripartite \ac{FG} for bilinear inference of $\bm{A}$ and $\bm{B}$ with $(M_\mathrm{A},M_\mathrm{B}) = (4,3)$. }
\vspace{-3mm}
\label{fig:f3}
\end{figure}
%%%%%%%%%%%%%%%%%%%%%%%%%%%%%%%%%%%%%%

Based on \eqref{eq:R3}, we perform \ac{MMSE} estimation of $\bm{A}$ and $\bm{B}$.
To exploit the point estimate of $\bm{H}$ and its associated error covariance matrices, together with the prior statistical model introduced in \eqref{eq:rc_model}, we adopt an iterative alternating estimation approach based on \ac{PDA} within the Bayesian bilinear inference framework~\cite{Ito2025bipda,Takahashi2024JCTDD,Rayan2025JCDRE,Ito2024,Iimori2023JACE}.
To facilitate the understanding of the message update rules, Fig.~\ref{fig:f3} illustrates in detail the \ac{FG} composed of the \acp{FN} $f_3$ corresponding to \eqref{eq:R3} and the \acp{VN} associated with the unknown variables in $\bm{A}$ and $\bm{B}$.
As shown in the figure, the \ac{FG} centered at $f_3$ is relatively \textit{sparse}, and since no short loops 
%of length four 
exist, high-accuracy convergence can be expected with a relatively small number of iterations.

Denoting the estimates of $\bm{A}$ and $\bm{B}$ by $\hat{\bm{A}}$ and $\hat{\bm{B}}$, respectively, their \acp{MSE} for all $i\in\mathcal{A}$ and $j\in\mathcal{B}$ are given by
\begin{subequations}
\begin{eqnarray}
    \!\!\!\!\!\!\!\!\!\!\!\!
    \hat{v}_{\mathrm{A},i}
    \!\!&\!\!\triangleq\!\!&\!\!
    \mathbb{E}\left[|\tilde{\bm{A}}(i,i)|^2\right],\ \tilde{\bm{A}}(i,i)\triangleq \bm{A}(i,i)-\hat{\bm{A}}(i,i),\\
    \!\!\!\!\!\!\!\!\!\!\!\!
    \hat{v}_{\mathrm{B},j} 
    \!\!&\!\!\triangleq\!\!&\!\!
    \mathbb{E}\left[|\tilde{\bm{B}}(j,j)|^2\right],\ \tilde{\bm{B}}(j,j)\triangleq \bm{B}(j,j)-\hat{\bm{B}}(j,j),
\end{eqnarray}
\end{subequations}
where $\tilde{\bm{A}}$ and $\tilde{\bm{B}}$ represent the corresponding estimation error matrices.
At the first iteration, the estimates and their \acp{MSE} are appropriately initialized as $\hat{\bm A} = \bm{I}_{M_\mathrm{A}}$, $\hat{\bm B} = \bm{I}_{M_\mathrm{B}}$, and $\hat{v}_{\mathrm{A},i}=\hat{v}_{\mathrm{B},j} = 1$ for all $i\in\mathcal{A}$ and $j\in\mathcal{B}$.

First, to estimate the $(i,i)$-th entry of $\bm{A}$, we focus on the $i$-th row of the observation in \eqref{eq:R3}, which yields the following observation equation:
\begin{equation}
\label{eq:R3i}
    \left(\bm R_3^{\mathsf{T}}\right)_i = \left(\bm B \bm H\right)_i \bm A(i,i) + \left(\bm W_3^{\mathsf{T}}\right)_i \in \mathbb{C}^{M_{\mathsf{B}} \times 1}.
\end{equation}
Substituting $\bm{H} = \hat{\bm{H}} + \tilde{\bm{H}}$ and $\bm{B} = \hat{\bm{B}} + \tilde{\bm{B}}$ into \eqref{eq:R3i} yields
\begin{eqnarray}
\label{eq:R3i_v2}
\!\!\!\!&\!\!\!\!\!\!\!\!&\!\!\!\!
\left(\bm R_3^{\mathsf{T}}\right)_i
=
\left(\hat{\bm{B}}\hat{\bm{H}}\right)_i\bm{A}(i,i) \nonumber \\
\!\!\!\!&\!\!\!\!\!\!\!\!&\!\!\!\!
\qquad\qquad+
\left(\tilde{\bm{B}}\hat{\bm{H}}\right)_i\bm{A}(i,i)
+
\left(\hat{\bm{B}}\tilde{\bm{H}}\right)_i\bm{A}(i,i) \nonumber \\
\!\!\!\!&\!\!\!\!\!\!\!\!&\!\!\!\!
\qquad\qquad+
\left(\tilde{\bm{B}}\tilde{\bm{H}}\right)_i\bm{A}(i,i)
+
\left(\bm W_3^{\mathsf{T}}\right)_i.
\end{eqnarray}
In accordance with the \ac{CLT}, the terms in the second and third lines of \eqref{eq:R3i_v2} are collectively approximated by a multivariate complex Gaussian random vector.
This approximation is referred to as \ac{VGA}~\cite{Chockalingam2014}.
Accordingly, the conditional \ac{PDF} of \eqref{eq:R3i_v2}, given $\bm{A}(i,i)$, can be expressed as
\begin{eqnarray}
\label{eq:pdfA}
\!\!\!\!&\!\!\!\!\!\!\!\!&\!\!\!\!
p_{\left(\bm{\mathsf{R}}_3^{\mathsf{T}}\right)_i|\bm{\mathsf{A}}(i,i)}
\left(
\left(\bm{R}_3^{\mathsf{T}}\right)_i|\bm{A}(i,i)
\right) \nonumber \\
\!\!\!\!&\!\!\!\!\!\!\!\!&\!\!\!\!
\quad
\propto
\exp\Bigl[
-\left(
\left(\bm{R}_3^{\mathsf{T}}\right)_i-\left(\hat{\bm{B}}\hat{\bm{H}}\right)_i\bm{A}(i,i)
\right)^\mathsf{H}
% \mathrm{Diag}\left(
% \bm{v}_{\mathrm{A},i}
% \right)^{-1}
\bm{V}_{\mathrm{A},i}^{-1}
\nonumber \\
\!\!\!\!&\!\!\!\!\!\!\!\!&\!\!\!\!
\qquad\qquad\qquad
\times \left(
\left(\bm{R}_3^{\mathsf{T}}\right)_i-\left(\hat{\bm{B}}\hat{\bm{H}}\right)_i\bm{A}(i,i)
\right)
\Bigr],
\end{eqnarray}
where
\begin{eqnarray}
\label{eq:VAi}
\!\!\!\!&\!\!\!\!\!\!\!\!&\!\!\!\!
\bm{V}_{\mathrm{A},i}
\triangleq
\mathbb{E}\Bigl[
\left(
\left(\bm R_3^{\mathsf{T}}\right)_i-\left(\hat{\bm{B}}\hat{\bm{H}}\right)_i\bm{A}(i,i)
\right) \nonumber \\
\!\!\!\!&\!\!\!\!\!\!\!\!&\!\!\!\!
\qquad\qquad\qquad
\times
\left(
\left(\bm R_3^{\mathsf{T}}\right)_i-\left(\hat{\bm{B}}\hat{\bm{H}}\right)_i\bm{A}(i,i)
\right)^\mathsf{H}
\Bigr] \nonumber \\
\!\!\!\!&\!\!\!\!\!\!\!\!&\!\!\!\!
\quad=
\bm{\varPsi}_\mathrm{A}+\hat{\bm{B}}\bm{\varOmega}_\mathrm{B}\hat{\bm{B}}^\mathsf{H} +
\mathrm{diag}\left(v_{\mathrm{A},i1},\dots,v_{\mathrm{A},iM_\mathrm{B}}\right),
\end{eqnarray}
%$\bm v_{\mathrm{A},i} \triangleq \left[v_{\mathrm{A},i1},\dots,v_{\mathrm{A},iM_\mathrm{B}}\right]^{\mathsf{T}}\in\mathbb{R}^{M_\mathrm{B}\times 1}$ with
with
\begin{equation}
v_{\mathrm{A},ij} \triangleq 
\left(
|\hat{\bm H}(j, i)|^2 + \bm{\varOmega}_\mathrm{B}(j,j)
\right)
\hat{v}_{\mathrm{B},j}.
\end{equation}

% \begin{eqnarray}
% \label{eq:vAi}
% \!\!\!\!&\!\!\!\!\!\!\!\!&\!\!\!\!
% v_{\mathrm{A},ij}
% \triangleq
% \mathbb{E}\left[
% \left|\left(\tilde{\bm{B}}\hat{\bm{H}}\right)_i
% +
% \left(\hat{\bm{B}}\tilde{\bm{H}}\right)_i
% +
% \left(\tilde{\bm{B}}\tilde{\bm{H}}\right)_i\right|^2
% %\left(\bm R_3^{\mathsf{T}}\right)_i-\left(\hat{\bm{B}}\hat{\bm{H}}\right)_i\bm{A}(i,i)|^2
% \right]|\bm{A}(i,i)|^2 \nonumber \\
% \!\!\!\!&\!\!\!\!\!\!\!\!&\!\!\!\!
% \quad=
% |\hat{\bm H}(j, i)|^2
% \hat{v}_{\mathrm{B},j}
% +
% \hat{v}_{\mathrm{H}}
% |\hat{\bm B}(j,j)|^2
% +
% \hat{v}_{\mathrm{H}}
% \hat{v}_{\mathrm{B},j}.
% \end{eqnarray}

Next, by expanding and rearranging the exponent in \eqref{eq:pdfA} and completing the square \ac{w.r.t.} $\bm{A}(i,i)$, we have
\begin{eqnarray}
\label{eq:likA}
\!\!\!\!&\!\!\!\!\!\!\!\!&\!\!\!\!
p_{\left(\bm{\mathsf{R}}_3^{\mathsf{T}}\right)_i|\bm{\mathsf{A}}(i,i)}
\left(
\left(\bm{R}_3^{\mathsf{T}}\right)_i|\bm{A}(i,i)
\right) \nonumber \\
\!\!\!\!&\!\!\!\!\!\!\!\!&\!\!\!\!
\qquad
\propto
\exp\left[-\frac{|\bar{\bm{A}}(i,i)-\bm{A}(i,i)|^2}{\bar{v}_{\mathrm{A},i}}\right],
\end{eqnarray}
with
\begin{subequations}
\begin{eqnarray}
\bar{\bm{A}}(i,i)
\!\!&\!\!=\!\!&\!\!
\frac{1}{\psi_{\mathrm{A},i}}
\left(\hat{\bm B}\hat{\bm H}\right)_i^{\mathsf{H}} 
\bm{V}_{\mathrm{A},i}^{-1}
\left(\bm R_3^{\mathsf{T}}\right)_i, \\
\bar{v}_{\mathrm{A},i}
\!\!&\!\!=\!\!&\!\!
\frac{1}{\psi_{\mathrm{A},i}},
\end{eqnarray}
\label{eq:barA}%
\end{subequations}
where
\begin{equation}
\psi_{\mathrm{A},i} \triangleq \left(\hat{\bm B}\hat{\bm H}\right)_i^{\mathsf{H}} \bm{V}_{\mathrm{A},i}^{-1} \left(\hat{\bm B}\hat{\bm H}\right)_i.
\end{equation}
From \eqref{eq:likA}, $\bar{\bm{A}}(i,i)$ can be regarded as an \ac{AWGN} observation of $\bm{A}(i.i)$.

Finally, based on the prior statistical model in \eqref{eq:rc_model}, we consider the prior distribution of $\bm{A}(i,i)$.
By substituting \eqref{eq:ta} and \eqref{eq:ra} into \eqref{eq:A}, we obtain
\begin{eqnarray}
\label{eq:Aii}
    \bm{A}(i,i)
    \!\!&\!\!=\!\!&\!\!
    \frac{r_{\mathrm{A},i}}{t_{\mathrm{A},i}}
    = 
    \frac{\left(1+\epsilon_{\mathrm{RA},i}\right)e^{\mathrm{j}\theta_{\mathrm{RA},i}}}
    {\left(1+\epsilon_{\mathrm{TA},i}\right)e^{\mathrm{j}\theta_{\mathrm{TA},i}}} =
    \left[1+\frac{\epsilon_i}{1+\epsilon_{\mathrm{TA},i}}\right]e^{\mathrm{j}\theta_i} \nonumber \\
    \!\!&\!\!\approx\!\!&\!\!
    e^{\mathrm{j}\theta_i} + \left(1-\epsilon_{\mathrm{TA},i}\right)\epsilon_ie^{\mathrm{j}\theta_i} \approx
    e^{\mathrm{j}\theta_i} + \epsilon_ie^{\mathrm{j}\theta_i}
\end{eqnarray}
where $\epsilon_i \triangleq \epsilon_{\mathrm{RA},i}-\epsilon_{\mathrm{TA},i}$ and $\theta_i \triangleq \theta_{\mathrm{RA},i}-\theta_{\mathrm{TA},i}$.
%, with $\epsilon_i\sim\mathcal{N}\left(0,2\sigma_\epsilon^2\right)$ and $\theta_i\sim\mathcal{N}\left(0,2\sigma_\theta^2\right)$.
%
The above approximations are obtained under the assumption that $\sigma_\epsilon\ll 1$, where the first approximation uses the first-order expansion $1/\left(1+\epsilon_{\mathrm{TA},i}\right)\approx 1-\epsilon_{\mathrm{TA},i}$, and the second approximation neglects the second-order term $\epsilon_{\mathrm{TA},i}\epsilon_i$.
Since the second term $\epsilon_ie^{\mathrm{j}\theta_i}$ has zero mean and a total variance of $2\sigma_\epsilon^2$, its contribution is, with overwhelming probability, sufficiently small compared with that of the first term; therefore, the second term is neglected in this paper.
In this case, the prior distribution of $\bm{A}(i,i)$ can be modeled using a von Mises distribution based on an exponential-family approximation.

\begin{Lemma}
When $\bm{A}(i,i)$ lies on the complex unit circle, the \ac{MMSE} estimate of $\bm{A}(i,i)$ can be computed using the von Mises denoiser as follows:
\begin{equation}
    \hat{\bm{A}}(i,i)
    =
    \eta\left(\bar{\bm{A}}(i,i); \bar{v}_{\mathrm{A},i}\right) 
    =
    \frac{I_1(|\zeta_{\mathrm{A}}|)}{I_0(|\zeta_{\mathrm{A}}|)}
e^{\mathrm{j}\cdot\mathrm{arg}\left(\zeta_{\mathrm{A}}\right)},
\end{equation}
where $I_n(\cdot)$ denotes the modified Bessel function of the first kind of order $n$ and
\begin{eqnarray}
    \zeta_{\mathrm{A}} = \frac{2\bar{\bm{A}}(i,i)}{\bar{v}_{\mathrm{A},i}}.
\end{eqnarray}
The corresponding posterior \ac{MSE} is given by
\begin{equation}
\hat{v}_{\mathrm{A},i}
=
\bar{v}_{\mathrm{A},i}\cdot \frac{\partial \eta(\bar{\bm{A}}(i,i); \bar{v}_{\mathrm{A},i})}{\partial \bar{\bm{A}}(i,i)}
=
1 - \left(\frac{I_1(|\zeta_{\mathrm{A}}|)}{I_0(|\zeta_{\mathrm{A}}|)}\right)^{2}.
\end{equation}
\end{Lemma}
\begin{IEEEproof}
    The result follows by evaluating the von Mises denoiser in Appendix A at $\beta=0$ and $r=1$.
\end{IEEEproof}

In practice, the (unwrapped) phase error variance is generally difficult to estimate accurately, and consequently the corresponding concentration parameter $\beta$ of the von Mises distribution cannot be reliably determined.
Therefore, in this paper, we adopt a non-informative prior on the phase in the denoising process by setting $\beta=0$, which corresponds to a circularly uniform denoiser.

This completes the \ac{MMSE} update of $\hat{\bm{A}}$.

% Since the diagonal elements of $\bm{A}$ are distributed on the unit circle in the complex plane, their \ac{MMSE} estimates and the corresponding \acp{MSE} can be obtained using the von Mises denoiser derived in Appendix C at $\beta=0$ and $r = 1$, \textit{i.e.},
% %
% \begin{subequations}
% \begin{eqnarray}
% \hat{\bm{A}}(i,i)
% \!\!&\!\!=\!\!&\!\!
% \eta\left(\bar{\bm{A}}(i,i); \bar{v}_{\mathrm{A},i} \right),\\ 
% %     =
% %     |\alpha|\frac{I_1(|\zeta_{\mathrm{z}}|)}{I_0(|\zeta_{\mathrm{z}}|)}
% % e^{\mathrm{j}\cdot\mathrm{arg}\left(\zeta_{\mathrm{z}}\right)},
% \hat{v}_{\mathrm{A},i}
% \!\!&\!\!=\!\!&\!\!
% \bar{v}_{\mathrm{A},i}\cdot \frac{\partial \eta(\bar{\bm{A}}(i,i); \bar{v}_{\mathrm{A},i})}{\partial \bar{\bm{A}}(i,i)}.
% \end{eqnarray}
% \end{subequations}

The estimation of $\bm{B}(j,j)$ is performed in exactly the same manner as that of $\bm{A}(i,i)$.
First, to estimate the $(j,j)$-th entry of $\bm{B}$, we focus on the $j$-th column of the observation in \eqref{eq:R3}, which yields the following observation equation:
\begin{equation}
\label{eq:R3j}
\left(\bm R_3\right)_j = \left(\bm A \bm H^{\mathsf{T}}\right)_j \bm B(j,j) + \left(\bm W_3\right)_j \in \mathbb{C}^{M_{\mathsf{A}} \times 1}.
\end{equation}
Substituting $\bm{H} = \hat{\bm{H}} + \tilde{\bm{H}}$ and $\bm{A} = \hat{\bm{A}} + \tilde{\bm{A}}$ into \eqref{eq:R3j} yields
\begin{eqnarray}
\label{eq:R3j_v2}
\!\!\!\!&\!\!\!\!\!\!\!\!&\!\!\!\!
\left(\bm R_3\right)_j
=
\left(\hat{\bm{A}}\hat{\bm{H}}^{\mathsf{T}}\right)_j \bm{B}(j,j)  \nonumber \\
\!\!\!\!&\!\!\!\!\!\!\!\!&\!\!\!\!
\qquad\qquad+
\left(\tilde{\bm{A}}\hat{\bm{H}}^{\mathsf{T}}\right)_j \bm{B}(j,j)
+
\left(\hat{\bm{A}}\tilde{\bm{H}}^{\mathsf{T}}\right)_j \bm{B}(j,j) \nonumber \\
\!\!\!\!&\!\!\!\!\!\!\!\!&\!\!\!\!
\qquad\qquad+
\left(\tilde{\bm{A}}\tilde{\bm{H}}^{\mathsf{T}}\right)_j \bm{B}(j,j)
+
\left(\bm W_3\right)_j.
\end{eqnarray}
In accordance with the \ac{CLT}, the terms in the second and third lines of \eqref{eq:R3j_v2} are collectively approximated by a multivariate complex Gaussian random vector.
Accordingly, the conditional \ac{PDF} of \eqref{eq:R3j_v2}, given $\bm{B}(j,j)$, can be expressed as
\begin{eqnarray}
\label{eq:pdfB}
\!\!\!\!&\!\!\!\!\!\!\!\!&\!\!\!\!
p_{\left(\bm{\mathsf{R}}_3\right)_j|\bm{\mathsf{B}}(j,j)}
\left(
\left(\bm{R}_3\right)_j|\bm{B}(j,j)
\right) \nonumber \\
\!\!\!\!&\!\!\!\!\!\!\!\!&\!\!\!\!
\quad
\propto
\exp\Bigl[
-\left(
\left(\bm{R}_3\right)_j-\left(\hat{\bm{A}}\hat{\bm{H}}^{\mathsf{T}}\right)_j \bm{B}(j,j)
\right)^\mathsf{H}
\bm{V}_{\mathrm{B},j}^{-1}
% \mathrm{Diag}\left(
% \bm{v}_{\mathrm{B},j}
% \right)^{-1}
\nonumber \\
\!\!\!\!&\!\!\!\!\!\!\!\!&\!\!\!\!
\qquad\qquad\qquad
\times \left(
\left(\bm{R}_3\right)_j-\left(\hat{\bm{A}}\hat{\bm{H}}^{\mathsf{T}}\right)_j \bm{B}(j,j)
\right)
\Bigr],
\end{eqnarray}
where
\begin{eqnarray}
\label{eq:VBj}
\!\!\!\!&\!\!\!\!\!\!\!\!&\!\!\!\!
\bm{V}_{\mathrm{B},j}
\triangleq
\mathbb{E}\Bigl[
\left(
\left(\bm{R}_3\right)_j-\left(\hat{\bm{A}}\hat{\bm{H}}^{\mathsf{T}}\right)_j \bm{B}(j,j)
\right) \nonumber \\
\!\!\!\!&\!\!\!\!\!\!\!\!&\!\!\!\!
\qquad\qquad\qquad
\times
\left(
\left(\bm{R}_3\right)_j-\left(\hat{\bm{A}}\hat{\bm{H}}^{\mathsf{T}}\right)_j \bm{B}(j,j)
\right)^\mathsf{H}
\Bigr] \nonumber \\
\!\!\!\!&\!\!\!\!\!\!\!\!&\!\!\!\!
\quad=
\bm{\varOmega}_\mathrm{A}+\hat{\bm{A}}\bm{\varPsi}_\mathrm{B}\hat{\bm{A}}^\mathsf{H} +
\mathrm{diag}\left(v_{\mathrm{B},j1},\dots,v_{\mathrm{B},jM_\mathrm{A}}\right),
\end{eqnarray}
%$\bm v_{\mathrm{A},i} \triangleq \left[v_{\mathrm{A},i1},\dots,v_{\mathrm{A},iM_\mathrm{B}}\right]^{\mathsf{T}}\in\mathbb{R}^{M_\mathrm{B}\times 1}$ with
with
\begin{equation}
v_{\mathrm{B},ji} \triangleq 
\left(
|\hat{\bm H}(j, i)|^2 + \bm{\varPsi}_\mathrm{B}(i,i)
\right)
\hat{v}_{\mathrm{A},i}.
\end{equation}

% where $\bm v_{\mathrm{B},j} \triangleq \left[v_{\mathrm{B},j1},\dots,v_{\mathrm{B},jM_\mathrm{A}}\right]^{\mathsf{T}}\in\mathbb{R}^{M_\mathrm{A}\times 1}$ with
% %
% \begin{eqnarray}
% \label{eq:vBj}
% \!\!\!\!&\!\!\!\!\!\!\!\!&\!\!\!\!
% v_{\mathrm{B},ji}
% \triangleq
% \mathbb{E}\left[
% |\left(\bm{R}_3\right)_j-\left(\hat{\bm{A}}\hat{\bm{H}}^{\mathsf{T}}\right)_j \bm{B}(j,j)|^2
% \right] \nonumber \\
% \!\!\!\!&\!\!\!\!\!\!\!\!&\!\!\!\!
% \quad=
% \frac{\sigma^2}{2}
% +
% \hat{v}_{\mathrm{A},i}
% |\hat{\bm H}(j, i)|^2
% +
% \hat{v}_{\mathrm{H}}
% |\hat{\bm A}(i,i)|^2
% +
% \hat{v}_{\mathrm{A},i}
% \hat{v}_{\mathrm{H}}.
% \end{eqnarray}

Next, by expanding and rearranging the exponent in \eqref{eq:pdfB} and completing the square \ac{w.r.t.} $\bm{B}(j,j)$, we have
\begin{eqnarray}
\label{eq:likB}
\!\!\!\!&\!\!\!\!\!\!\!\!&\!\!\!\!
p_{\left(\bm{\mathsf{R}}_3\right)_j|\bm{\mathsf{B}}(j,j)}
\left(
\left(\bm{R}_3\right)_j|\bm{B}(j,j)
\right) \nonumber \\
\!\!\!\!&\!\!\!\!\!\!\!\!&\!\!\!\!
\qquad
\propto
\exp\left[-\frac{|\bar{\bm{B}}(j,j)-\bm{B}(j,j)|^2}{\bar{v}_{\mathrm{B},j}}\right],
\end{eqnarray}
with
\begin{subequations}
\begin{eqnarray}
\bar{\bm{B}}(j,j)
\!\!&\!\!=\!\!&\!\!
\frac{1}{\psi_{\mathrm{B},j}}
\left(\hat{\bm A}\hat{\bm H}^{\mathsf{T}}\right)_j^{\mathsf{H}}
\bm{V}_{\mathrm{B},j}^{-1} \left(\bm R_3\right)_j, \\
\bar{v}_{\mathrm{B},j}
\!\!&\!\!=\!\!&\!\!
\frac{1}{\psi_{\mathrm{B},j}},
\end{eqnarray}
\label{eq:barB}%
\end{subequations}
where
\begin{equation}
\psi_{\mathrm{B},j} \triangleq \left(\hat{\bm A}\hat{\bm H}^{\mathsf{T}}\right)_j^{\mathsf{H}} \bm{V}_{\mathrm{B},j}^{-1} \left(\hat{\bm A}\hat{\bm H}^{\mathsf{T}}\right)_j.
\end{equation}

Finally, the prior distribution of $\bm{B}(j,j)$ can also be modeled using a von Mises distribution for all $j\in\mathcal{B}$, following exactly the same procedure as in \eqref{eq:Aii}.
By applying the von Mises denoiser with $\beta=0$ and $r = 1$, the corresponding \ac{MMSE} estimates and their associated \acp{MSE} can be obtained, \textit{i.e.},
\begin{subequations}
\begin{eqnarray}
\!\!\!\!&\!\!\!\!\!\!\!\!&\!\!\!\!\!\!\!\!
\hat{\bm{B}}(j,j)
=
\eta\left(\bar{\bm{B}}(j,j); \bar{v}_{\mathrm{B},j} \right)
=
\frac{I_1(|\zeta_{\mathrm{B}}|)}{I_0(|\zeta_{\mathrm{B}}|)}
e^{\mathrm{j}\cdot\mathrm{arg}\left(\zeta_{\mathrm{B}}\right)},\\ 
\!\!\!\!&\!\!\!\!\!\!\!\!&\!\!\!\!\!\!\!\!
\hat{v}_{\mathrm{B},j}
=
\bar{v}_{\mathrm{B},j}\cdot \frac{\partial \eta\left(\bar{\bm{B}}(j,j); \bar{v}_{\mathrm{B},j}\right)}{\partial \bar{\bm{B}}(j,j)}
=
1 - \left(\frac{I_1(|\zeta_{\mathrm{B}}|)}{I_0(|\zeta_{\mathrm{B}}|)}\right)^{2}\!\!\!,
\end{eqnarray}
\end{subequations}
where
\begin{equation}
    \zeta_\mathrm{B} = \frac{2\bar{\bm{B}}(j,j)}{\bar{v}_{\mathrm{B},j}}.
\end{equation}

\subsubsection{Estimation of $\gamma$}

All estimates obtained in the preceding steps are substituted into \eqref{eq:R4}, which yields
\begin{eqnarray}
\label{eq:R4v2}
\bm R_4
\!\!&\!\!=\!\!&\!\!
\gamma
\left(\hat{\bm A} + \tilde{\bm A}\right)
\hat{\bm Z}^{\mathsf{T}}
\left(\hat{\bm B} + \tilde{\bm B}\right) + \bm W_4 \nonumber\\
\!\!&\!\!=\!\!&\!\!
\gamma \underbrace{
\hat{\bm A} \hat{\bm Z}^{\mathsf{T}} \hat{\bm B}
}_{\triangleq \hat{\bm{D}}}
+
\gamma
\underbrace{
\left(
\tilde{\bm A} \hat{\bm Z}^{\mathsf{T}} \hat{\bm B} 
%+ \hat{\bm A} \tilde{\bm Z}^{\mathsf{T}} \hat{\bm B} 
+ \hat{\bm A} \hat{\bm Z}^{\mathsf{T}} \tilde{\bm B} 
%+ \tilde{\bm A} \tilde{\bm Z}^{\mathsf{T}} \hat{\bm B} \nonumber\\
+
\tilde{\bm A} \hat{\bm Z}^{\mathsf{T}} \tilde{\bm B} 
%+ \hat{\bm A} \tilde{\bm Z}^{\mathsf{T}} \tilde{\bm B} 
%+ \tilde{\bm A} \tilde{\bm Z}^{\mathsf{T}} \tilde{\bm B}
\right)
}_{\triangleq \tilde{\bm{D}}}
+ \bm W_4 \nonumber \\
\!\!&\!\!=\!\!&\!\!
\gamma\hat{\bm{D}} + \gamma\tilde{\bm{D}} + \bm W_4,
\end{eqnarray}
where we assume $\hat{\bm{Z}}\approx \bm{Z}$.
% \begin{eqnarray}
% \label{eq:R4v2}
% \bm R_4
% \!\!&\!\!=\!\!&\!\!
% \gamma
% \left(\hat{\bm A} + \tilde{\bm A}\right)
% \left(\hat{\bm Z} + \tilde{\bm Z}\right)^{\mathsf{T}}
% \left(\hat{\bm B} + \tilde{\bm B}\right) + \bm W_4 \nonumber\\
% \!\!&\!\!=\!\!&\!\!
% \gamma\hat{\bm A} \hat{\bm Z}^{\mathsf{T}} \hat{\bm B} \nonumber \\
% \!\!&\!\!\quad\!\!&\!\!
% +
% \gamma \Bigl( \tilde{\bm A} \hat{\bm Z}^{\mathsf{T}} \hat{\bm B} 
% + \hat{\bm A} \tilde{\bm Z}^{\mathsf{T}} \hat{\bm B} 
% + \hat{\bm A} \hat{\bm Z}^{\mathsf{T}} \tilde{\bm B}
% + \tilde{\bm A} \tilde{\bm Z}^{\mathsf{T}} \hat{\bm B} \nonumber\\
% \!\!&\!\!\quad\!\!&\!\!
% +
% \tilde{\bm A} \hat{\bm Z}^{\mathsf{T}} \tilde{\bm B} 
% + \hat{\bm A} \tilde{\bm Z}^{\mathsf{T}} \tilde{\bm B} 
% + \tilde{\bm A} \tilde{\bm Z}^{\mathsf{T}} \tilde{\bm B}\Bigr)
% + \bm W_4 \nonumber \\
% \!\!&\!\!=\!\!&\!\!
% \gamma\hat{\bm{D}} + \gamma\tilde{\bm{D}} + \bm W_4,
% \end{eqnarray}
%
% where we define
% %
% \begin{eqnarray}
% \hat{\bm{D}}
% \!\!&\!\!\triangleq\!\!&\!\!
% \hat{\bm A} \hat{\bm Z}^{\mathsf{T}} \hat{\bm B}, \nonumber \\
% \tilde{\bm{D}}
% \!\!&\!\!\triangleq\!\!&\!\!
% \tilde{\bm A} \hat{\bm Z}^{\mathsf{T}} \hat{\bm B} 
% + \hat{\bm A} \tilde{\bm Z}^{\mathsf{T}} \hat{\bm B} 
% + \hat{\bm A} \hat{\bm Z}^{\mathsf{T}} \tilde{\bm B}
% + \tilde{\bm A} \tilde{\bm Z}^{\mathsf{T}} \hat{\bm B} \nonumber\\
% \!\!&\!\!\quad\!\!&\!\!
% +
% \tilde{\bm A} \hat{\bm Z}^{\mathsf{T}} \tilde{\bm B} 
% + \hat{\bm A} \tilde{\bm Z}^{\mathsf{T}} \tilde{\bm B} 
% + \tilde{\bm A} \tilde{\bm Z}^{\mathsf{T}} \tilde{\bm B}.
% \end{eqnarray}
%
For \ac{MMSE} estimation of $\gamma$, \eqref{eq:R4v2} is rewritten in the following vectorized form:
\begin{equation}
\label{eq:R4vec}
    \mathrm{vec}\left(\bm{R}_4\right) = \underbrace{\mathrm{vec}\left(\hat{\bm{D}}\right)}_{\triangleq \hat{\bm{d}}}\gamma +
    \underbrace{\mathrm{vec}\left(\tilde{\bm{D}}\right)}_{\triangleq \tilde{\bm{d}}}\gamma +
    \underbrace{\mathrm{vec}\left(\bm W_4\right)}_{\triangleq \bm{w}_4}.
\end{equation}
In accordance with the \ac{CLT}, the second and third terms in \eqref{eq:R4vec} are collectively approximated by a multivariate complex Gaussian random vector.
Accordingly, the conditional \ac{PDF} of \eqref{eq:R4vec}, given $\bm{R}_4$, can be expressed as
\begin{eqnarray}
\label{eq:pdfgamma}
\!\!\!\!&\!\!\!\!\!\!\!\!&\!\!\!\!
p_{\mathrm{vec}\left(\bm{\mathsf{R}}_4\right)|\mathrm{\gamma}}
\left(
\mathrm{vec}\left(\bm{R}_4\right)|\gamma
\right) \\
\!\!\!\!&\!\!\!\!\!\!\!\!&\!\!\!\!
\quad
\propto
\exp\Bigl[
-\left(
\mathrm{vec}\left(\bm{R}_4\right)-\hat{\bm{d}}\gamma
\right)^\mathsf{H}
% \mathrm{Diag}\left(
% \bm{v}_{\gamma}
% \right)^{-1}
\bm{V}_\mathrm{\gamma}^{-1}
% \nonumber \\
% \!\!\!\!&\!\!\!\!\!\!\!\!&\!\!\!\!
% \qquad\qquad\qquad
% \times 
\left(
\mathrm{vec}\left(\bm{R}_4\right)-\hat{\bm{d}}\gamma
\right)
\Bigr],\nonumber 
\end{eqnarray}
where
\begin{eqnarray}
\label{eq:Vg}
\!\!\!\!&\!\!\!\!\!\!\!\!&\!\!\!\!
\bm{V}_\mathrm{\gamma}
\triangleq
\mathbb{E}\Bigl[
\left(
\mathrm{vec}\left(\bm{R}_4\right)-\hat{\bm{d}}\gamma
\right) 
\left(
\mathrm{vec}\left(\bm{R}_4\right)-\hat{\bm{d}}\gamma
\right)^\mathsf{H}
\Bigr] \nonumber \\
\!\!\!\!&\!\!\!\!\!\!\!\!&\!\!\!\!
\quad=
\bm{\varSigma}_\mathrm{A} +
\bm{V}_\mathrm{\gamma}^\mathrm{D},
\end{eqnarray}
%$\bm v_{\mathrm{A},i} \triangleq \left[v_{\mathrm{A},i1},\dots,v_{\mathrm{A},iM_\mathrm{B}}\right]^{\mathsf{T}}\in\mathbb{R}^{M_\mathrm{B}\times 1}$ with
with $\bm{\varSigma}_\mathrm{A}\triangleq \bm{\varPsi}_\mathrm{A}\otimes\bm{\varOmega}_\mathrm{A}$, $\bm{V}_\mathrm{\gamma}^\mathrm{D}\triangleq\mathrm{diag}\left(v_{\gamma,11},\dots,v_{\gamma,M_{\mathrm{A}}M_{\mathrm{B}}}\right)$, 
\begin{eqnarray}
\label{eq:vgamma}
\!\!\!\!&\!\!\!\!\!\!\!\!&\!\!\!\!
v_{\gamma,ij} \triangleq 
\phi_\gamma|\hat{\bm Z}(j,i)|^2 \nonumber \\
\!\!\!\!&\!\!\!\!\!\!\!\!&\!\!\!\!
\qquad
\times
\left(
\hat{v}_{\mathrm{A},i} |\hat{\bm B}(j,j)|^2
+ |\hat{\bm A}(i,i)|^2\hat{v}_{\mathrm{B},j} 
+ \hat{v}_{\mathrm{A},i}\hat{v}_{\mathrm{B},j} 
\right),
\end{eqnarray}
% where $\bm v_{\gamma} \triangleq \left[v_{\gamma,11},\dots,v_{\gamma,M_{\mathrm{A}}M_{\mathrm{B}}}\right]^{\mathrm{T}}\in\mathbb{R}^{M_\mathrm{A}M_\mathrm{A}\times 1}$ with
%
% \begin{eqnarray}
% \label{eq:vGij}
% \!\!\!\!&\!\!\!\!\!\!\!\!&\!\!\!\!
% v_{\gamma,ij}
% \triangleq
% \mathbb{E}\left[
% |\bm{R}_4(i,j)-\hat{\bm{D}}(i,j)\gamma|^2
% \right] \nonumber \\
% \!\!\!\!&\!\!\!\!\!\!\!\!&\!\!\!\!
% \quad
% =\frac{\sigma^2}{2} + 
% \Bigl( v_{\mathrm{A},i}|\hat{\bm Z}(j,i)|^2 |\hat{\bm B}(j,j)|^2
% + |\hat{\bm A}(i,i)|^2|\hat{\bm Z}(j,i)|^2 v_{\mathrm{B},j} 
% \nonumber\\
% \!\!\!\!&\!\!\!\!\!\!\!\!&\!\!\!\!
% \quad
% +|\hat{\bm A}(i,i)|^2|\hat{\bm B}(j,j)|^2 v_{\mathrm{Z}}(j,i) 
% + v_{\mathrm{A},i}v_{\mathrm{B},j}|\hat{\bm Z}(j,i)|^2 
% \nonumber\\
% \!\!\!\!&\!\!\!\!\!\!\!\!&\!\!\!\!
% \quad
% + v_{\mathrm{A},i}v_{\mathrm{Z}}(j,i)|\hat{\bm B}(j,j)|^2 
% + v_{\mathrm{B},j}v_{\mathrm{Z}}(j,i)|\hat{\bm A}(i,i)|^2 \nonumber\\
% \!\!\!\!&\!\!\!\!\!\!\!\!&\!\!\!\!
% \quad
% + v_{\mathrm{A},i}v_{\mathrm{B},j}v_{\mathrm{Z}}(j,i)\Bigr)\times \phi_\gamma,
% \end{eqnarray}
%
and $\phi_\gamma \triangleq \mathbb{E}\left[|\gamma|^2\right]$.

By expanding and rearranging the exponent in \eqref{eq:pdfgamma} and completing the square \ac{w.r.t.} $\gamma$, we have
\begin{equation}
\label{eq:likgamma}
p_{\mathrm{vec}\left(\bm{\mathsf{R}}_4\right)|\mathrm{\gamma}}
\left(
\mathrm{vec}\left(\bm{R}_4\right)|\gamma
\right) 
\propto
\exp\left[-\frac{|\bar{\gamma}-\gamma|^2}{\bar{v}_{\gamma}}\right],
\end{equation}
with
\begin{equation}
\bar{\gamma}
=
\frac{\hat{\bm{d}}^{\mathsf{H}}
\bm{V}_\mathrm{\gamma}^{-1}
%\mathrm{Diag}\left(\bm{v}_{\gamma}\right)^{-1} 
\mathrm{vec}\left(\bm{R}_4\right)}{\psi_{\gamma}},\quad
\bar{v}_{\gamma}
=
\frac{1}{\psi_{\gamma}},
\label{eq:bargamma}%
\end{equation}
where
\begin{equation}
\psi_{\gamma} \triangleq \hat{\bm{d}}^{\mathsf{H}} \bm{V}_\mathrm{\gamma}^{-1} \hat{\bm{d}}.
\end{equation}
%
% From \eqref{eq:gamma}, $\gamma$ represents the complex gain of the forward and reverse paths, for which specifying a particular prior distribution is difficult.
% %
% In such cases, it is common to adopt a complex Gaussian distribution as the prior.
% %
% Assuming that $\phi_\gamma \triangleq \mathbb{E}\left[|\gamma|^2\right]$ is known, the \ac{MMSE} estimate of $\gamma$ and its corresponding \ac{MSE}, derived using \eqref{eq:likgamma} as the likelihood function, are given by \tktk{The performance can be further improved by using the von Mises denoiser here. We plan to revise this accordingly.}
% \egl{OK}
% %
% \begin{eqnarray}
%     \hat{\gamma} = \frac{\phi_\gamma}{\phi_\gamma+\bar{v}_\gamma}\bar{\gamma}, \quad
%     \hat{v}_\gamma = \frac{\phi_\gamma\bar{v}_\gamma}{\phi_\gamma+\bar{v}_\gamma}.
% \end{eqnarray}
Since $\phi_\gamma$ used in \eqref{eq:vgamma} is generally unavailable, its treatment requires further discussion.
Ideally, the true long-term statistic $\phi_\gamma \triangleq \mathbb{E}\left[|\gamma|^2\right]$ is available during system setup, in which case it can be directly used.
When $\phi_\gamma$ is unknown, several alternatives can be considered.
The simplest approach is to operate the system with $\phi_\gamma = 1$, noting that the ultimate goal of calibration is to achieve $\alpha = \beta$.
This corresponds to determining the reliability of the estimate solely based on the \ac{SNR}.
In this paper, as a more adaptive alternative, we approximate the true long-term statistic $\phi_\gamma$ by an instantaneous estimate of $|\gamma|^2$, obtained via the \ac{MoM} described below.

% \begin{Lemma}
% From \eqref{eq:R2}, the elements of $\bm{R}_2$ are statistically independent, and hence a consistent estimate of $|\alpha|$ can be obtained via the \ac{MoM} as
% %
% \begin{equation}
% \label{eq:alphahat}
%     \hat{|\alpha|}
%     = \sqrt{
%         \frac{1}{M_\mathrm{A}M_\mathrm{B}}\sum_{i\in\mathcal{A}}\sum_{j\in\mathcal{B}} |\bm{R}_2(j,i)|^{2} - \frac{\sigma^2}{2}
%       }.
% \end{equation}
% \end{Lemma}

% \begin{IEEEproof}
%     The result follows directly by evaluating the \ac{MoM} estimator derived in Appendix C at $r=|\alpha|$ and $v=\sigma^2/2$.
% \end{IEEEproof}

\begin{Lemma}
A consistent estimate of $|\gamma|^2$ based on the observation in \eqref{eq:R4} can be obtained via the \acf{MoM} as follows:
\begin{equation}
\label{eq:mom}
    |\check{\gamma}|^2 =\frac{|q|^2-u}{u^2+s},
\end{equation}
where
\begin{subequations}
\begin{eqnarray}
    \label{eq:q}
    q
    \!\!&\!\!\triangleq\!\!&\!\!
    \hat{\bm{d}}^\mathsf{H}
    \bm{\varSigma}_\mathrm{A}^{-1}
    \mathrm{vec}\left(\bm{R}_4\right), \\
    u
    \!\!&\!\!\triangleq\!\!&\!\!
    \hat{\bm{d}}^\mathsf{H}
    \bm{\varSigma}_\mathrm{A}^{-1}
    \hat{\bm{d}}, \\
    s
    \!\!&\!\!\triangleq\!\!&\!\!
    \hat{\bm{d}}^\mathsf{H}
    \bm{\varSigma}_\mathrm{A}^{-1}
    \bm{V}_\mathrm{\gamma}^\mathrm{D}
    \bm{\varSigma}_\mathrm{A}^{-1}
    \hat{\bm{d}}.
\end{eqnarray}
\end{subequations}

% \begin{equation}
%     \acute{\bm R}_4(i,j)
%     = \frac{\bm R_4(i,j)}{\bigl|\hat{\bm Z}(j,i)\bigr|},
%     \quad i\in\mathcal{A},\; j\in\mathcal{B},
% \end{equation}
% % 
% \begin{equation}
% \label{eq:abs2_gamma}
%     |\hat{\gamma}|
%     =\sqrt{
%         \frac{1}{M_\mathrm{A}M_\mathrm{B}}
%         \sum_{i\in\mathcal{A}}\sum_{j\in\mathcal{B}}
%         \left(
%             \bigl|\acute{\bm{R}}_4(i,j)\bigr|^{2}
%             -
%             \frac{\sigma^2}{2\bigl|\hat{\bm Z}(j,i)\bigr|^2}
%         \right)
%         }.
% \end{equation}
\end{Lemma}

\begin{IEEEproof}
From the definition in \eqref{eq:q} and the vectorized observation model in \eqref{eq:R4vec}, we have
\begin{equation}
\label{eq:q2}
    q=
    \underbrace{\hat{\bm{d}}^\mathsf{H}\bm{\varSigma}_\mathrm{A}^{-1}\hat{\bm{d}}}_{=u}\gamma
    + \underbrace{\hat{\bm{d}}^\mathsf{H}\bm{\varSigma}_\mathrm{A}^{-1}\tilde{\bm{d}}}_{\triangleq \kappa_1}\gamma
    + \underbrace{\hat{\bm{d}}^\mathsf{H}\bm{\varSigma}_\mathrm{A}^{-1}\bm{w}_4}_{\triangleq\kappa_2}.
\end{equation}
Since $\mathbb{E}[\tilde{\bm{d}}]=\mathbb{E}\left[\bm{w}_4\right]=\bm{0}$, it follows that $\mathbb{E}\left[\kappa_1\right]=\mathbb{E}\left[\kappa_2\right]=0$.
Moreover, 
\begin{subequations}
\begin{eqnarray}
    \mathbb{E}\left[|\kappa_1|^2\right]
    \!\!&\!\!=\!\!&\!\!
    \hat{\bm{d}}^\mathsf{H}
    \bm{\varSigma}_\mathrm{A}^{-1}
    \bm{V}_\mathrm{\gamma}^\mathrm{D}
    \bm{\varSigma}_\mathrm{A}^{-1}
    \hat{\bm{d}} = s, \\
    \mathbb{E}\left[|\kappa_2|^2\right]
    \!\!&\!\!=\!\!&\!\!
    \hat{\bm{d}}^\mathsf{H}
    \bm{\varSigma}_\mathrm{A}^{-1}
    \hat{\bm{d}} = u.
\end{eqnarray}
\end{subequations}
By squaring both sides of \eqref{eq:q2} and taking the expectation, we obtain
\begin{equation}
\mathbb{E}\left[|q|^2\right]
=
\left(u^2+s\right)|\gamma|^2+u,
% +2\Re\left\{u\gamma \mathbb{E}\left[(\gamma\kappa_1+\kappa_2)^*\right]\right\}+2\Re\left\{|\gamma|^2\mathbb{E}\left[\kappa_1\right]\right\}
\end{equation}
where we use the independence of $\tilde{\bm{d}}$ and $\bm{w}_4$.
Equating the theoretical moment $\mathbb{E}\left[|q|^2\right]$ with its single-sample counterpart $|q|^2$ and solving for $|\gamma|^2$ yields \eqref{eq:mom}.
\end{IEEEproof}

Finally, since an estimate of $|\gamma|$ is obtained by taking the square root of \eqref{eq:mom}, we can model the prior distribution of $\gamma$ as a circularly uniform distribution on the complex circle with radius $|\check{\gamma}|$.
By applying the von Mises denoiser with $\beta=0$ and $r=|\check{\gamma}|$, the corresponding \ac{MMSE} estimate and its associated \ac{MSE} can be obtained, \textit{i.e.},
\begin{subequations}
\begin{eqnarray}
\!\!\!\!&\!\!\!\!\!\!\!\!&\!\!\!\!\!\!\!\!
\hat{\gamma}
=
\eta\left(\bar{\gamma}; \bar{v}_{\gamma} \right)
=
|\check{\gamma}|
\frac{I_1(|\zeta_{\mathrm{\gamma}}|)}{I_0(|\zeta_{\mathrm{\gamma}}|)}
e^{\mathrm{j}\cdot\mathrm{arg}\left(\zeta_{\mathrm{\gamma}}\right)},\\ 
\!\!\!\!&\!\!\!\!\!\!\!\!&\!\!\!\!\!\!\!\!
\hat{v}_{\gamma}
=
\bar{v}_{\mathrm{\gamma}}\cdot \frac{\partial \eta\left(\bar{\gamma}; \bar{v}_{\mathrm{\gamma}}\right)}{\partial \bar{\gamma}}
=
|\check{\gamma}|^2\left(1 - \left(\frac{I_1(|\zeta_{\mathrm{\gamma}}|)}{I_0(|\zeta_{\mathrm{\gamma}}|)}\right)\right)^{2}\!\!\!,
\end{eqnarray}
\end{subequations}
where
\begin{equation}
    \zeta_\mathrm{\gamma} = \frac{2|\check{\gamma}|}{\bar{v}_{\mathrm{\gamma}}}\bar{\gamma}.
\end{equation}

\begin{algorithm}[t]
\caption{:\ MMSE-based calibration algorithm}
\label{alg:mmse}
\begin{algorithmic}[1]
\Require$\bm R_1,\bm R_2,\bm R_3,\bm R_4,(\bm{\varOmega}_\mathrm{A},\bm{\varOmega}_\mathrm{B},\bm{\varPsi}_\mathrm{A},\bm{\varPsi}_\mathrm{B}$ if available)
\Ensure $\hat{\bm H},\hat{\bm Z},\hat{\bm A},\hat{\bm B},\hat{\gamma}$

\vspace{0.5mm}
\LineComment{// 1) Estimation of $\bm H$} 
\State $\hat{\bm H} \gets \bm{R}_1$

\vspace{0.5mm}
\LineComment{// 2) Estimation of $\bm Z$} 
\State $\hat{\bm Z} \gets \mathcal{S}\{\bm{R}_2\}$ 

% \vspace{0.5mm}
% \LineComment{// 2) Estimation of $\bm Z$} 
% \State $\hat{|\alpha|}
% = \left(
% \frac{1}{M_\mathrm{A}M_\mathrm{B}}\sum_{i\in\mathcal{A}}\sum_{j\in\mathcal{B}} |\bm{R}_2(j,i)|^{2} - \frac{\sigma^2}{2}\right)^{1/2}$

% \State $\forall i\in\mathcal{A},\forall j\in\mathcal{B}: \zeta_{\mathrm{z}} = \frac{4\hat{|\alpha|}}{\sigma^2}\bm{R}_2(j,i)$

% \State $\forall i\in\mathcal{A},\forall j\in\mathcal{B}: \hat{\bm{Z}}(j,i)
% \gets
% \hat{|\alpha|}\frac{I_1(|\zeta_{\mathrm{z}}|)}{I_0(|\zeta_{\mathrm{z}}|)}
% e^{\mathrm{j}\cdot\mathrm{arg}\left(\zeta_{\mathrm{z}}\right)}$
      
% \State $\forall i\in\mathcal{A},\forall j\in\mathcal{B}: \hat{v}_{\mathrm{z},ij}
% \gets
%     \hat{|\alpha|}^2\left(1 - \left(\frac{I_1(|\zeta_{\mathrm{z}}|)}{I_0(|\zeta_{\mathrm{z}}|)}\right)^{2}\right)$

\vspace{0.5mm}
\LineComment{// 3) Estimation of $\bm{A}$ and $\bm{B}$} 

\vspace{0.5mm}
\State $\hat{\bm A} \gets \bm I_{M_\mathrm{A}},\ \hat{\bm B} \gets \bm I_{M_{\mathrm{B}}}$  \Comment{Initialization}
\State $\forall i\in\mathcal{A},\forall j\in\mathcal{B}: \hat{v}_{\mathrm{A},i}\gets1, \hat{v}_{\mathrm{B},j}\gets1$ 
\Repeat
\State $\forall i\in\mathcal{A}:$ Compute $\bm{V}_{\mathrm{A},i}$  using \eqref{eq:VAi}.

\State $\forall i\in\mathcal{A}: \psi_{\mathrm{A},i} \gets \left(\hat{\bm B}\hat{\bm H}\right)_i^{\mathsf{H}} \bm{V}_{\mathrm{A},i}^{-1} \left(\hat{\bm B}\hat{\bm H}\right)_i$

\State $\forall i\in\mathcal{A}: \bar{\bm{A}}(i,i)
\gets
\frac{1}{\psi_{\mathrm{A},i}}
\left(\hat{\bm B}\hat{\bm H}\right)_i^{\mathsf{H}} \bm{V}_{\mathrm{A},i}^{-1} \left(\bm R_3^{\mathsf{T}}\right)_i$

\State $\forall i\in\mathcal{A}: \bar{v}_{\mathrm{A},i}
\gets
\frac{1}{\psi_{\mathrm{A},i}}$

\State $\forall i\in\mathcal{A}: \zeta_{\mathrm{A}} \gets \frac{2}{\bar{v}_{\mathrm{A},i}}\bar{\bm{A}}(i,i)$

\State $\forall i\in\mathcal{A}: \hat{\bm{A}}(i,i)
\gets
\frac{I_1(|\zeta_{\mathrm{A}}|)}{I_0(|\zeta_{\mathrm{A}}|)}
e^{\mathrm{j}\cdot\mathrm{arg}\left(\zeta_{\mathrm{A}}\right)}$
      
\State $\forall i\in\mathcal{A}: \hat{v}_{\mathrm{A},i}
\gets
    \left(1 - \left(\frac{I_1(|\zeta_{\mathrm{A}}|)}{I_0(|\zeta_{\mathrm{A}}|)}\right)^{2}\right)$

\State $\forall j\in\mathcal{B}:$ Compute $\bm{V}_{\mathrm{B},j}$ using \eqref{eq:VBj}.

\State $\forall j\in\mathcal{B}: \psi_{\mathrm{B},j} \gets \left(\hat{\bm A}\hat{\bm H}^{\mathsf{T}}\right)_j^{\mathsf{H}} \bm{V}_{\mathrm{B},j}^{-1} \left(\hat{\bm A}\hat{\bm H}^{\mathsf{T}}\right)_j$

\State $\forall j\in\mathcal{B}: \bar{\bm{B}}(j,j) \gets \frac{1}{\psi_{\mathrm{B},j}}
\left(\hat{\bm A}\hat{\bm H}^{\mathsf{T}}\right)_j^{\mathsf{H}} \bm{V}_{\mathrm{B},j}^{-1} \left(\bm R_3\right)_j$

\State $\forall j\in\mathcal{B}: \bar{v}_{\mathrm{B},j}
\gets
\frac{1}{\psi_{\mathrm{B},j}}$

\State $\forall j\in\mathcal{B}: \zeta_{\mathrm{B}} \gets \frac{2}{\bar{v}_{\mathrm{B},j}}\bar{\bm{B}}(j,j)$

\State $\forall j\in\mathcal{B}: \hat{\bm{B}}(j,j)
\gets
\frac{I_1(|\zeta_{\mathrm{B}}|)}{I_0(|\zeta_{\mathrm{B}}|)}
e^{\mathrm{j}\cdot\mathrm{arg}\left(\zeta_{\mathrm{B}}\right)}$
      
\State $\forall j\in\mathcal{B}: \hat{v}_{\mathrm{B},j}
\gets
    \left(1 - \left(\frac{I_1(|\zeta_{\mathrm{B}}|)}{I_0(|\zeta_{\mathrm{B}}|)}\right)^{2}\right)$

% \State $\hat{\bm A}(i,i) = \dfrac{I_1(|\zeta_{\bar{\mathrm{A}}}|)}{I_1(|\zeta_{\bar{\mathrm{A}}}|}e^{\mathrm{j}\cdot\arg(\zeta_{\bar{\mathrm{A}}})}$
% \State $\hat{v}_{\mathrm{A}, i} = 1 - \left(\dfrac{I_1(|\zeta_{\bar{\mathrm{A}}}|)}{I_1(|\zeta_{\bar{\mathrm{A}}}|}\right)^2$

%\State Calculate effective variances $\bm v_{\mathrm{B},j}$
% \State $\bar{\bm B}(j,j)\gets\dfrac{(\hat{\bm A}\hat{\bm H}^{\mathsf{T}})_j^{\mathsf{H}} \mathrm{Diag}(\bm v_{\mathrm{B},j})^{-1} (\bm R_3)_j}{(\hat{\bm A}\hat{\bm H}^{\mathsf{T}})_j^{\mathsf{H}} \mathrm{Diag}(\bm v_{\mathrm{B},j})^{-1} (\hat{\bm A}\hat{\bm H}^{\mathsf{T}})_j}$
% \State $\bar{v}_{\mathrm{B},j}\gets\dfrac{1}{(\hat{\bm A}\hat{\bm H}^{\mathsf{T}})_j^{\mathsf{H}} \mathrm{Diag}(\bm v_{\mathrm{B},j})^{-1} (\hat{\bm A}\hat{\bm H}^{\mathsf{T}})_j}$

% \State $\hat{\bm B}(j,j) = \dfrac{I_1(|\zeta_{\bar{\mathrm{B}}}|)}{I_1(|\zeta_{\bar{\mathrm{B}}}|}e^{\mathrm{j}\cdot\arg(\zeta_{\bar{\mathrm{B}}})}$
% \State $\hat{v}_{\mathrm{B}, i} = 1 - \left(\dfrac{I_1(|\zeta_{\bar{\mathrm{B}}}|)}{I_1(|\zeta_{\bar{\mathrm{B}}}|}\right)^2$

\Until{Num. of iterations reaches $N_{\mathrm{Iter}}$}

\vspace{0.5mm}
\LineComment{// 4) Estimation of $\gamma$} 

\State Compute $\bm{V}_{\mathrm{\gamma}}$ and $\bm{V}_{\mathrm{\gamma}}^\mathrm{D}$ using \eqref{eq:Vg}.

\State $\hat{\bm{d}} \gets \mathrm{vec}\left(\hat{\bm A} \hat{\bm Z}^{\mathsf{T}} \hat{\bm B}\right), \quad q\gets \hat{\bm{d}}^\mathsf{H}
    \bm{\varSigma}_\mathrm{A}^{-1}
    \mathrm{vec}\left(\bm{R}_4\right)$
    
\State $u
    \gets
    \hat{\bm{d}}^\mathsf{H}
    \bm{\varSigma}_\mathrm{A}^{-1}
    \hat{\bm{d}},\quad
    s
    \gets
    \hat{\bm{d}}^\mathsf{H}
    \bm{\varSigma}_\mathrm{A}^{-1}
    \bm{V}_\mathrm{\gamma}^\mathrm{D}
    \bm{\varSigma}_\mathrm{A}^{-1}
    \hat{\bm{d}}.$
\State $|\check{\gamma}|^2 \gets \frac{|q|^2-u}{u^2+s}$

% \State $|\hat{\gamma}|
%     =\sqrt{
%         \frac{1}{M_\mathrm{A}M_\mathrm{B}}
%         \sum_{i\in\mathcal{A}}\sum_{j\in\mathcal{B}}
%         \left(
%             \bigl|\acute{\bm{R}}_4(i,j)\bigr|^{2}
%             -
%             \frac{\sigma^2}{2\bigl|\hat{\bm Z}(j,i)\bigr|^2}
%         \right)
%         }$
%\State Compute $\bm v_{\gamma}$ using \eqref{eq:vGij}.

\State $\psi_{\gamma} \gets \hat{\bm{d}}^{\mathsf{H}} \bm{V}_{\mathrm{\gamma}}^{-1} \hat{\bm{d}}$

\State $\bar{\gamma}
\gets
\frac{1}{\psi_{\gamma}}\hat{\bm{d}}^{\mathsf{H}}\bm{V}_{\mathrm{\gamma}}^{-1} \mathrm{vec}\left(\bm{R}_4\right)$

\State $\zeta_{\mathrm{\gamma}} = \frac{2\hat{|\gamma|}}{\bar{v}_{\gamma}}\bar{\gamma}$

\State $\hat\gamma
\gets
{|\check{\gamma}|}\frac{I_1(|\zeta_{\mathrm{\gamma}}|)}{I_0(|\zeta_{\mathrm{\gamma}}|)}
e^{\mathrm{j}\cdot\mathrm{arg}\left(\zeta_{\mathrm{\gamma}}\right)}$
% \vspace{0.5mm}
% \LineComment{// Termination}
% \State $\hat{\gamma} = \dfrac{\bm t^{\mathsf{H}}(\mathrm{Diag}(\bm v_{\mathrm{eff},\gamma}))^{-1}\bm r_4}{1 + \bm t^{\mathsf{H}}(\mathrm{Diag}(\bm v_{\mathrm{eff},\gamma}))^{-1}\bm t}$ 
\end{algorithmic}
\end{algorithm}

With the above derivations, the proposed algorithm is specified and summarized in Algorithm~\ref{alg:mmse}.

\subsection{Complexity Analysis}

Finally, we evaluate the computational complexity of the proposed \ac{MMSE} algorithm (Algorithm \ref{alg:mmse}) in comparison with the basic \ac{NLS} (Algorithm \ref{alg:nls}) and the alternating-optimization \ac{NLS} (Algorithm \ref{alg:aonls}).
For all methods, the estimation of $\bm{Z}$ requires computing the best rank-one approximation of $\bm{R}_2$ via the operator $\mathcal{S}\{\cdot\}$, which, in the worst case, generally involves a full \ac{SVD}.
Hence, the computational complexity is on the order of
\begin{equation}
\label{eq:C_rank-one}
    C_\mathcal{S} = \mathcal{O}\left(M_{\mathrm{A}}M_{\mathrm{B}}\min\{M_{\mathrm{A}},M_{\mathrm{B}}\}\right).
\end{equation}
However, in practice, Krylov subspace methods such as the power method or the Lanczos algorithm can be employed to efficiently compute the largest singular value and its corresponding singular vectors iteratively, which reduces the computational cost to $\mathcal{O}(\tau M_{\mathrm{A}}M_{\mathrm{B}})$, where $\tau$ denotes the number of iterations and is typically a small constant in practice.

%This section evaluates the computational complexity, in terms of Big-O order, of Basic NLS (Algorithm \ref{alg:nls}), Alternating-opt. NLS (Algorithm \ref{alg:aonls}), and the proposed MMSE method described in Algorithm Algorithm \ref{alg:nmmse}.
% 
%Since the rank-one approximation $\mathcal{S}\{\cdot\}$ commonly appearing in all methods corresponds to computing the dominant singular component, we adopt the following upper bound when a full SVD is used:
  
\subsubsection{NLS (Algorithm \ref{alg:nls})}

For the basic \ac{NLS} algorithm, the dominant computational cost arises from the rank-one approximation for estimating $\bm Z$ and the $N_{\mathrm{Iter}}$ alternating updates for estimating $\bm A$ and $\bm B$.
Since $\bm{A}$ and $\bm{B}$ are diagonal matrices, their updates reduce to inner-product computations and element-wise operations; hence, the computational complexity per iteration for updating $\hat{\bm{A}}$ and $\hat{\bm{B}}$ is on the order of $\mathcal{O}\left(M_{\mathrm{A}}M_{\mathrm{B}}\right)$.
Moreover, the final estimation of $\gamma$ is also performed using only element-wise operations and trace computations, which requires $\mathcal{O}\left(M_{\mathrm{A}}M_{\mathrm{B}}\right)$ complexity.
Therefore, the overall computational complexity is given by
\begin{equation}
\label{eq:C_NLS}
    C_{\mathrm{NLS}} = \mathcal{O}\left(C_\mathcal{S} + N_{\mathrm{Iter}}M_{\mathrm{A}}M_{\mathrm{B}}\right).
\end{equation}

\subsubsection{Alternating-optimization. NLS (Algorithm \ref{alg:aonls})}

The alternating-optimization \ac{NLS} first obtains an initial estimate using the basic \ac{NLS}, and then improves the estimation accuracy through alternating optimization with outer iterations.
In each outer iteration, the rank-one approximation appears once for updating $\hat{\bm{Z}}$, and the inner iterations for updating $\hat{\bm{A}}$ and $\hat{\bm{B}}$ are repeated $N_{\mathrm{Iter}}$ times.
Compared with Algorithm \ref{alg:nls}, the newly introduced update of $\hat{\bm{H}}$ is formulated as a matrix inversion based on the least-squares criterion.
However, since $\bm{A}$ and $\bm{B}$ are diagonal matrices, this operation can be implemented element-wise, which reduces the computational complexity to the order of $\mathcal{O}\left(M_{\mathrm{A}}M_{\mathrm{B}}\right)$.
%Furthermore, although the update of $\bm H$ is expressed via matrix inversion based on least squares, exploiting the fact that $\bm A$ and $\bm B$ are diagonal allows it to be implemented with element-wise operations, so that the dominant term can be reduced to $\mathcal{O}\left(M_{\mathrm{A}}M_{\mathrm{B}}\right)$.
% 
Including the cost of the initial estimation by the basic \ac{NLS}, the overall computational complexity is given by
\begin{equation}
\label{eq:C_AO-NLS}
    C_{\mathrm{AO-NLS}} = \mathcal{O}\left((N_{\mathrm{opt}}+1)(C_\mathcal{S} + N_{\mathrm{Iter}}M_{\mathrm{A}}M_{\mathrm{B}})\right),
\end{equation}
where $N_{\mathrm{opt}}$ denotes the number of outer iterations required for convergence, \textit{i.e.}, until the decrease in the cost function becomes sufficiently small.

\subsubsection{MMSE (Algorithm \ref{alg:mmse})}

In Algorithm \ref{alg:mmse}, the dominant factors determining the computational complexity are the estimation of $\bm{A}$ and $\bm{B}$, as well as the estimation of $\gamma$.
Note that the denoising operations based on the von Mises denoiser are all performed in a scalar-wise manner and can be efficiently implemented using table lookups of the modified Bessel functions.
As discussed above, the proposed \ac{MMSE} algorithm is generalized to incorporate knowledge of the covariance matrices of the observation noise, and its computational complexity largely depends on how much of this information is exploited in the estimation process.
In the following, we evaluate two cases: i) using the full covariance matrices in \eqref{eq:covariance}, and ii) using only their diagonal elements (\textit{i.e.}, the variances).

\textbf{i) Full covariance matrices:} When the full covariance matrix is utilized, the estimation of $\bm{A}$ involves matrix inversions of $\bm V_{\mathrm{A},i}$ for all $i\in\mathcal{A}$, which results in a computational complexity on the order of $\mathcal{O}\left(M_{\mathrm{A}}M_{\mathrm{B}}^3\right)$ per update.
Similarly, the update of $\hat{\bm{B}}$ requires $\mathcal{O}\left(M_{\mathrm{A}}^3M_{\mathrm{B}}\right)$ operations.
These alternating updates are repeated $N_\mathrm{Iter}$ times.
Furthermore, the estimation of $\gamma$ involves a matrix inversion of $\bm V_{\gamma}\in\mathbb{C}^{M_{\mathrm{A}}M_{\mathrm{B}}\times M_{\mathrm{A}}M_{\mathrm{B}}}$, whose computational complexity is on the order of $\mathcal{O}\left(M_{\mathrm{A}}^3M_{\mathrm{B}}^3\right)$.
Therefore, the overall computational complexity of the algorithm is given by
\begin{eqnarray}
\!\!\!\!\!\!&\!\!\!\!\!\!\!\!&\!\!\!\!\!\!
C_{\mathrm{MMSE}} \nonumber \\
\!\!\!\!\!\!&\!\!\!\!\!\!\!\!&\!\!\!\!\!\!
= 
\mathcal{O}\Big(
    C_\mathcal{S}
    + N_{\mathrm{Iter}}\left(
        M_{\mathrm{A}}M_{\mathrm{B}}^3
        + M_{\mathrm{A}}^3M_{\mathrm{B}}
    \right) + M_{\mathrm{A}}^3M_{\mathrm{B}}^3
\Big).
\end{eqnarray}
Since the covariance matrices represent long-term statistics, they typically vary slowly over time.
Hence, the matrix inversion associated with $\bm{\varSigma}_\mathrm{A}$ can be regarded as an offline operation and excluded from the online complexity evaluation.
In this case, the estimation of $\gamma$ is dominated by matrix–vector multiplications, and the overall computational complexity can be reduced to
\begin{eqnarray}
\!\!\!\!\!\!&\!\!\!\!\!\!\!\!&\!\!\!\!\!\!
C_{\mathrm{MMSE}} \nonumber \\
\!\!\!\!\!\!&\!\!\!\!\!\!\!\!&\!\!\!\!\!\!
= 
\mathcal{O}\Big(
    C_\mathcal{S}
    + N_{\mathrm{Iter}}\left(
        M_{\mathrm{A}}M_{\mathrm{B}}^3
        + M_{\mathrm{A}}^3M_{\mathrm{B}}
    \right) + M_{\mathrm{A}}^2M_{\mathrm{B}}^2
\Big).
\end{eqnarray}
It is worth noting that achieving improved estimation accuracy by exploiting the full covariance information inevitably comes at the cost of increased computational complexity.

\textbf{ii) Diagonal variance matrices:} When only the diagonal elements of the covariance matrices are used, the matrix inversions of $\bm V_{\mathrm{A},i}$, $\bm V_{\mathrm{B},j}$, and $\bm V_\mathrm{\gamma}$ are simplified to element-wise reciprocal operations.
Consequently, the per-iteration complexity for estimating $\bm A$ and $\bm B$, as well as for estimating $\gamma$, is reduced to $\mathcal{O}\left(M_{\mathrm{A}}M_{\mathrm{B}}\right)$, which is comparable to that of the \ac{NLS}-based algorithms.
Therefore, in this case, the computational complexity is given by
\begin{equation}
\label{eq:C_MMSE-diag}
    C_{\mathrm{MMSE}} = \mathcal{O}\left(C_\mathcal{S} + N_{\mathrm{Iter}}M_{\mathrm{A}}M_{\mathrm{B}}\right).
\end{equation}
In other words, when the prior knowledge of the observation noise is limited to the variance matrix, the proposed \ac{MMSE} algorithm operates with a computational complexity on the same order as that of the basic \ac{NLS} algorithm.

\section{Performance Assessment}
\label{chap:simu}

%\subsection{Comparison with Conventional Methods}

%Throughout this paper, $\bm{G}$ is modeled as a Rayleigh fading channel whose all entries are \ac{i.i.d.} complex Gaussian random variables with zero mean and unit variance, \textit{i.e.}, $\mathcal{CN}(0,\sigma_\mathrm{G}^2)$. \egl{I suggest wait with introducing this assumption until it is really needed.}\tktk{I agree. In fact, this assumption is not strictly necessary for the algorithm. We will use this in the simulation settings.} \egl{sounds good}

% The vectors $\bm{h}$ and $\bm{g}$ are chosen as random columns of the $M_{\mathrm{A}} \times M_{\mathrm{A}}$ and $M_{\mathrm{B}} \times M_{\mathrm{B}}$ \ac{DFT} matrices, respectively, thereby modeling \ac{LoS} propagation between $\mathrm{A}$ and $\mathrm{R}$ and between $\mathrm{R}$ and $\mathrm{B}$. \egl{same here, when is this assumption needed? It is quite specific. I would wait with introducing it until it's needed.}\tktk{As you pointed out, this assumption is not necessary. Thanks to this comment, we realized that the algorithm can be designed even without this assumption.}\egl{OK!}

% The inverse noise variance, $1/\sigma^2$, corresponds to the \ac{SNR} observed at each antenna of $\mathrm{B}$ when the repeater $\mathrm{R}$ is inactive and the single antenna of $\mathrm{A}$ transmits a unit-power pilot signal.

%%%%%%%%%%%%%%%%%%%%%%%%%%%%%%%%%%%%%%
% FIG: RMSE
%%%%%%%%%%%%%%%%%%%%%%%%%%%%%%%%%%%%%%
\begin{figure*}[!t]
\begin{center}
	\subfloat[$(M_{\mathrm{A}},M_{\mathrm{B}}) = (4,3)$.]{
	\includegraphics[width=0.98\columnwidth,keepaspectratio=true]{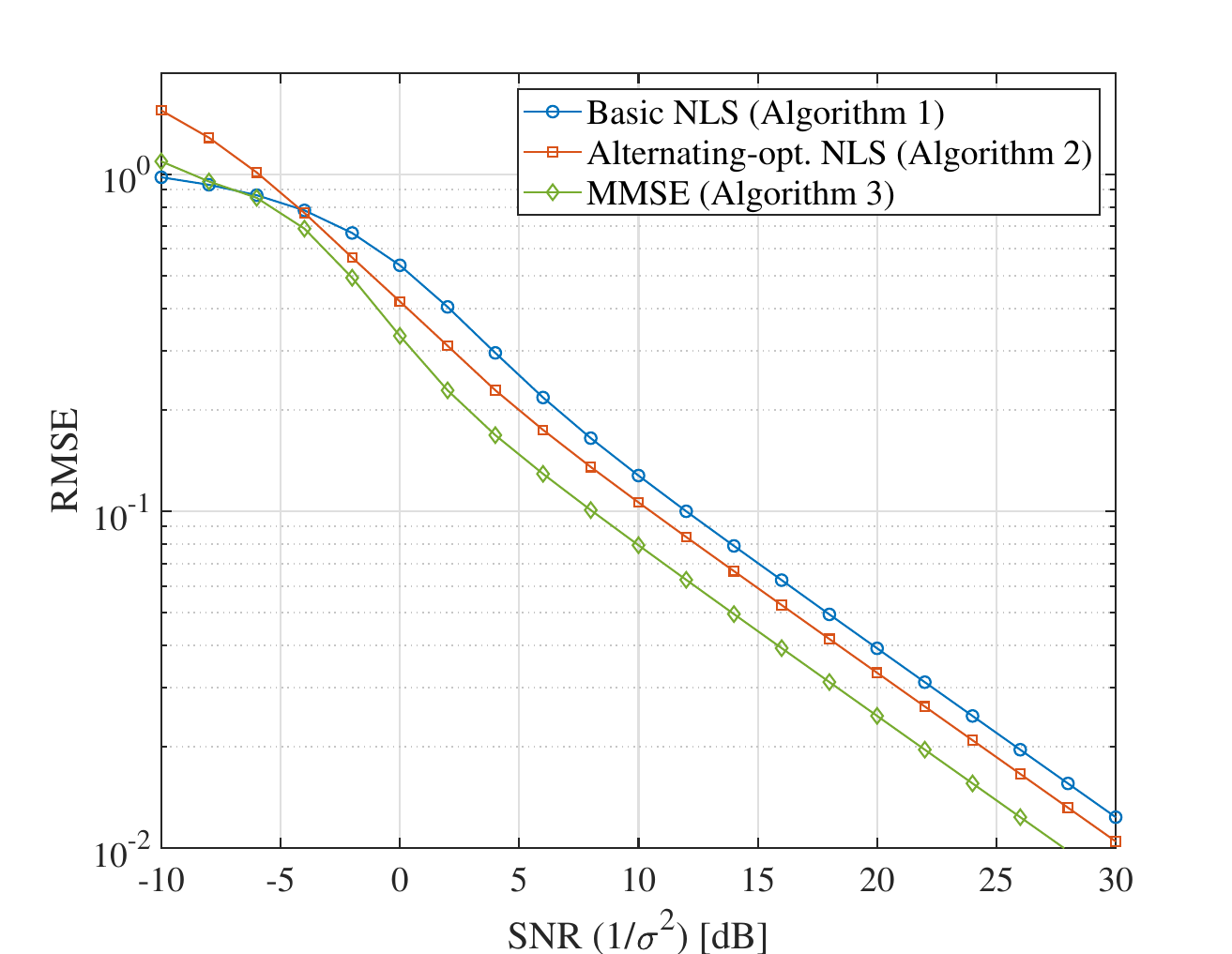}
	\label{fig:(4-3)}
	}
 	\hspace{1mm}
	\subfloat[$(M_{\mathrm{A}},M_{\mathrm{B}}) = (8,8)$.]{
	\includegraphics[width=0.98\columnwidth,keepaspectratio=true]{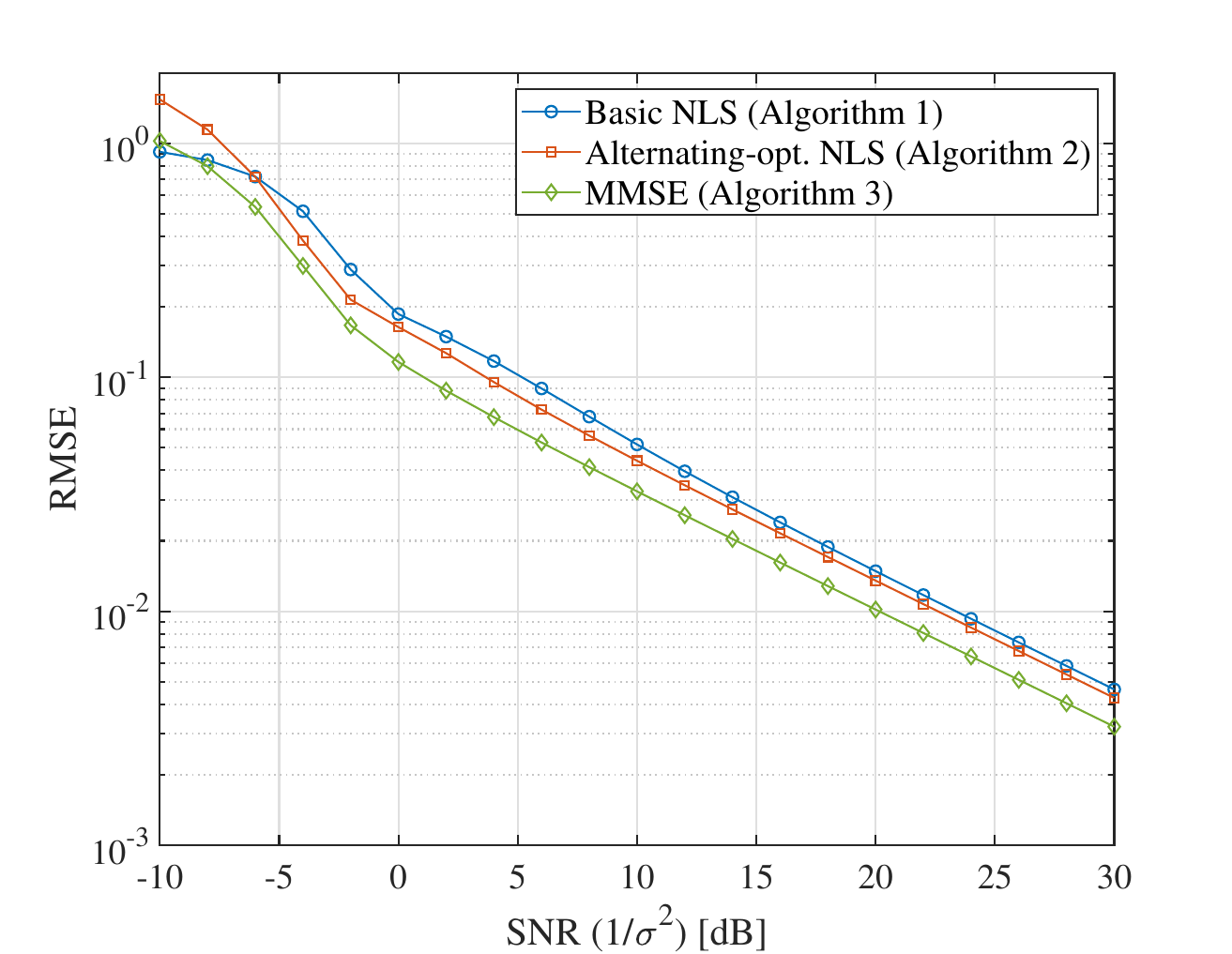}
	\label{fig:(8-8)}
	}
 	\vspace{1mm}
	\subfloat[$(M_{\mathrm{A}},M_{\mathrm{B}}) = (32,16)$.]{
	\includegraphics[width=0.98\columnwidth,keepaspectratio=true]{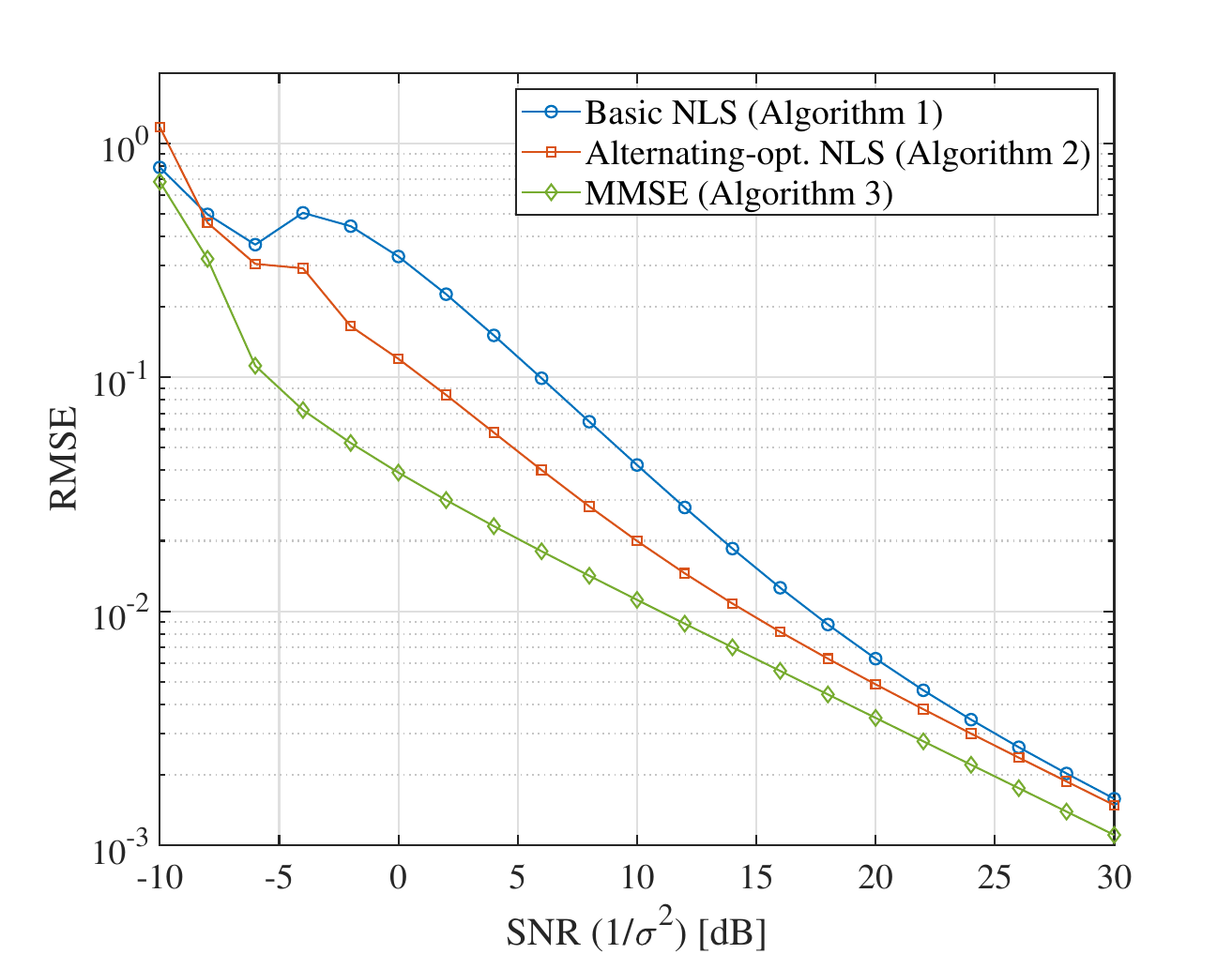}
	\label{fig:(32-16)}
	}
 	\hspace{1mm}
	\subfloat[$(M_{\mathrm{A}},M_{\mathrm{B}}) = (64,32)$.]{
	\includegraphics[width=0.98\columnwidth,keepaspectratio=true]{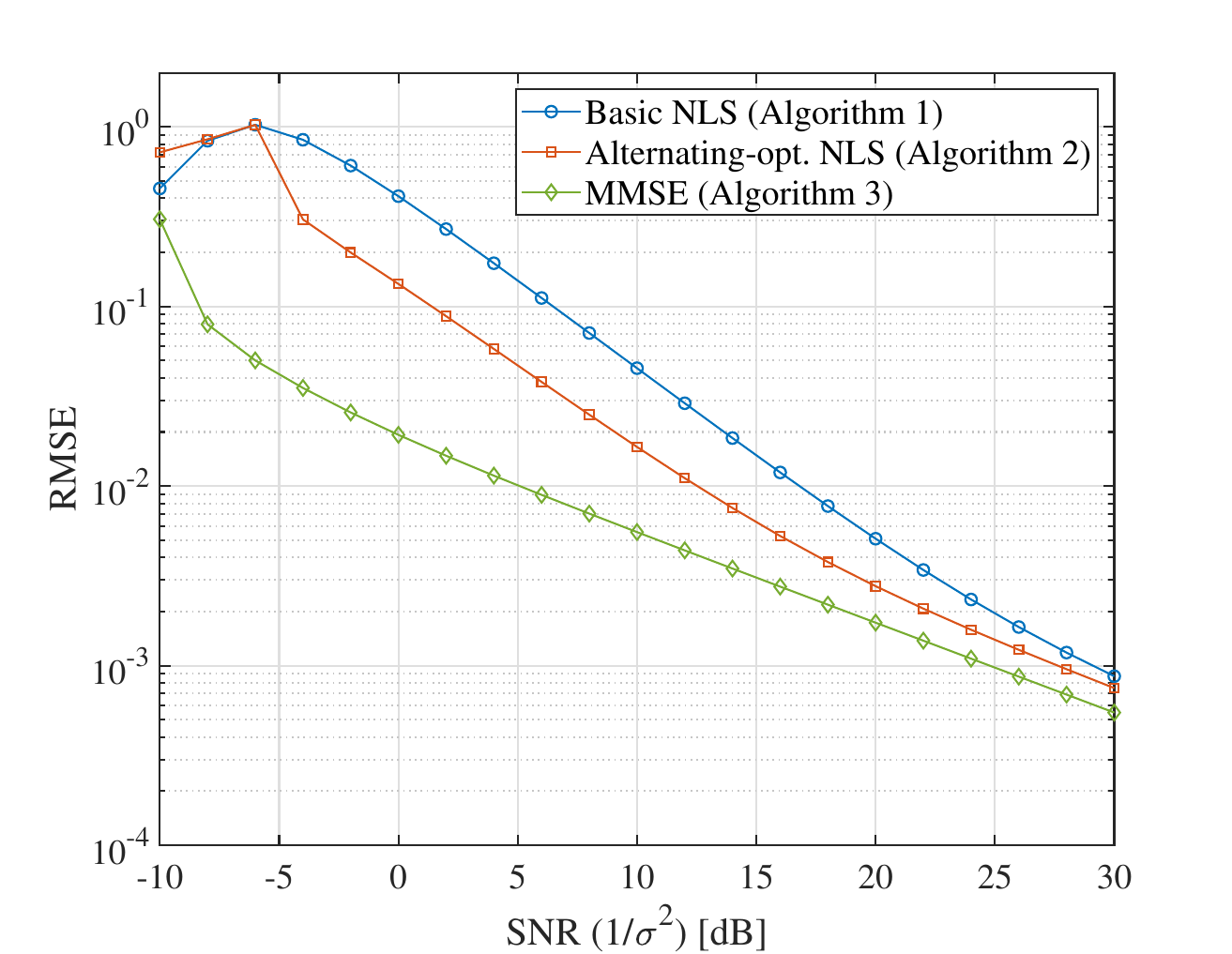}
	\label{fig:(64-32)}
	}
    \vspace{2mm}
	\caption{\Ac{RMSE} of $\hat{\gamma}$ for different antenna configurations.
    % \egl{increase the font size a bit in the figures?}
    % \tktk{Thank you. It has been revised.}
    }
	\label{fig:rmse}
	\vspace{-4mm}
\end{center}
\end{figure*}
%%%%%%%%%%%%%%%%%%%%%%%%%%%%%%%%%%%%%%

To verify the effectiveness of the proposed method, computer simulations are conducted.
As an illustrative example, following~\cite{Larsson2024}, the system parameters listed in Table~\ref{tab:param} are adopted.
The vectors $\bm{h}$ and $\bm{g}$ are chosen as random columns of the $M_{\mathrm{A}} \times M_{\mathrm{A}}$ and $M_{\mathrm{B}} \times M_{\mathrm{B}}$ \ac{DFT} matrices, respectively, which model \ac{LoS} propagation between $\mathrm{A}$ and $\mathrm{R}$, and between $\mathrm{R}$ and $\mathrm{B}$.
The direct channel $\bm{G}$ between $\mathrm{A}$ and $\mathrm{B}$, excluding $\mathrm{R}$, consists of \ac{i.i.d.} $\mathcal{CN}(0,1)$ entries, modeling Rayleigh fading.
The diagonal elements of $\bm{T}_\mathrm{A}$, $\bm{R}_\mathrm{A}$, $\bm{T}_\mathrm{B}$, and $\bm{R}_\mathrm{B}$ in \eqref{eq:rc} have unit magnitude and phases uniformly distributed over $[-\pi, \pi]$~\cite{Vieira2014GCOM}.
%\egl{I think in practice, the phases within each array would be approximately constant but there may be a phase error between the arrays. I.e. $t_{Ai} \approx c$ and $t_{Bi} \approx c e^{j\phi}$ for some $\phi$ uniformly distributed over $[-\pi, \pi]$. It seems a bit of a stretch that the phase between antennas in an array would be drift relative to each other so fast that these phases are i.i.d. Or maybe I misunderstand what you're modeling here}
%\egl{But these elements are not independent, correct? Could we clarify that? What I meant is that perhaps *within an array* we don't have large phase errors (reciprocity calibration within the array can be done infrequently) but *between the two arrays*the phase may drift more quickly because of oscillator drifts. }
%
This setting corresponds to a phase-dominant non-reciprocity scenario, where amplitude mismatches are negligible and non-reciprocity is mainly caused by phase drift/phase errors, as explained in Section II.
%
%Although the phase errors are modeled as Gaussian random variables in \eqref{eq:rc_model}, their variances are highly scenario-dependent in practice; hence, we adopt a worst-case evaluation for the phase uncertainty, which can be interpreted as the limiting case $\sigma_\theta^2\to \infty$.
%
The entries of the noise matrices $\bm W_{\mathrm{B}}^0$, $\bm W_{\mathrm{A}}^0$, $\bm W_{\mathrm{B}}^1$, and $\bm W_{\mathrm{A}}^1$ in \eqref{eq:measurements} are independently drawn from $\mathcal{CN}(0,\sigma^2)$, where the inverse noise variance, $1/\sigma^2$, represents the \ac{SNR} measured at any antenna of $\mathrm{B}$ when a single antenna of $\mathrm{A}$ transmits with unit-power and the repeater $\mathrm{R}$ is inactive.
This corresponds to modeling the observation noise as \ac{AWGN}, which is equivalent to setting $\bar{\bm{\varOmega}}_\mathrm{A}=\bar{\bm{\varPsi}}_\mathrm{B} = \sigma^2\bm{I}_{M_\mathrm{A}}$ and $\bar{\bm{\varOmega}}_\mathrm{B}=\bar{\bm{\varPsi}}_\mathrm{A} = \sigma^2\bm{I}_{M_\mathrm{B}}$ in \eqref{eq:covariance}.

%%%%%%%%%%%%%%%%%%%%%%%%%%%%%%%%%%%%%%
% Table: Pram.
%%%%%%%%%%%%%%%%%%%%%%%%%%%%%%%%%%%%%%
\begin{table}[t]
%\vspace{-5mm}
  \caption{Simulation Parameters}
  \label{tab:param}
  \centering
  \scalebox{0.9}{
  \begin{tabular}{c c c} 
    \hline
    Parameter & Symbol & Value \\
    \hline
    %Fading variance & $\sigma^2_{\mathrm{G}}$ & $1$ \\
    Forward gain of the repeater & $|\alpha|^2$ [dB]& $10$ \\
    Reverse gain of the repeater & $|\beta|^2$ [dB]& $10$ \\
    Num. of iterations for $\bm A, \bm B$ estimation & $N_{\mathrm{Iter}}$ & $100$ \\
    %Num. of outer iterations in Algorithm \ref{alg:aonls}  & $N_{\mathrm{opt}}$ & $25$ \\
    Num. of simulation trials & - & $10^5$ \\
    \hline
  \end{tabular}
  }
\vspace{-5mm}
\end{table}
%%%%%%%%%%%%%%%%%%%%%%%%%%%%%%%%%%%%%%

Under the above settings, it is worth noting that the amount of additional prior information exploited by the proposed \ac{MMSE} algorithm, compared with the \ac{NLS}-based methods, is minimal, and its computational complexity is of the same order, as indicated in \eqref{eq:C_MMSE-diag}.

\subsection{RMSE Performance}

Fig.~\ref{fig:rmse} shows the \ac{RMSE} of $\hat{\gamma}$ as a function of the \ac{SNR} for the \ac{NLS} (Algorithm \ref{alg:nls}), the alternating-optimization \ac{NLS} (Algorithm \ref{alg:aonls}), and the proposed \ac{MMSE} (Algorithm \ref{alg:mmse}).
To evaluate the impact of the transmit and receive array sizes, four antenna configurations are considered, namely, $(M_\mathrm{A},M_\mathrm{B})=(4,3)$, $(8,8)$, $(32,16)$, and $(64,32)$.
These results clearly highlight the fundamental differences between the proposed \ac{MMSE} algorithm and the existing \ac{NLS}-based methods.

We first focus on the results for relatively small array sizes shown in Figs.~\ref{fig:rmse}(a) and \ref{fig:rmse}(b).
In the moderate-to-high \ac{SNR} regime (\ac{SNR} $>5$ dB), all methods are observed to operate stably.
For both the \ac{NLS} and \ac{MMSE}-based methods, the \ac{RMSE} decreases approximately linearly with \ac{SNR} (\textit{e.g.}, a $20$ dB increase in \ac{SNR} leads to roughly a tenfold reduction in \ac{RMSE}).
This behavior can be explained by the fact that the observation noise is Gaussian and the curvature of the log-likelihood function, \textit{i.e.}, the Fisher information, scales proportionally with \ac{SNR}.
In this regime, the local linearization around the true parameter is valid and the Jacobian matrix is full rank, so that the error scaling of the \ac{NLS} and \ac{MMSE} estimators coincides, which is   consistent with estimation theory.
Nevertheless, since the \ac{MMSE} estimator exploits the prior statistical modeling as well as the noise variance, it consistently achieves better performance in terms of the constant factor.
Specifically, the proposed \ac{MMSE} method provides approximately $4$ dB and $2$ dB performance gains over the basic \ac{NLS} and the alternating-optimization \ac{NLS}, respectively.
It is also worth noting that the alternating-optimization \ac{NLS} incurs an additional computational cost proportional to $N_\mathrm{opt}$, whereas the proposed \ac{MMSE} method achieves these gains with a computational complexity comparable to that of the basic \ac{NLS}. %which further highlights its practical advantage.

We next consider the results for relatively large array sizes shown in Figs.~\ref{fig:rmse}(c) and \ref{fig:rmse}(d), where the difference between \ac{NLS} and \ac{MMSE} becomes more pronounced.
As the array dimension increases, the \ac{NLS} estimator, which essentially solves a geometric optimization problem, becomes increasingly sensitive to high-dimensional effects.
In particular, the singular value distribution of the Jacobian matrix tends to deteriorate, and the curvature of the log-likelihood function becomes small along many directions.
As a consequence, the error covariance matrix of the \ac{NLS} estimator becomes ill-conditioned, and the resulting \ac{RMSE} no longer scales proportionally with the \ac{SNR}, not even in the  high-\ac{SNR} regime.
Since calibration tasks are typically performed under low-\ac{SNR} conditions, such behavior of the \ac{NLS} estimator is undesirable in practical systems.

In contrast, for the \ac{MMSE} estimator, the prior distribution introduces a regularization term into the Hessian matrix, which keeps the estimation problem well-conditioned.
As a result, even in high-dimensional and low-\ac{SNR} regimes, the \ac{RMSE} exhibits a stable behavior and remains close to the ideal $1/\mathrm{SNR}$ scaling.
Moreover, as the dimensionality increases, the posterior distribution tends to concentrate due to the concentration-of-measure phenomenon, which leads to a self-averaging effect that mitigates the impact of local nonlinearities and noise fluctuations.
This effect further stabilizes the \ac{MMSE} estimator and improves the accuracy of the Gaussian approximations based on the \ac{CLT}, which are extensively employed in the proposed algorithm.
Consequently, the proposed \ac{MMSE} method follows the theoretical $1/\mathrm{SNR}$ scaling even at low \ac{SNR}\footnote{Although these numerical results are consistent with theoretical scaling laws, deriving a rigorous closed-form \ac{CRLB} for the present bilinear inference problem with latent variables is nontrivial and is left for future work.}.
In particular, for the largest array configuration in Fig.~\ref{fig:rmse}(d), the proposed \ac{MMSE} method achieves performance gains of approximately $14$ dB and $10$ dB over the basic \ac{NLS} and the alternating-optimization \ac{NLS}, respectively, at \ac{RMSE} $=10^{-1}$, indicating that large-scale arrays are especially beneficial for the proposed \ac{MMSE} approach.

\subsection{Convergence Behavior of $\bm A$ and $\bm B$ Estimation}

%%%%%%%%%%%%%%%%%%%%%%%%%%%%%%%%%%%%%%
% FIG: Convergence
%%%%%%%%%%%%%%%%%%%%%%%%%%%%%%%%%%%%%%
\begin{figure}[!t]
\begin{center}
	\subfloat[$(M_{\mathrm{A}},M_{\mathrm{B}}) = (4,3)$.]{
	\includegraphics[width=0.98\columnwidth,keepaspectratio=true]{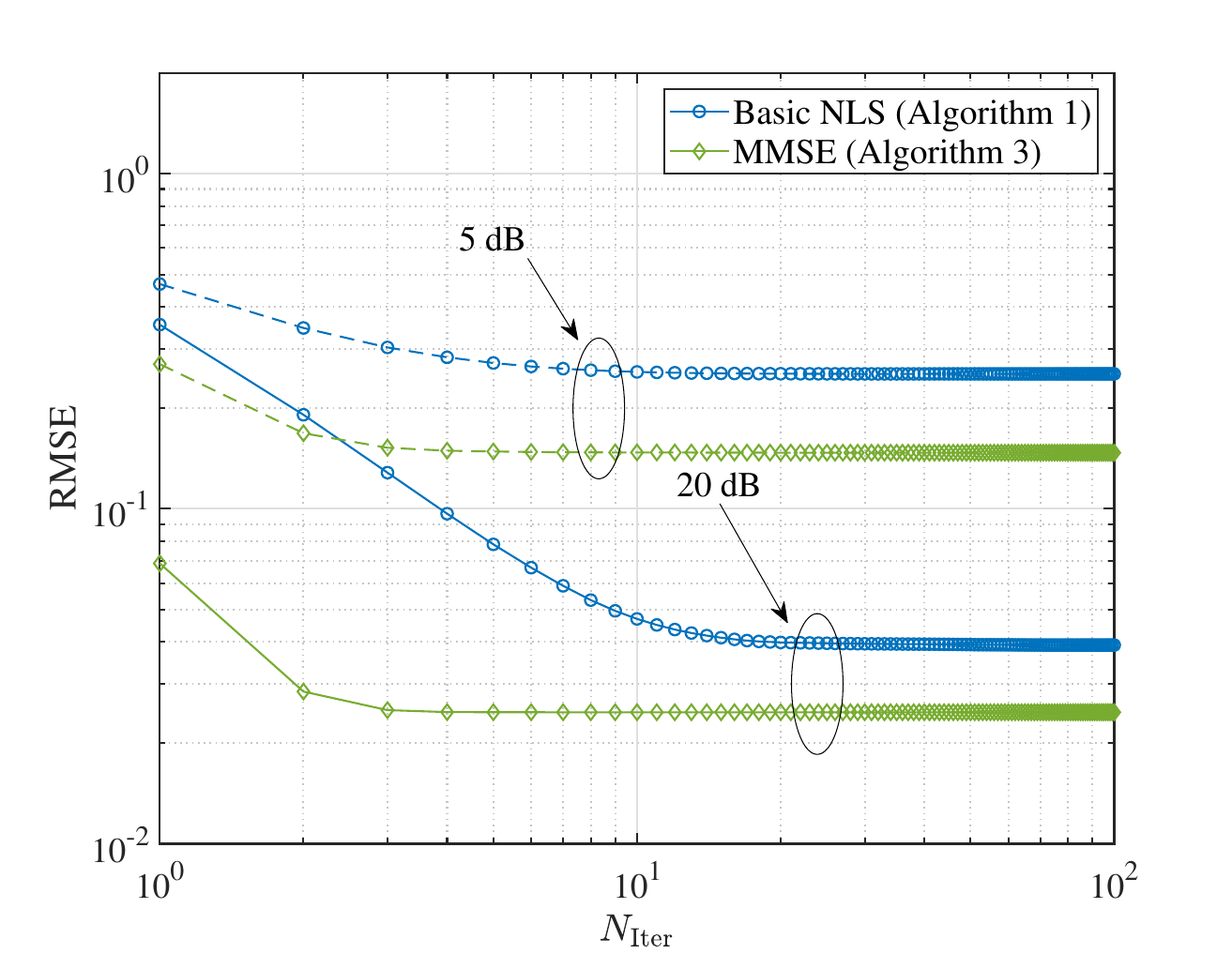}
	\label{fig:20dB_(4-3)}
	}
    \vspace{0mm}
	\subfloat[$(M_{\mathrm{A}},M_{\mathrm{B}}) = (64,32)$.]{
	\includegraphics[width=0.98\columnwidth,keepaspectratio=true]{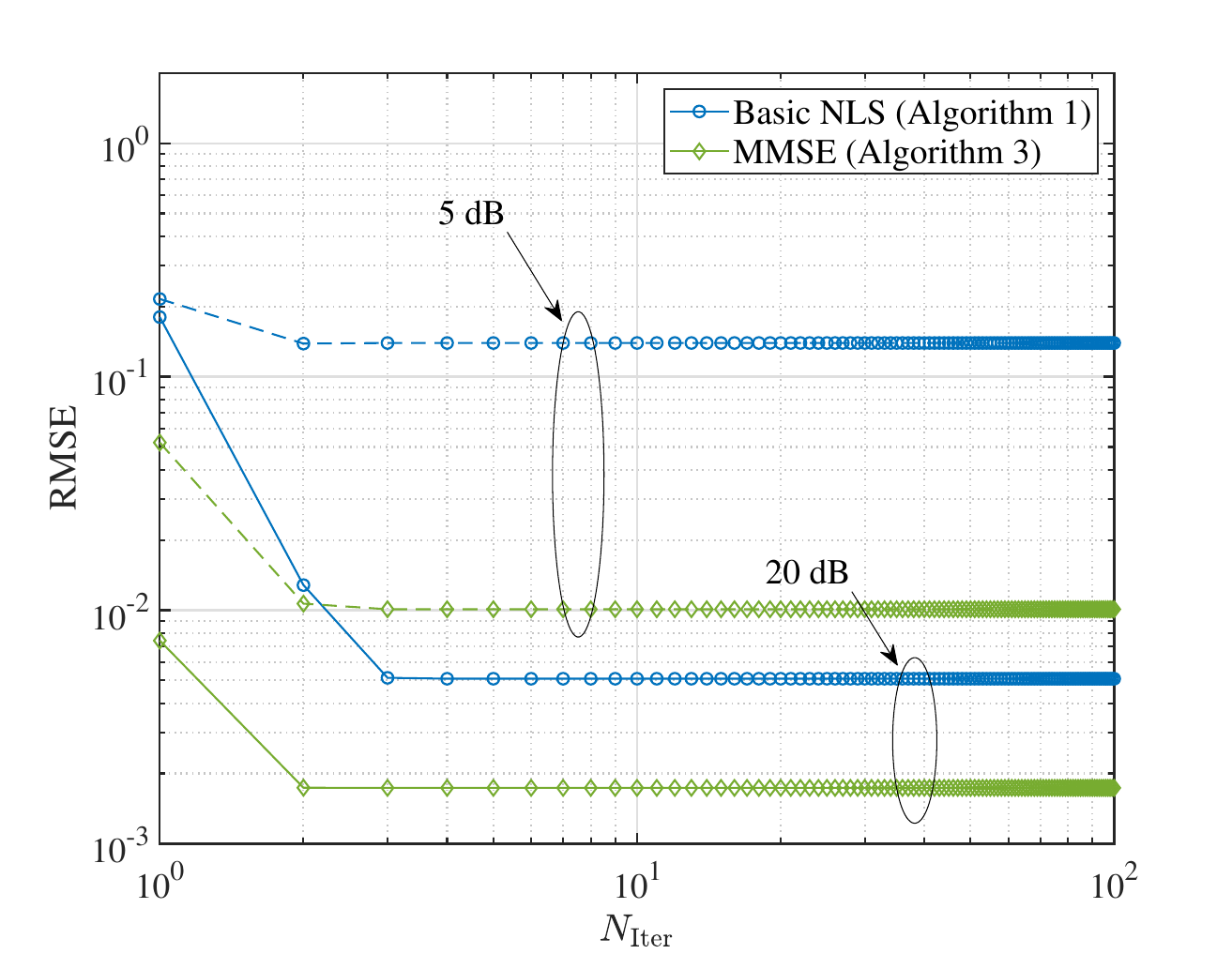}
	\label{fig:20dB_(64-32)}
	}
    \vspace{2mm}
	\caption{Iterative convergence behavior of \ac{NLS}-based and \ac{MMSE}-based reciprocity calibration algorithms.}
	\label{fig:convergence}
	\vspace{-3mm}
\end{center}
\end{figure}
%%%%%%%%%%%%%%%%%%%%%%%%%%%%%%%%%%%%%%

To further clarify the behavior of the proposed algorithms, we next investigate their iterative convergence by focusing on the bilinear inference of $\bm{A}$ and $\bm{B}$.
Fig.~\ref{fig:convergence} shows the \ac{RMSE} of $\hat{\gamma}$ as a function of the number of iterations $N_\mathrm{Iter}$.
For the relatively small array size shown in Fig.~\ref{fig:convergence}(a), the estimation performance of the basic \ac{NLS} improves monotonically with the number of iterations and eventually converges to a stable solution.
However, even in the high-\ac{SNR} regime (\ac{SNR} $=20$ dB), approximately $20$ iterations are required to reach convergence.
For the large array configuration shown in Fig.~\ref{fig:convergence}(b), the convergence behavior of the \ac{NLS} deteriorates significantly in the low-\ac{SNR} regime.
In this case, increasing the number of iterations does not lead to substantial performance improvements, and the estimation error often saturates at a relatively high level.
This phenomenon can be attributed to the intrinsic nonconvexity of the bilinear \acs{LS} problem and the high-dimensional geometric effects discussed previously.
%
%Specifically, the likelihood function becomes nearly flat along many directions, and the optimization landscape contains numerous shallow local minima, so that iterative NLS algorithms fail to extract additional useful information from further iterations.

In contrast, the proposed \ac{MMSE} method exhibits a fundamentally different convergence behavior.
Since the bilinear inference step is formulated as a sequence of Bayesian updates, each iteration effectively refines the posterior distribution by incorporating additional statistical information.
As a result, the estimation performance improves steadily with the number of iterations and converges rapidly, typically within $4$ iterations for all considered scenarios.
In particular, for the large-scale array shown in Fig.~\ref{fig:convergence}(b), an even faster convergence is observed due to the self-averaging effect induced by high dimensionality, which is consistent with the results discussed in the previous subsection.

%These results indicate that the proposed \ac{MMSE} algorithm converges within a small number of iterations and maintains robust \ac{w.r.t.} both the array size and the \ac{SNR}.
%
%Therefore, from both computational and performance perspectives, the proposed approach provides a highly efficient and reliable solution for reciprocity calibration.

\section{Conclusion}
\label{chap:concl}
% 
%Root ../main.tex
% 
%===============
% Section 5.
%===============
%
This paper proposed a novel Bayesian reciprocity calibration framework for repeater-assisted \ac{MIMO}. %  systems under a unified \ac{MMSE} criterion.
By explicitly incorporating a prior statistical model of the reciprocity coefficients and the measurement noise, the proposed method enables a principled Bayesian formulation of the calibration problem, resulting in a robust estimation framework.
To effectively exploit the phase structure of the reciprocity coefficients, a von Mises denoiser was developed.
In addition, a \ac{MoM}-based approach was introduced to adaptively estimate the required second-order statistics, avoiding the need for unavailable long-term prior information.
Simulation results demonstrated that the proposed \ac{MMSE} algorithm can significantly outperform existing deterministic \ac{NLS}-based methods in terms of estimation accuracy, at comparable computational complexity.
Moreover, the proposed method exhibits fast and stable convergence and remains robust \ac{w.r.t.} both the array size and the \ac{SNR}.
Overall, it provides an efficient and practically implementable solution for reciprocity calibration in repeater-assisted \ac{MIMO} systems.

\appendices
{%\color{blue}
\section{Denoiser Design for von Mises Prior Distribution}
\label{chap: deriv_gabp}

%This subsection derives the \ac{MMSE} estimator for complex random variables that are uniformly distributed on the unit circle. 
%
%This estimator plays a central role in the proposed \ac{NMMSE} algorithm, where the diagonal elements of matrices such as $\bm{A}$ and $\bm{B}$ are constrained to have unit magnitude. 

Consider the following \ac{AWGN} observation model:
\begin{equation}
\label{eq:vonmodel}
    y = x + w,\quad w \sim \mathcal{CN}(0,v),
\end{equation}
where $x \triangleq r e^{\mathrm{j}\theta}$ is a point on the circle of radius $r$ in the complex plane, and its phase $\theta$ follows the von Mises distribution
\begin{equation}
\label{eq:von}
p_\mathrm{\theta}(\theta) = \frac{1}{2\pi I_0(\beta)}\exp(\beta\cos(\theta-\mu)),
\end{equation}
with $I_n(\cdot)$ denoting the modified Bessel function of the first kind of order $n$ and the parameters $\mu$ and $\beta$ representing the location and concentration of the distribution, respectively.
These parameters are analogous to the mean and precision of a Gaussian distribution.

Under the assumption that the observation $y$ and the prior distribution in \eqref{eq:von} are known, we derive the \ac{MMSE} estimator of 
$x$ based on the model in \eqref{eq:vonmodel} as
\begin{equation}
\label{eq:mmse}
    \hat{x}_\mathrm{MMSE}
    =
    \mathbb{E}\left[x|y\right]
    =
    r\cdot \mathbb{E}\left[e^{\mathrm{j}\theta}|y\right].
\end{equation}

First, the likelihood function for a given observation $y$ is given by
\begin{eqnarray}
    p_{\mathsf{y|\theta}}(y|\theta)
    \!\!&\!\!=\!\!&\!\!
    \frac{1}{\pi v}
    \exp\left[-\frac{|y - re^{\mathrm{j}\theta}|^{2}}{v}\right] \nonumber \\
    \!\!&\!\!\propto\!\!&\!\!
    \exp\left[
    \frac{2r}{v}\Re\left\{ye^{-\mathrm{j}\theta}\right\}
    \right].
\end{eqnarray}

Next, using Bayes’ rule, the posterior distribution can be expressed as
\begin{eqnarray}
\label{eq:post}
    p_{\mathsf{\theta|y}}(\theta|y)
    \!\!&\!\!\propto\!\!&\!\!
    p_{\mathsf{y|\theta}}(y|\theta)p_\mathsf{\theta}(\theta) \nonumber \\
    \!\!&\!\!\propto\!\!&\!\!
    \exp\left[
    \frac{2r}{v}\Re\left\{ye^{-\mathrm{j}\theta}\right\}
    +
    \beta\cos\left(\theta-\mu\right)
    \right] \nonumber \\
    \!\!&\!\!=\!\!&\!\!
    \exp\left[
    \Re\left\{
    \left(
    \frac{2r}{v}y
    +
    \beta e^{\mathrm{j}\mu} 
    \right)
    e^{-\mathrm{j}\theta}
    \right\}
    \right]
    \nonumber \\
    \!\!&\!\!=\!\!&\!\!
    \exp\left[
    \Re\left\{
    \zeta\left(y,r,v,\mu,\beta\right)
    e^{-\mathrm{j}\theta}
    \right\}
    \right],
\end{eqnarray}
where
\begin{equation}
\label{eq:zeta}
    \zeta\left(y,r,v,\mu,\beta\right)
    \triangleq
    \frac{2r}{v}y
    +
    \beta e^{\mathrm{j}\mu}.
\end{equation}
From \eqref{eq:post} and \eqref{eq:zeta}, the posterior distribution can be again represented by the von Mises distribution as
\begin{equation}
\label{eq:post_von}
    p_{\mathsf{\theta|y}}(\theta|y)
    =
    \frac{1}{2\pi I_0(|\zeta|)}\exp\left[|\zeta|\cos\left(\theta-\mathrm{arg}\left(\zeta\right)\right)\right].
\end{equation}

Finally, substituting \eqref{eq:post_von} into \eqref{eq:mmse} yields
\begin{eqnarray}
\!\!&\!\!\!\!&\!\!
    \hat{x}_\mathrm{MMSE}
    =
    r\cdot \mathbb{E}\left[e^{\mathrm{j}\theta}|y\right] \nonumber \\
\!\!&\!\!\!\!&\!\!
\quad=
    r\cdot
    \int_0^{2\pi}
    e^{\mathrm{j}\theta}\cdot p_{\mathsf{\theta|y}}(\theta|y)d\theta \nonumber \\
\!\!&\!\!\!\!&\!\!
\quad=
    \frac{r e^{\mathrm{j}\cdot\mathrm{arg}\left(\zeta\right)} }{I_0(|\zeta|)}\cdot
    \underbrace{
    \frac{1}{2\pi}\int_{0}^{2\pi}
    e^{\mathrm{j}\theta}
    \exp\left[|\zeta|\cos\left(\theta\right)\right]d\theta
    }_{=I_1(|\zeta|)} \nonumber \\
\!\!&\!\!\!\!&\!\!
\quad=
r\frac{I_1(|\zeta|)}{I_0(|\zeta|)}
e^{\mathrm{j}\cdot\mathrm{arg}\left(\zeta\right)}
\triangleq \eta\left(y;v\right),
\end{eqnarray}
where the function $\eta$ represents the \ac{MMSE} solution in \eqref{eq:vonmodel} expressed as a function of $y$ and $v$, and such a formulation is often referred to as a \textit{\ac{BOD}}.
The mathematical manipulations involving the modified Bessel functions are provided in Appendix B.

For the \ac{BOD}, it is known that the following identity holds between the \ac{MMSE} solution and its \ac{MSE}:
\begin{equation}
        v\cdot\frac{\partial \eta(y;v)}{\partial y}
    = \mathbb{E}\left[|x-\eta(y;v)|^{2}|y\right].
\end{equation}
Thus, the posterior variance is expressed as
\begin{equation}
\mathbb{E}\left[|x-\hat{x}_\mathrm{MMSE}|^{2}|y\right]
=
    r^2\left(1 - \left(\frac{I_1(|\zeta|)}{I_0(|\zeta|)}\right)^{2}\right).
\end{equation}

\section{Derivation of the Modified Bessel Function}

For completeness, we derive the analytical form of the modified Bessel function $I_n(\kappa)$.
Starting from the expansion of the exponential term, we obtain
\begin{equation}
\label{eq:exp}
    \int_0^{2\pi} e^{\mathrm{j}n\psi}e^{\kappa\cos\psi}d\psi = 
\sum_{m=0}^\infty\frac{\kappa^m}{m!}\int_0^{2\pi} e^{\mathrm{j}n\psi}\cos^m\psi\ d\psi.
\end{equation}
Using Euler’s formula, $\cos^m\psi$ can be expressed as a sum of complex exponentials:
\begin{equation}
\cos^m\psi = \frac{1}{2^m}(e^{\mathrm{j}\psi} + e^{-\mathrm{j}\psi})^m =
\frac{1}{2^m}\sum_{k=0}^m
\begin{pmatrix}
m \\
k
\end{pmatrix}
e^{\mathrm{j}(2k-m)\psi}.
\end{equation}
Substituting this into the integral yields
\begin{equation}
    \int_0^{2\pi} e^{\mathrm{j}n\psi}\cos^m\psi\ d\psi
    =
    \frac{1}{2^m}\sum_{k=0}^m
\begin{pmatrix}
m \\
k
\end{pmatrix}
\int_0^{2\pi}
e^{\mathrm{j}(n+2k-m)\psi} d\psi.
\end{equation}
The integral
\begin{equation}
    \int_0^{2\pi}
    e^{\mathrm{j}q\psi} d\psi
    =
    \begin{cases}
        2\pi & q=0,\\
        0 & q\neq 0,
    \end{cases}
\end{equation}
implies that only the terms satisfying $n + 2k - m = 0$ contribute to the integral.
This condition can be parameterized as $m=n+2l$, and the integral is zero for all other values of $m$.
Therefore, for $m=n+2l$ we obtain
\begin{eqnarray}
    \int_0^{2\pi} e^{\mathrm{j}n\psi}\cos^m\psi d\psi = 
2\pi\cdot
\frac{(n+2l)!}{2^{n+2l}\cdot l!\cdot(n+l)!}.
\end{eqnarray}
Substituting this result back into \eqref{eq:exp}, we obtain
\begin{eqnarray}
\int_0^{2\pi}e^{\mathrm{j}n\psi}e^{\kappa\cos\psi}d\psi \!\!\!&=&\!\!\!
2\pi\sum_{l=0}^\infty \frac{1}{l!}\cdot\frac{1}{(n+l)!}\cdot\left(\frac{\kappa}{2}\right)^{n+2l} \nonumber\\
\!\!\!&=&\!\!\!
2\pi\left(\frac{\kappa}{2}\right)^{n}\sum_{l=0}^\infty \frac{1}{l!\cdot (n+l)!}\left(\frac{\kappa}{2}\right)^{2l} \nonumber\\
\!\!\!&=&\!\!\!
2\pi I_n(\kappa).
\end{eqnarray}
Using the standard definition, we thus arrive at the well-known series expansion.

% \section{\Acf{MoM} Based on \eqref{eq:vonmodel}}

% We again consider the model in \eqref{eq:vonmodel},
% \begin{equation}
%     y = r e^{j\theta} + w, \quad w \sim \mathcal{CN}(0,v),
% \end{equation}
% %
% where the phase $\theta$ is unknown and modeled as a random variable.
% %
% In this setting, when only independent samples $\{y_n\}_{n=1}^{N}$ of $y$ are available, and no realization of $\theta$ is observed, the \ac{MoM} can be employed to estimate the most plausible value of $r$ from these observations.

% Even though $\theta$ is unknown, the second moment of $y$ satisfies
% \begin{align}
%     \mathbb{E}\big[|y|^{2}\big]
%     &= \mathbb{E}\big[\,|r e^{j\theta} + w|^{2}\,\big] \\
%     &= r^{2} + v,
% \end{align}
% because $w$ is zero-mean and independent of $\theta$.  
% Hence,
% \begin{equation}
%     r^{2} = \mathbb{E}[|y|^{2}] - v.
% \end{equation}

% Replacing the expectation by the empirical average yields the following method-of-moments estimator:
% \begin{equation}
%     \hat{r}
%     = \sqrt{
%         \frac{1}{N}\sum_{n=1}^{N} |y_{n}|^{2} - v
%       }.
% \end{equation}

% This estimator can also be interpreted as a \ac{MoM}-based estimator for the Rice-distributed amplitude $|y|$.

% \section{Derivation of the CRLB}

% \appendices
% \input{TXT/appdix_calc_H}

\bibliographystyle{REF/IEEEtran}
\bibliography{REF/IEEEabrv,REF/conf_abbrv,REF/ref}

@string{ globecom = {Proc. GLOBECOM}}

@string{ icc = {Proc. ICC}}

@string{ pimrc = {Proc. PIMRC}}

@ARTICLE{Larsson2024,
  author={Larsson, Erik G. and Vieira, Joao and Frenger, Pål},
  journal={IEEE Wireless Commun. Lett.}, 
  title={Reciprocity Calibration of Dual-Antenna Repeaters}, 
  year={2024},
  volume={13},
  number={6},
  pages={1606-1610},
  doi={10.1109/LWC.2024.3382929}}

@ARTICLE{SaraComMag25,
  author  = {Sara Willhammar and Hiroki Iimori and Joao Vieira and
             Lars Sundstr{\"o}m and Fredrik Tufvesson and Erik G. Larsson},
  title   = {Achieving Distributed {MIMO} Performance with Repeater-Assisted Cellular Massive {MIMO}},
  journal = {IEEE Commun. Mag.},
  volume  = {63},
  number  = {3},
  pages   = {114--119},
  year    = {2025},
  month   = mar,
  doi     = {10.1109/MCOM.2025.10908557}
}

@article{Carvalho2025NCRIntro,
  author  = {Fco Italo G. Carvalho and Raul Victor de O. Paiva and
             Tarcisio F. Maciel and Victor F. Monteiro and
             Fco. Rafael M. Lima and Darlan C. Moreira and
             Diego A. Sousa and Behrooz Makki and
             Magnus {\AA}str{\"o}m and Lei Bao},
  title   = {{Network-controlled repeater -- An introduction}},
  journal = {{IEEE Commun. Stand. Mag.}},
  note    = {Early Access},
  year    = {2025},
  month   = aug,
  doi     = {10.1109/MCOMSTD.2025.3595035},
  publisher = {{IEEE}}
}

@INPROCEEDINGS{Gustavsson2014,
  author={Gustavsson, Ulf and Sanchéz-Perez, Cesar and Eriksson, Thomas and Athley, Fredrik and Durisi, Giuseppe and Landin, Per and Hausmair, Katharina and Fager, Christian and Svensson, Lars},
  booktitle={Proc. IEEE Glob. Commun. Conf. Workshops (GLOBECOM WS)}, 
  title={On the impact of hardware impairments on massive {MIMO}}, 
  year={2014},
  doi={10.1109/GLOCOMW.2014.7063447}}

@INPROCEEDINGS{Athley2015,
  author={Athley, Fredrik and Durisi, Giuseppe and Gustavsson, Ulf},
  booktitle={Proc. Eur. Conf. Antennas Propag. (EuCAP)}, 
  title={Analysis of Massive {MIMO} with hardware impairments and different channel models}, 
  year={2015},
  address={Lisbon, Portugal},
  doi={}}

@article{Marzetta2010,
  author  = {T. L. Marzetta},
  title   = {Noncooperative Cellular Wireless with Unlimited Numbers of Base Station Antennas},
  journal = {IEEE Trans. Wireless Commun.},
  volume  = {9},
  number  = {11},
  pages   = {3590--3600},
  year    = {2010}
}

@article{Larsson2014,
  author  = {E. G. Larsson and O. Edfors and F. Tufvesson and T. L. Marzetta},
  title   = {Massive {MIMO} for Next Generation Wireless Systems},
  journal = {IEEE Commun. Mag.},
  volume  = {52},
  number  = {2},
  pages   = {186--195},
  year    = {2014}
}

@inproceedings{Kaltenberger2010,
  author    = {F. Kaltenberger and D. Gesbert and R. Knopp},
  title     = {Relative Channel Reciprocity Calibration in {MIMO/TDD} Systems},
  booktitle = {Proc. Future Netw. Mobile Summit},
  address={Florence, Italy},
  month={Jun.},
  year      = {2010}
}

@article{Vieira2017,
  author  = {J. Vieira and F. Rusek and O. Edfors and S. Malkowsky and L. Liu and F. Tufvesson},
  title   = {Reciprocity Calibration for Massive {MIMO}: Proposal, Modeling, and Validation},
  journal = {IEEE Trans. Wireless Commun.},
  vol.    = {16},
  no.     = {5},
  pp.     = {3042--3056},
  year    = {2017}
}

@article{Bjornson2014_IT,
  author  = {E. Bj{\"o}rnson and J. Hoydis and M. Kountouris and M. Debbah},
  title   = {Massive {MIMO} Systems With Non-Ideal Hardware: Energy Efficiency, Estimation, and Capacity Limits},
  journal = {IEEE Trans. Inf. Theory},
  vol.    = {60},
  no.     = {11},
  pp.     = {7112--7139},
  year    = {2014}
}

@inproceedings{BourdouxRAWCON03,
  author    = {A. Bourdoux and B. Come and N. Khaled},
  title     = {Non-reciprocal transceivers in {OFDM/SDMA} systems: Impact and mitigation},
  booktitle = {Proc. IEEE Radio Wireless Conf. (RAWCON)},
  address   = {Boston, MA, USA},
  year      = {2003},
  doi       = {10.1109/RAWCON.2003.1227923}
}

@article{JiangOTACalibration,
  author  = {X. Jiang and A. Decurninge and K. Gopala and F. Kaltenberger and M. Guillaud and D. Slock},
  title   = {A Framework for Over-the-Air Reciprocity Calibration for {TDD} Massive {MIMO} Systems},
  journal = {IEEE Trans. Wireless Commun.},
  vol.    = {17},
  no.     = {9},
  pages   = {5975--5990},
  month   = sep,
  year    = {2018}
}

@inproceedings{MaPIMRC2015,
  author    = {Y. Ma and D. Zhu and B. Li and P. Liang},
  title     = {Channel Estimation Error and Beamforming Performance in Repeater-Enhanced Massive {MIMO} Systems},
  booktitle = {Proc. IEEE 26th Annu. Int. Symp. Pers., Indoor, Mobile Radio Commun. (PIMRC)},
  address   = {Hong Kong, China},
  month     = Aug,
  year      = {2015}
}

@ARTICLE{Ammar2022,
  author={Ammar, Hussein A. and Adve, Raviraj and Shahbazpanahi, Shahram and Boudreau, Gary and Srinivas, Kothapalli Venkata},
  journal={IEEE Commun. Surv. Tutor.}, 
  title={User-Centric Cell-Free Massive {MIMO} Networks: {A} Survey of Opportunities, Challenges and Solutions}, 
  year={2022},
  volume={24},
  number={1},
  pages={611-652},
  doi={10.1109/COMST.2021.3135119}
}

@ARTICLE{Ngo2017,
  author={Ngo, Hien Quoc and Ashikhmin, Alexei and Yang, Hong and Larsson, Erik G. and Marzetta, Thomas L.},
  journal={IEEE Trans. Wireless Commun.}, 
  title={Cell-Free Massive {MIMO} Versus Small Cells}, 
  year={2017},
  volume={16},
  number={3},
  pages={1834-1850},
  doi={10.1109/TWC.2017.2655515}
}

@ARTICLE{Emil2020cellfree,
  author={Bj\"ornson, Emil and Sanguinetti, Luca},
  journal={IEEE Trans. Wireless Commun.}, 
  title={Making Cell-Free Massive {MIMO} Competitive With {MMSE} Processing and Centralized Implementation}, 
  year={2020},
  volume={19},
  number={1},
  pages={77-90},
  doi={10.1109/TWC.2019.2941478}}

@ARTICLE{Papaza2020cellfree,
  author={Papazafeiropoulos, Anastasios and Kourtessis, Pandelis and Renzo, Marco Di and Chatzinotas, Symeon and Senior, John M.},
  journal={IEEE Trans. Veh. Technology}, 
  title={Performance Analysis of Cell-Free Massive {MIMO} Systems: A Stochastic Geometry Approach}, 
  year={2020},
  volume={69},
  number={4},
  pages={3523-3537},
  doi={10.1109/TVT.2020.2970018}}

@ARTICLE{Nie2020,
  author={Nie, Rongjiang and Chen, Li and Zhao, Nan and Chen, Yunfei and Yu, F. Richard and Wei, Guo},
  journal={IEEE Trans. Commun.}, 
  title={Relaying Systems With Reciprocity Mismatch: Impact Analysis and Calibration}, 
  year={2020},
  volume={68},
  number={7},
  pages={4035-4049},
  doi={10.1109/TCOMM.2020.2982632}
}

@INPROCEEDINGS{Iimori2023GLOBECOM,
  author    = {H. Iimori and E. Kurihara and T. Yoshida and J. Vieira and S. Malomsoky},
  title     = {Amplification Strategy in Repeater-Assisted {MIMO} Systems via Minorization Maximization},
  booktitle = {Proc. IEEE Global Communications Conf. (GLOBECOM)},
  year      = {2023},
  month     = dec,
  address   = {Kuala Lumpur, Malaysia},
  doi       = {10.1109/GLOBECOM54140.2023.10437535}
}

@ARTICLE{Vu2025RSMA,
  author    = {T.-H. Vu and N. H. Tu and V. N. Q. Bao},
  title     = {Study on Reconfigurable Repeater-Based {RSMA} Systems},
  journal   = {IEEE Wireless Commun. Lett.},
  volume    = {14},
  number    = {5},
  pages     = {1371--1375},
  year      = {2025},
  month     = may,
  doi       = {10.1109/LWC.2025.3542613}
}

@ARTICLE{Jopanya2025SWARM,
  author    = {P. Jopanya and D. P. M. Osorio},
  title     = {Enabling Drone Detection with SWARM Repeater-Assisted {MIMO ISAC}},
  journal   = {arXiv preprint arXiv:2509.19119},
  year      = {2025},
  month     = sep,
  url       = {https://arxiv.org/abs/2509.19119}
}

@ARTICLE{Andersson2026DynamicTDD,
  author    = {M. Andersson and A. Chowdhury and E. G. Larsson},
  title     = {Is Repeater-Assisted Massive {MIMO} Compatible with Dynamic {TDD}?},
  journal   = {arXiv preprint arXiv:2510.20998},
  year      = {2026},
  month     = jan,
  url       = {https://arxiv.org/abs/2510.20998}
}

@ARTICLE{Topal2025RA_MIMO,
  author    = {O. A. Topal and \"O. T. Demir and E. Bj\"ornson and C. Cavdar},
  title     = {Fair and Energy-Efficient Activation Control Mechanisms for Repeater-Assisted Massive {MIMO}},
  journal   = {arXiv preprint arXiv:2504.03428},
  year      = {2025},
  month     = apr,
  url       = {https://arxiv.org/abs/2504.03428}
}

@ARTICLE{Chowdhury2026DualAntenna,
  author    = {A. Chowdhury and E. G. Larsson},
  title     = {On the Performance of Dual-Antenna Repeater Assisted Bi-Static {MIMO ISAC}},
  journal   = {arXiv preprint arXiv:2511.17980},
  year      = {2026},
  month     = jan,
  url       = {https://arxiv.org/abs/2511.17980}
}

@ARTICLE{Bai2026RepeaterSwarm,
  author  = {J. Bai and A. Chowdhury and A. Hansson and E. G. Larsson},
  title   = {Repeater Swarm-Assisted Cellular Systems: {Interaction} Stability and Performance Analysis},
  journal = {IEEE Trans. Wireless Commun.},
  volume  = {25},
  pages   = {10018--10034},
  year    = {2026},
  month   = jan,
  doi     = {10.1109/TWC.2025.3647294}
}

@ARTICLE{Le2025ARRNOMA,
  author  = {A.-T. Le and T.-H. Vu and N. H. Tu and T. N. Nguyen and L.-T. Tu and M. Voznak},
  title   = {Active-Reconfigurable-Repeater-Assisted {NOMA} Networks in Internet of Things: Reliability, Security, and Covertness},
  journal = {IEEE Internet Things J.},
  volume  = {12},
  number  = {7},
  pages   = {8759--8772},
  year    = {2025},
  month   = apr,
  doi     = {10.1109/JIOT.2024.3503278}
}

@INPROCEEDINGS{Larsson2024StabilityRepeaters,
  author    = {E. G. Larsson and J. Bai},
  title     = {Stability Analysis of Interacting Wireless Repeaters},
  booktitle = {Proc. IEEE Int. Workshop Signal Process. Adv. Wireless Commun. (SPAWC)},
  year      = {2024},
  month     = sep,
  address   = {Lucca, Italy},
  doi       = {10.1109/SPAWC60668.2024.10694183}
}

@INPROCEEDINGS{Ito2025bipda,
  author={Ito, Kenta and Takahashi, Takumi and Ibi, Shinsuke and Sampei, Seiichi},
  booktitle={Proc. IEEE Int. Conf. Commun. (ICC)}, 
  title={Bayesian Joint Channel and Data Estimation for Correlated Large {MIMO} with Non-orthogonal Pilots}, 
  year={2021},
  address={Montreal, Canada},
  month=Jun,
  doi={10.1109/ICC42927.2021.9500313}}

@ARTICLE{Takahashi2024JCTDD,
  author={Takahashi, Takumi and Iimori, Hiroki and Ishibashi, Koji and Ibi, Shinsuke and de Abreu, Giuseppe Thadeu Freitas},
  journal={IEEE Trans. Wireless Commun.}, 
  title={Bayesian Bilinear Inference for Joint Channel Tracking and Data Detection in Millimeter-Wave {MIMO} Systems}, 
  year={2024},
  volume={23},
  number={9},
  pages={11136-11153},
  doi={10.1109/TWC.2024.3379122}}

@ARTICLE{Rayan2025JCDRE,
  author={Ranasinghe, Kuranage Roche Rayan and Seok Rou, Hyeon and Thadeu Freitas de Abreu, Giuseppe and Takahashi, Takumi and Ito, Kenta},
  journal={IEEE Trans. Wireless Commun.}, 
  title={Joint Channel, Data, and Radar Parameter Estimation for {AFDM} Systems in Doubly-Dispersive Channels}, 
  year={2025},
  volume={24},
  number={2},
  pages={1602-1619},
  doi={10.1109/TWC.2024.3510935}}

@article{Iimori2021GrantFree,
  author    = {Hiroki Iimori and Takumi Takahashi and Koji Ishibashi and
               Giuseppe Thadeu Freitas de Abreu and Wei Yu},
  title     = {Grant-Free Access via Bilinear Inference for Cell-Free {MIMO} With Low-Coherence Pilots},
  journal   = {IEEE Trans. Wireless Commun.},
  volume    = {20},
  number    = {11},
  pages     = {7694--7710},
  year      = {2021},
  month     = nov,
  doi       = {10.1109/TWC.2021.3088125},
  issn      = {1536-1276},
}

@ARTICLE{Takahashi2023ADCcellfree,
  author={Takahashi, Takumi and Iimori, Hiroki and Ando, Kengo and Ishibashi, Koji and Ibi, Shinsuke and de Abreu, Giuseppe Thadeu Freitas},
  journal={IEEE Trans. Wireless Commun.}, 
  title={Bayesian Receiver Design via Bilinear Inference for Cell-Free Massive {MIMO} With Low-Resolution {ADC}s}, 
  year={2023},
  volume={22},
  number={7},
  pages={4756-4772},
  doi={10.1109/TWC.2022.3228326}}

@Book{Chockalingam2014,
Title                    = {Large {MIMO} {S}ystems},
Author                   = {Chockalingam, A. and Rajan, B. Sundar},
Publisher                = {Cambridge University Press},
Year                     = {2014}
}

@ARTICLE{Ito2024,
  author={Ito, Kenta and Takahashi, Takumi and Ibi, Shinsuke and Sampei, Seiichi},
  journal={IEEE Trans. Commun.}, 
  title={Bilinear Gaussian Belief Propagation for Massive {MIMO} Detection With Non-Orthogonal Pilots}, 
  year={2024},
  volume={72},
  number={2},
  pages={1045-1061},
  doi={10.1109/TCOMM.2023.3325479}}

@ARTICLE{Iimori2023JACE,
  author={Iimori, Hiroki and Takahashi, Takumi and Ishibashi, Koji and de Abreu, Giuseppe Thadeu Freitas and González G., David and Gonsa, Osvaldo},
  journal={IEEE Trans. Wireless Commun.}, 
  title={Joint Activity and Channel Estimation for Extra-Large {MIMO} Systems}, 
  year={2022},
  volume={21},
  number={9},
  pages={7253-7270},
  doi={10.1109/TWC.2022.3157271}}

@INPROCEEDINGS{Vieira2014GCOM,
  author={Vieira, Joao and Rusek, Fredrik and Tufvesson, Fredrik},
  booktitle={Proc. IEEE Glob. Commun. Conf. (GLOBECOM)}, 
  title={Reciprocity calibration methods for massive {MIMO} based on antenna coupling}, 
  year={2014},
  month={Dec.},
  address={Austin, USA},
  doi={10.1109/GLOCOM.2014.7037384}}

\end{document}